\documentclass[11pt,letterpaper]{article}
\pdfoutput=1

\usepackage[utf8]{inputenc}
\usepackage[bookmarks]{hyperref}
\hypersetup{colorlinks=true,citecolor=blue,linkcolor=blue,filecolor=blue,urlcolor=blue,pdftitle={Improving quantum communication rates with permutation-invariant codes},pdfauthor={Sujeet Bhalerao and Felix Leditzky}}
\usepackage{amsmath,amssymb,amsthm}
\usepackage{dsfont}
\usepackage{fullpage}
\usepackage{multirow}
\usepackage{makecell}
\usepackage{siunitx}
\usepackage[affil-it]{authblk}
\usepackage{mathdots}
\usepackage{cleveref}
\usepackage{float} 
\crefformat{enumi}{(#2#1#3)}
\crefmultiformat{enumi}{(#2#1#3)}{ and~(#2#1#3)}{, (#2#1#3)}{ and~(#2#1#3)}
\crefrangeformat{enumi}{(#3#1#4)--(#5#2#6)}

\usepackage{enumerate}
\usepackage{mathtools}

\usepackage{booktabs}

\allowdisplaybreaks[4]
\numberwithin{equation}{section}

\newtheorem{theorem}{Theorem}[section]
\newtheorem{lemma}[theorem]{Lemma}
\newtheorem{proposition}[theorem]{Proposition}

\theoremstyle{definition}

\usepackage{tikz}
\usepackage{pgfplots}
 \pgfplotsset{compat=1.18}
\usepackage{xcolor}
\usepackage{comment}

\definecolor{dgreen}{HTML}{006600}
\definecolor{lgreen}{HTML}{B3FFB3}

\usepackage{caption}
\usepackage{subcaption}
\usepackage{setspace}

\newcommand{\cA}{\mathcal{A}}
\newcommand{\cB}{\mathcal{B}}

\newcommand{\cD}{\mathcal{D}}
\newcommand{\cE}{\mathcal{E}}
\newcommand{\cF}{\mathcal{F}}

\newcommand{\cH}{\mathcal{H}}

\newcommand{\cN}{\mathcal{N}}

\newcommand{\cS}{\mathcal{S}}

\newcommand{\cU}{\mathcal{U}}

\newcommand{\kS}{\mathfrak{S}}

\newcommand{\bC}{\mathbb{C}}

\newcommand{\one}{\mathds{1}}
\newcommand{\eps}{\varepsilon}

\DeclareMathOperator{\tr}{tr}

\DeclareMathOperator{\id}{id}

\DeclareMathOperator{\GL}{GL}
\DeclareMathOperator{\SL}{SL}
\DeclareMathOperator{\Sym}{Sym}
\DeclareMathOperator{\End}{End}

\DeclareMathOperator{\gl}{\mathfrak{gl}}
\DeclareMathOperator{\ssyt}{SSYT}

\newcommand{\sumi}{\sum\nolimits}

\newcommand{\ox}{\otimes}

\newcommand{\bb}[1]{\text{BB84},#1}
\newcommand{\twoP}[1]{\text{2P},#1}
\newcommand{\dep}[1]{\text{dep},#1}

\usepackage{ytableau}
\usepackage{physics}
\usepackage{bm}
\newcommand{\be}{\begin{equation}}
\newcommand{\ee}{\end{equation}}

\usepackage[backend=biber,sorting=none,style=numeric-comp,giveninits=true,doi=false,isbn=false,url=false,maxbibnames=20,maxcitenames=2]{biblatex}
\renewbibmacro{in:}{}
\newbibmacro{string+doi}[1]{\iffieldundef{doi}{#1}{\href{https://dx.doi.org/\thefield{doi}}{#1}}}
\DeclareFieldFormat{title}{\usebibmacro{string+doi}{\mkbibemph{#1}}}
\DeclareFieldFormat[article]{title}{\usebibmacro{string+doi}{\mkbibquote{#1}}}
\DeclareFieldFormat[incollection]{title}{\usebibmacro{string+doi}{\mkbibquote{#1}}}                   
\DeclareFieldFormat[inproceedings]{title}{\usebibmacro{string+doi}{\mkbibquote{#1}}}

\title{Improving quantum communication rates\\[.25em] with permutation-invariant codes}
\author[1]{Sujeet Bhalerao\thanks{\href{mailto:sgb4@illinois.edu}{\texttt{sgb4@illinois.edu}}}}
\author[1,2]{Felix Leditzky\thanks{\href{mailto:leditzky@illinois.edu}{\texttt{leditzky@illinois.edu}}}}
\affil[1]{Department of Mathematics, University of Illinois Urbana-Champaign, Urbana, IL 61801, USA}
\affil[2]{Illinois Quantum Information Science and Technology Center, University of Illinois Urbana-Champaign, Urbana, IL 61801, USA}

\begin{document}
\maketitle

\begin{abstract}
In this work we improve the quantum communication rates of various quantum channels of interest using permutation-invariant quantum codes.
We focus in particular on parametrized families of quantum channels and aim to improve bounds on their quantum capacity threshold, defined as the lowest noise level at which the quantum capacity of the channel family vanishes.
These thresholds are important quantities as they mark the noise level up to which faithful quantum communication is theoretically possible.
Our method exploits the fact that independent and identically distributed quantum channels preserve any permutation symmetry present at the input.
The resulting symmetric output states can be described succinctly using the representation theory of the symmetric and general linear groups, which we use to derive an efficient algorithm for computing the channel coherent information of a permutation-invariant code.
Our approach allows us to evaluate coherent information values for a large number of channel copies, e.g., at least 100 channel copies for qubit channels.
We apply this method to various physically relevant channel models, including general Pauli channels, the dephrasure channel, the generalized amplitude damping channel, and the damping-dephasing channel.
For each channel family we obtain improved lower bounds on their quantum capacities.
For example, for the 2-Pauli and BB84 channel families we significantly improve the best known quantum capacity thresholds derived in [Fern, Whaley 2008].
These threshold improvements are achieved using a repetition code-like input state with non-orthogonal code states, which we further analyze in our representation-theoretic framework.
\end{abstract}

\tableofcontents

\section{Introduction}

Quantum information processing offers great promise in efficiently solving hard and practically relevant computational and simulation problems with various applications to real-world problems \cite{aaronson2025future}.
This promise has resulted in substantial attention from academia, industry, and the public sector.
However, quantum systems are notoriously prone to errors induced by noise, posing significant challenges to control quantum systems so that a useful quantum computation can be performed.
Noise in quantum systems is typically modeled as a \emph{quantum channel}.
The \emph{quantum capacity} of a quantum channel characterizes, in a precise information-theoretic way, the fundamental limits of quantum information processing in the presence of noise modeled by the channel.
In this characterization, the task of protecting quantum information from noise is cast as a communication problem in which Alice aims to send quantum information to Bob through a noisy quantum channel connecting them.
This noisy communication link is equivalent to the environmental noise affecting a quantum device, and both are modeled mathematically as a quantum channel.

The quantum capacity of a quantum channel can be expressed in terms of an entropic optimization problem involving the channel coherent information \cite{schumacher1996entanglement,schumacher1996quantum,barnum1998information,barnum2000capacities,lloyd1997capacity,shor2002quantum,devetak2005private} (see \Cref{sec:quantum-capacity}).
However, this formula involves a so-called regularization of the channel coherent information over an unbounded number of channel copies, which turns the quantum capacity formula into a generally intractable optimization problem.
Indeed, regularization is necessary for many quantum channels because of the effect of \emph{superadditivity}, whereby the value of the channel coherent information can strictly increase with the number of channel copies \cite{shor1996syndrome,divincenzo1998capacity,smith2006degenerate,smith2008zero,fern2008lower,smith2011quantum,brandao2012noise,cubitt2015unbounded,jackson2017degenerate,leditzky2018dephrasure,lim2018activation,lim2019activation,noh2020enhanced,bausch2020neural,bausch2021errorthresholds,siddhu2021positivity,siddhu2021entropic,filippov2021capacity,filippov2022multipartite,sidhardh2022exploring,singh2023simultaneous,leditzky2023generic,siddhu2024dampingdephasing,wu2025superadditivity,wu2025small}.
There are a few solvable classes of quantum channels for which the quantum capacity can be calculated exactly (see \Cref{sec:degradable-antidegradable}), but for many channels of interest superadditivity prevents us from determining their quantum capacity.
As a result, for such channels we do not know the exact fundamental limits of faithful quantum information processing, which obstructs the design of efficient error correction and quantum communication protocols achieving the fundamental Shannon-theoretic limit.
In this context, particular interest lies in determining the exact quantum capacity \emph{threshold} of parametrized channel families, marking the noise level up to which faithful quantum information processing is possible.

The intractability of determining these thresholds---and more generally the exact value of the quantum capacity---challenges us to develop techniques for finding good lower bounds on the quantum capacity of a quantum channel.
This typically involves the optimization of the channel coherent information of a given channel over some (possibly restricted) ansatz of multipartite quantum states, together with an efficient algorithm to compute entropies of the resulting output states.
Some techniques in the existing literature make use of quantum error-correcting codes \cite{shor1996syndrome,divincenzo1998capacity,smith2006degenerate,fern2008lower,jackson2017degenerate,leditzky2018dephrasure}, neural network states \cite{bausch2020neural}, graph states \cite{bausch2021errorthresholds}, and methods inspired from condensed matter theory and high energy physics \cite{fan2024overcoming,steinberg2024far,lee2025exact,niwa2025coherent}.

In this work, we consider a restricted optimization of the channel coherent information over multipartite input states, which we call ``codes'' in this work, with suitable symmetries that are preserved by the quantum channel.
These symmetries enable us to use tools from representation theory to develop an efficient algorithm to compute the channel coherent information of such states.
We will show in the present paper that, despite this symmetry restriction in the ansatz, our method yields improved bounds on the quantum capacity and increased quantum capacity thresholds for a number of important quantum channels.

\subsection{Main results} 

We develop a representation-theoretic method to compute the channel coherent information of a family of permutation-invariant codes given by convex mixtures of independent and identically distributed (i.i.d.) quantum states.
Our main technical result is \Cref{thm:coherent-info-iid-mixture}, which presents a formula for the channel coherent information of such states in terms of representation-theoretic data that can be efficiently evaluated.
We also provide two alternative simplifications of this formula for convex mixtures of \emph{pure} i.i.d.~states in \Cref{thm: coherent-info-pure-iid,thm: coherent-info-pure-iid-purified}.

We then apply our optimization method to a variety of quantum channels that are known to exhibit superadditivity of coherent information
in \Cref{sec:results}.
We first focus on Pauli channels, in particular the 2-Pauli channel and the BB84 channel.
For both channels we find new codes that significantly increase the best known quantum capacity thresholds (\Cref{fig: 2paulithreshold,fig: bb84threshold}, respectively).
The optimal states achieving these thresholds resemble repetition codes, except the code states are non-orthogonal.
We carry out a representation-theoretic analysis of how the channel coherent information of these channels behaves as a function of the angle between the non-orthogonal code states, in order to explain what channels benefit from non-orthogonal repetition codes (\Cref{fig:2pauli-CI-by_partition,fig:bb84-CI-by-partition}, respectively).
We also perform a similar analysis for arbitrary Pauli channels (\Cref{fig: pauli-simplex-thresholds}), and investigate why non-orthogonal repetition codes do \emph{not} increase the depolarizing channel threshold over that of usual (orthogonal) repetition codes (\Cref{fig:depol-CI-by-partition}).
We then study the dephrasure channel, the generalized amplitude damping channel, and the damping-dephasing channel.
Our method yields improved achievable rates for the dephrasure and damping-dephasing channels in the mid-noise regime (\Cref{fig:dephrasure,fig: damp-dephasing,fig: damp-dephasing-higher-k}), and improved rates and thresholds for the generalized amplitude damping channel (\Cref{fig: gadc-thresholds}).
MATLAB code used to obtain these numerical results is publicly available at \cite{perm-inv-codes-github}.

\subsection{Structure of the paper}
Our paper is structured as follows.
In \Cref{sec:preliminaries} we fix notation and basic definitions, and give a brief review of the quantum capacity of a quantum channel, solvable channel models, and superadditivity of coherent information.
In \Cref{sec:rep-theory} we review some basics of representation theory, Schur-Weyl duality, and explicit constructions of the irreducible representations of the unitary group.
We use these results in \Cref{sec:coherent-information-symmetries} to derive our main theoretical result, an efficient algorithm to compute the channel coherent information of certain permutation-invariant input codes.
In \Cref{sec:results} we apply this formula to study achievable rates of quantum information transmission and quantum capacity thresholds for a number of channels.
We conclude in \Cref{sec:conclusion} with a discussion of our main results and future directions of research.
In the appendices, we list optimal codes found in this paper (\Cref{app:optimal-perminv-codes}), give an explicit construction of the irreducible representations of $\GL(2)$ on symmetric subspaces (\Cref{app:gl2_irreps}), derive explicit formulas for the coherent information of weighted repetition codes for Pauli channels (\Cref{app:ci-pauli-rep}) and the damping-dephasing channel (\Cref{app:ci-damp-deph-rep}), and give details on our numerical studies (\Cref{app: numerics}).

\section{Preliminaries}
\label{sec:preliminaries}

\subsection{Notation and definitions}

In this paper we consider finite-dimensional Hilbert spaces that we denote by $\cH$, and we use labels $A$, $B$, etc.~to denote subsystems.
The identity operator on a Hilbert space $\cH_A$ is denoted by $\one_A$.
A quantum state or density matrix is a positive semidefinite operator with unit trace.
A quantum state $\psi$ of rank 1 is called pure. 
It may be identified with a normalized vector $|\psi\rangle\in\cH$ satisfying $\psi = \ketbra{\psi}$.
In this situation we use the symbols $\psi$ and $\ketbra{\psi}$ interchangeably.

A quantum channel $\mathcal{N}\colon \mathcal{B}(\mathcal{H}_A) \rightarrow \mathcal{B}(\mathcal{H}_B)$ (or $\cN\colon A\to B$ for short) is a linear completely positive trace-preserving map from the algebra $\mathcal{B}(\mathcal{H}_A)$ of linear operators on $\cH_A$ to those on $\cH_B$.
The identity channel on a system $A$ is denoted by $\id_A$. We use $\log$ to denote the logarithm to the base $2.$
For every quantum channel $\cN\colon A\to B$ there exists an environment space $\cH_E$ and an isometry $V\colon \cH_A\to \cH_B\otimes \cH_E$ such that $\cN(\rho)=\tr_E(V\rho V^\dagger)$ \cite{wilde2013quantum}.
Tracing out the output system $B$ instead of $E$ yields a \emph{complementary channel} $\cN_c$ modeling the loss of information to the environment, defined as $\cN_c(\rho) = \tr_B(V\rho V^\dagger)$.
Note that, given a channel isometry $V$ for a quantum channel $\cN$, every isometry of the form $(\one_B\otimes U_E)V$ for some unitary $U_E$ on $\cH_E$ is also a channel isometry for $\cN$, and hence complementary channels are not unique.
However, all entropic quantities considered in this paper are invariant under this unitary degree of freedom, and hence the particular choice of $\cN_c$ does not matter.

We denote by $\kS_n$ the symmetric group of $\lbrace 1,\dots,n\rbrace$, and by $\GL(d)$ the group of invertible operators acting on $\mathbb{C}^d$.
We write $\Lambda(n,d)$ for the set of partitions of $n$ into at most $d$ parts, or equivalently, the set of Young diagrams with $n$ boxes and at most $d$ rows.

\subsection{Quantum capacity}
\label{sec:quantum-capacity}

\subsubsection{Definition and coding theorem}

The quantum capacity $Q(\mathcal{N})$ of a quantum channel $\mathcal{N}\colon \mathcal{B}(\mathcal{H}_A) \rightarrow \mathcal{B}(\mathcal{H}_B)$ quantifies the maximum rate at which quantum information can be reliably transmitted through the channel in the asymptotic limit of many channel uses. It can be expressed as a regularized entropic quantity \cite{schumacher1996entanglement,schumacher1996quantum,barnum1998information,barnum2000capacities,lloyd1997capacity,shor2002quantum,devetak2005private}:
\begin{align}
Q(\mathcal{N}) = \lim_{n \rightarrow \infty} \frac{1}{n} Q^{(1)}(\mathcal{N}^{\otimes n}),
\label{eq:quantum-capacity}
\end{align}
where $Q^{(1)}(\mathcal{N}) = \max_{\rho} I_c(\mathcal{N}, \rho)$ is the \emph{channel coherent information}, with the coherent information $I_c(\mathcal{N}, \rho)$ defined as
\begin{align}
I_c(\mathcal{N}, \rho) = S(\mathcal{N}(\rho)) - S((\mathcal{N} \otimes \mathrm{id}_R)(\ketbra{\psi})).\label{eq:coherent-info-purification}
\end{align}
Here, $\rho$ is the input to the channel, $\ket{\psi}$ is a purification of $\rho$ to a reference system $R$, and $S(\rho)=-\tr\rho\log\rho$ denotes the von Neumann entropy. The coherent information $I_c(\cN,\rho)$ can be equivalently expressed in terms of a complementary channel $\mathcal{N}_c$ to the environment as
\begin{align}
I_c(\mathcal{N}, \rho) = S(\mathcal{N}(\rho)) - S(\mathcal{N}_c(\rho)).
\label{eq:coherent-info-complementary}
\end{align}
Because of the unitary invariance of von Neumann entropy the coherent information $I_c(\cN,\rho)$ does not depend on the particular choice of complementary channel.

\subsubsection{Degradable and antidegradable channels}
\label{sec:degradable-antidegradable}

The formula \eqref{eq:quantum-capacity} for the quantum capacity of a quantum channel $\cN$ involves a regularization over $n$, the number of channel copies.
This regularization makes the quantum capacity intractable to compute in most cases (see \Cref{sec:superadditivity} below).
However, there are a few classes of channels for which the quantum capacity can be determined exactly.
Important examples of such solvable channel models are degradable and antidegradable channels, which we now define.

Recall that for a given channel $\cN\colon A\to B$ one can choose an isometry $V\colon \cH_A\to\cH_B\otimes \cH_E$ such that $\cN(X) = \tr_E V X V^\dagger$.
A complementary channel $\cN_c\colon A\to E$ is then defined via $\cN_c(X)=\tr_BV X V^\dagger$.
A channel $\cN$ is called \emph{degradable}, if there exists another quantum channel $\cD\colon B\to E$ satisfying $\cD\circ\cN = \cN_c$.
Intuitively, the loss of information to the environment can be simulated locally by the receiver, a concept originally considered in classical network information theory.
The ability to locally degrade the channel output to that of the environment has a profound consequence on a channel's ability to transmit quantum information: The channel coherent information becomes additive, $Q^{(1)}(\cN^{\otimes n}) = n Q^{(1)}(\cN)$ for all $n\in\mathbb{N}$, and hence the regularization in \eqref{eq:quantum-capacity} vanishes and we obtain a \emph{single-letter} formula for the quantum capacity, $Q(\cN) = Q^{(1)}(\cN)$ \cite{devetak2005simultaneous}.

The dual concept is that of an \emph{antidegradable} channel $\cN\colon A\to B$ (with complementary channel $\cN_c\colon A\to E$), for which there exists a quantum channel $\cA\colon E\to B$ satisfying $\cN = \cA\circ\cN_c$.
A data-processing argument shows that the channel coherent information of an antidegradable channel vanishes, $Q^{(1)}(\cN^{\otimes n})=0$ for all $n\in\mathbb{N}$, and hence we also have $Q(\cN)=0$ for such channels.
The latter identity can also be understood as an instance of the no-cloning theorem \cite{park1970transition,wootters1982nocloning,dieks1982epr}.

There are many examples of degradable and antidegradable channels.
For example, dephasing channels $\rho\mapsto (1-p)\rho + p Z\rho Z$ are degradable for all $p\in[0,1]$, and amplitude damping channels $\cA_{\gamma}$ (equal to $\cA_{\gamma,0}$ as defined in \eqref{eq:gadc}) are degradable (resp.~antidegradable) for $\gamma\in[0,1/2]$ (resp.~$\gamma\in[1/2,1]$).
From an information-theoretic point of view, the capabilities of (anti-)degradable channels to transmit quantum information are completely understood due to their additivity properties.\footnote{We note here that there are channel classes that are neither degradable nor antidegradable, yet their quantum capacity can be determined exactly. 
Examples include `less noisy' channels \cite{watanabe2012capable}, entanglement-binding channels \cite{horodecki2000binding}, direct sums of partial traces \cite{gao2018tro}, the `platypus' channels \cite{leditzky2023platypus} (up to the validity of the `spin alignment conjecture' \cite{leditzky2023platypus,alhejji2024spin}), multi-level generalizations of amplitude damping channels \cite{chessa2021mad,chessa2023resonant}, partially coherent direct sum channels \cite{chessa2021partially}, and the channel considered in \cite{smith2025additivity}.
}
The situation changes dramatically when considering compositions of such channels, with recent examples including the dephrasure channel (see \Cref{sec:dephrasure}) \cite{leditzky2018dephrasure} or the damping-dephasing channel (see \Cref{sec:damping-dephasing}) \cite{siddhu2024dampingdephasing}.
These channels exhibit superadditivity of coherent information, which we discuss in the next section.

\subsubsection{Superadditivity}
\label{sec:superadditivity}

The channel coherent information $Q^{(1)}$ exhibits the phenomenon of superadditivity, whereby entangled inputs across multiple channel uses can achieve higher communication rates than product inputs.
There are at least two types of superadditivity of coherent information:
In the first, two \emph{different} channels jointly used together yield a larger coherent information than the sum of the individual quantities. That is, there are channels $\cN_1$ and $\cN_2$ such that $Q^{(1)}(\mathcal{N}_1 \otimes \mathcal{N}_2) > Q^{(1)}(\mathcal{N}_1) + Q^{(1)}(\mathcal{N}_2)$, and this superadditivity may even hold for the quantum capacity $Q$ itself \cite{smith2008zero,smith2011quantum,brandao2012noise,lim2018activation,lim2019activation,leditzky2023generic,wu2025superadditivity,wu2025small}.
The first example of this type of superadditivity is known as \emph{superactivation} of quantum capacity, whereby two channels that are useless on their own (that is, each having zero quantum capacity) can be used together to faithfully transmit quantum information \cite{smith2008zero}.

The other type of superadditivity, and the one relevant for our work, involves multiple tensor copies of the \emph{same} channel: 
there are channels $\cN$ and natural numbers $n$ such that \cite{shor1996syndrome,divincenzo1998capacity,smith2006degenerate,fern2008lower,cubitt2015unbounded,jackson2017degenerate,leditzky2018dephrasure,noh2020enhanced,bausch2020neural,bausch2021errorthresholds,siddhu2021positivity,siddhu2021entropic,filippov2021capacity,filippov2022multipartite,sidhardh2022exploring,singh2023simultaneous,siddhu2024dampingdephasing,wu2025small}
\begin{align}
    Q^{(1)}(\cN^{\otimes n}) > n Q^{(1)}(\cN).
\end{align}
Early examples of this type of superadditivity were observed in \cite{shor1996syndrome,divincenzo1998capacity}, showing that for the qubit depolarizing channel $\cD_q(X) = (1-q)X+q \tr(X) \one/2$ there exist $n$ and values of $q$ such that
\begin{align} 
    Q(\cD_q)\geq \frac{1}{n}Q^{(1)}(\cD_q^{\otimes n})> Q^{(1)}(\cD_q),
    \label{eq:depolarizing-superadd}
\end{align}
and this may occur even when $Q^{(1)}(\cD_q)=0$.
Hence, superadditivity can significantly increase achievable rates for faithful quantum information transmission.
However, it also turns the quantum capacity formula \eqref{eq:quantum-capacity} into an unbounded optimization problem that is generally intractable to solve except for solvable channel models like those mentioned in \Cref{sec:degradable-antidegradable}.
For example, determining the exact quantum capacity of the qubit depolarizing channel for $q\in(0,1/3)$ remains a major open problem in quantum information theory.
Remarkably, superadditivity of coherent information has recently been observed experimentally for the dephrasure channel \cite{leditzky2018dephrasure,yu2020experimental}.

The superadditivity in \eqref{eq:depolarizing-superadd} for the coherent information of the qubit depolarizing channel (as well as many other channels) is achieved by a simple repetition code $|\phi_n\rangle_{RA^n} = (|0\rangle_R\otimes |0\rangle_A^{\otimes n}+|1\rangle_R\otimes |1\rangle_A^{\otimes n})/\sqrt{2}$, where $R$ denotes a reference system and $A^n$ are the channel input qubits.
This code corresponds to the mixed input state 
\begin{align}
    \phi_n = \frac{1}{2}\left( \ketbra{0}^{\otimes n} + \ketbra{1}^{\otimes n}\right).
\end{align}
This quantum state is \emph{permutation-invariant}, and an i.i.d.~quantum channel $\cN^{\otimes n}$ preserves this symmetry.
This observation is the starting point of our analysis of coherent information using representation-theoretic methods, which we develop in \Cref{sec:coherent-information-symmetries}.
To prepare this discussion, we first review the necessary background in representation theory  in the following section.

\section{Representation theory of the general linear group}
\label{sec:rep-theory}

\subsection{Basics of representation theory}
\label{sec:rep-theory-basics}

A representation $(\varphi,V)$ of a group $G$ consists of a vector space $V$ (also called representation space) and a group homomorphism $\varphi\colon G \to \GL(V)$, where $\GL(V)$ denotes the group of invertible operators acting on $V$.
A subspace $W\leq V$ is called \emph{invariant} if $\varphi(g)w\in W$ for all $g\in G$ and $w\in W$.
A representation $(\varphi,V)$ is called \emph{irreducible} if $\lbrace 0_V\rbrace$ (with $0_V$ denoting the zero vector in $V$) and $V$ are the only invariant subspaces.
We often abbreviate ``irreducible representation'' as ``irrep''.
Two representations $(\varphi,V)$ and $(\psi,W)$ of a group $G$ are called \emph{equivalent} if there exists an isomorphism $\alpha\colon V\to W$ such that $\psi(g)\circ \alpha = \alpha\circ\varphi(g)$ for all $g\in G$.

In the following we restrict our discussion to representations over $\mathbb{C}$.
A representation $(\varphi,V)$ is called completely reducible if it can be written as a direct sum of irreducible representations.
Every finite-dimensional representation of a finite or compact group (over $\mathbb{C}$) can be chosen unitary (that is, $\varphi(g)$ is unitary for all $g\in G$), and is then also completely reducible.
In particular, any finite-dimensional representation $(\varphi,V)$ of a finite or compact group $G$ can be decomposed as
\begin{align}
    V \cong \bigoplus\nolimits_{\lambda} V_\lambda \otimes M_\lambda, \label{eq:representation-decomposition}
\end{align}
where the direct sum runs over finitely many pairwise inequivalent irreps $(\varphi_\lambda,V_\lambda)$ of $G$ indexed by $\lambda$, and $M_\lambda$ is the \emph{multiplicity space} of the irrep $(\varphi_\lambda,V_\lambda)$.
That is, the dimension of $M_\lambda$ is the number of times the irrep $V_\lambda$ appears in the decomposition \eqref{eq:representation-decomposition}.
This decomposition is known as \emph{isotypical decomposition} of a representation.
The summands $W_\lambda = V_\lambda\otimes M_\lambda\cong V_\lambda^{\oplus \dim M_\lambda}$ collecting all irreps $V_\lambda$ of the same type are unique, whereas decomposing $W_\lambda$ into $\dim M_\lambda$ copies of $V_\lambda$ essentially corresponds to a suitable basis choice \cite[Chs.~2.6, 2.7]{serre1977linear}.

The following standard result of representation theory is well-known and easy to prove \cite{fulton2013representation,etingof2011introduction}, but immensely useful in applications:
\begin{lemma}[Schur's Lemma]
\label{lem:schur}
    Let $(\varphi,V)$ and $(\psi,W)$ be irreducible representations of a group $G$ over $\mathbb{C}$, and let $f\colon V\to W$ be a $G$-linear map, that is, a linear map satisfying $f\circ \varphi(g) = \psi(g)\circ f$ for all $g\in G$.
    Then the following holds:
    \begin{enumerate}[{\normalfont (i)}]
        \item Either $f$ is an isomorphism (and $(\varphi,V)$ and $(\psi,W)$ are equivalent), or $f=0$.
        \item If $V=W$, then $f = \lambda \one$ for some $\lambda\in\mathbb{C}$.
    \end{enumerate}
\end{lemma}
Schur's Lemma has an important consequence for $G$-invariant operators that is crucial for our results: Let $(\varphi,V)$ be a representation of a finite or compact group $G$, and let $X$ be an operator on $V$ satisfying $\varphi(g)X=X\varphi(g)$ for all $g\in G$.
Then, with respect to the decomposition \eqref{eq:representation-decomposition}, the operator $X$ has the form 
\begin{align}
    X \cong \bigoplus\nolimits_\lambda \one_{V_\lambda}\otimes X_\lambda,
    \label{eq:G-invariant-operators}
\end{align}
where each $X_\lambda$ acts on the multiplicity space $M_\lambda$.

\subsection{Schur-Weyl duality}
\label{sec:schur-weyl}

We consider the following representations of the symmetric group $\kS_n$ and the general linear group $\GL(d)$ on the tensor space $(\mathbb{C}^d)^{\otimes n}$. 
The group $\GL(d)$ acts diagonally by 
\begin{align}
	g \cdot (|v_1\rangle \otimes \cdots \otimes |v_n\rangle) &= g|v_1\rangle \otimes \cdots \otimes g |v_n\rangle \quad\text{for $g \in \GL(d)$,}
	\intertext{and $\kS_n$ acts by permuting tensor factors,}
	\sigma\cdot (|v_1\rangle \otimes \cdots \otimes |v_n\rangle) &= |v_{\sigma^{-1}(1)}\rangle \otimes \cdots \otimes |v_{\sigma^{-1}(n)}\rangle  \quad\text{for $\sigma \in \kS_n.$}
\end{align}  

These two actions commute, and furthermore span each other's commutant in $\cB((\mathbb{C}^d)^{\otimes n})$, which is known as Schur-Weyl duality \cite{fulton2013representation,etingof2011introduction}.
This duality provides a useful decomposition of the representation space $(\mathbb{C}^d)^{\otimes n}$ under the action of $\GL(d)\times \kS_n$:
\begin{align}
(\mathbb{C}^d)^{\otimes n} \cong \bigoplus_{\lambda \vdash_d n} V_\lambda^d \otimes S_\lambda,
\label{eq:schur-weyl-duality}
\end{align}
where $\lambda \vdash_d n$ denotes a partition of $n$ with at most $d$ parts, $V_\lambda^d$ is the irreducible $\GL(d)$-representation indexed by $\lambda$, and $S_\lambda$ is the irreducible $\kS_n$-representation indexed by $\lambda$.
Note that \eqref{eq:schur-weyl-duality} is an instance of the decomposition of a representation into irreps discussed in \Cref{sec:rep-theory-basics}, applied to either the $\GL(d)$- or the $\kS_n$-representation.
That is, the $\kS_n$-irreps $S_\lambda$ are the multiplicity spaces for the $\GL(d)$-irreps $V_\lambda^d$, and vice versa.

The dimension of $S_\lambda$ is given by the hook length formula
\begin{align}
\dim S_\lambda = \frac{n!}{\prod_{\square \in \lambda} h_\square},
\end{align}
where $h_\square$ is the hook length of the box $\square$ in the Young diagram $\lambda$, defined as the total number of boxes in the diagram below and to the right of that box, including the box itself.
The dimension of $V_\lambda^d$ is given by Weyl's dimension formula
\begin{align}
    \dim V_\lambda^d = \prod_{1\leq i<j\leq d} \frac{\lambda_i-\lambda_j+j-i}{j-i} \leq (n+1)^{d(d-1)/2},
    \label{eq:weyl-dimension-formula}
\end{align}
where the polynomial upper bound follows from a simple counting argument \cite{christandl2006phd}.

We refer to \cite{fulton2013representation,etingof2011introduction} for further details of this decomposition, and to \cite{harrow2005phd,christandl2006phd} for a treatment in the context of quantum information theory.
Here, we focus on explicit constructions of the irreducible representations $V_\lambda^d$ appearing in \eqref{eq:schur-weyl-duality}, for which we first introduce the concepts of semistandard Young tableaux and Gelfand-Tsetlin patterns.

\subsection{Semistandard Young tableaux and Gelfand-Tsetlin patterns}
\label{sec:ssyt-gt}

A \textit{semistandard Young tableau} (SSYT) of shape $\lambda$ with entries in $\{1, 2, \ldots, d\}$ is a filling of the Young diagram $\lambda$ such that entries are weakly increasing along rows from left to right, and strictly increasing along columns from top to bottom. 
We will sometimes write $\ssyt_d(\lambda)$ for the set of all SSYT with entries in $\{1, 2, \ldots, d\}$ of shape $\lambda$.
The number of SSYT of shape $\lambda$ with entries in $\{1, \ldots, d\}$ equals $\dim V_\lambda^d$.

A \textit{Gelfand-Tsetlin pattern} (GT pattern) for $\mathrm{GL}(d)$ is a triangular array of integers
\begin{align}
	\begin{matrix}
		\lambda_{d,1} & & \lambda_{d,2} & \cdots & \lambda_{d,d-1} && \lambda_{d,d} \\
		&  \lambda_{d-1,1} & & \cdots &&\lambda_{d-1,d-1} & \\
		&&  \ddots & & \iddots & \\
		& &  &\lambda_{1,1} &&&
	\end{matrix}
	\label{eq:gelfand-tsetlin}
\end{align}
satisfying the interlacing conditions
\begin{align}
\lambda_{k,i} \geq \lambda_{k-1,i} \geq \lambda_{k,i+1}
\end{align}
for all valid indices. 
The top row $(\lambda_{d,1}, \ldots, \lambda_{d,d})$ with $\lambda_{d,1} \geq \cdots \geq \lambda_{d,d} \geq 0$ determines a partition $\lambda=(\lambda_{d,1},\dots,\lambda_{d,d})$ with $d$ parts. 
There is a bijection between SSYTs of shape $\lambda$ with entries in $\{1, \ldots, d\}$ and GT patterns with top row $\lambda$ (see, e.g., \cite[Ch.~7]{stanley2023enumerative}).

\subsection{Constructing \texorpdfstring{$\GL(d)$}{GL(d)}-irreps via Lie algebra representations}
\label{sec:molev}

Let $\lambda$ be a Young diagram of height less than or equal to $d$. 
The irreducible representation $(q_\lambda,V_\lambda^d)$ of $\GL(d)$ associated with $\lambda$ is a polynomial representation, i.e., the matrix elements of the representation matrix $q_\lambda(g)$ are polynomial functions of the matrix elements of $g\in\GL(d)$.
It is conveniently defined in terms of irreps of the Lie algebra $\gl(d)$ consisting of all complex $(d\times d)$-matrices with commutator $[X,Y]=XY-YX$ as the Lie bracket.
Here, a representation of a Lie algebra $\mathfrak{g}$ on a vector space $V$ is a map $\varphi\colon \mathfrak{g}\to \End(V)$ such that $\varphi([g,h]_\mathfrak{g})=[\varphi(g),\varphi(h)]_V$, where $[\cdot,\cdot]_\mathfrak{g}$ is the Lie bracket on $\mathfrak{g}$, and $[\cdot,\cdot]_V$ is the usual matrix commutator on $V$.
A $\gl(d)$-irrep $(\phi_\lambda,L_\lambda^d)$, whose construction we describe below, defines the $\GL(d)$-irrep $V_\lambda^d$ via the exponential map:
If $A=\exp(X)\in\GL(d)$ for $X\in\gl(d)$, then $A\mapsto \exp(\phi_\lambda(X))$ defines an irreducible representation of $\GL(d)$ on $V_\lambda^d$.

We now give a brief overview of how to construct the irreps $(\phi_\lambda,L_\lambda^d)$ of $\gl(d)$, referring to \cite{humphreys2012introduction,molev2006gelfand} for details.
Let $E_{ij}$ for $i,j = 1, \ldots, d$ be the matrix with a $1$ in the $(i,j)$-th position and zeros elsewhere. 
The $E_{ij}$ form the standard basis of $\gl(d)$, and the commutation relations are given by $[E_{ij}, E_{kl}] = \delta_{jk}E_{il} - \delta_{il}E_{kj}$.
Because of these relations it suffices to restrict our attention to the subset $\lbrace E_{k,k}\rbrace\cup \lbrace E_{k,k+1}\rbrace \cup \lbrace E_{k+1,k}\rbrace$.

Finite-dimensional irreducible representations of $\gl(d)$ are highest weight representations. 
A highest weight $\lambda$ is a $d$-tuple of complex numbers $(\lambda_1, \ldots, \lambda_d)$ such that the differences $\lambda_i - \lambda_{i+1}$ are non-negative integers for $i=1, \ldots, d-1$. 
For polynomial representations of $\GL(d)$, these $\lambda_i$ are integers forming a partition $\lambda_1 \ge \lambda_2 \ge \cdots \ge \lambda_d \ge 0$.
Since all $\GL(d)$-irreps appearing in the Schur-Weyl decomposition \eqref{eq:schur-weyl-duality} of $(\mathbb{C}^d)^{\otimes n}$ are polynomial, we will also restrict our discussion to polynomial $\gl(d)$-irreps.
Each such irrep, denoted $L_\lambda^d$, contains a unique (up to a scalar multiple) non-zero vector $\xi$, called the highest weight vector. 
This vector satisfies $E_{ii} \xi = \lambda_i \xi$ for all $i=1, \ldots, d$, and $E_{ij} \xi = 0$ for all $1 \le i < j \le d$.

The following discussion is taken from \cite{molev2006gelfand}.
There exists an orthogonal basis of $L_\lambda^d$ parametrized by Gelfand-Tsetlin (GT) patterns (as defined in \Cref{sec:ssyt-gt}) corresponding to the highest weight $\lambda=(\lambda_1, \ldots, \lambda_d)$. 
For a given GT pattern $\Lambda=(\lambda_{k,i})$ we denote these basis vectors by $\xi_{\Lambda}$.
For convenience, we also define $l_{k,i} = \lambda_{k,i} - i + 1$.
We then use the following result:

\begin{theorem}[{\cite[Thm.~2.3]{molev2006gelfand}}]
\label{thm:molev}
There exists a basis $\{\xi_{\Lambda}\}$ of $L_\lambda^d$, parametrized by GT patterns $\Lambda$ associated with the highest weight $\lambda$, such that the action of the generators of $\gl(d)$ is given by the following formulas:
\begin{itemize} 
\item 
For $k=1, \ldots, d$,
\begin{align}
E_{k,k} \xi_{\Lambda} = \left( \sum_{i=1}^{k} \lambda_{k,i} - \sum_{i=1}^{k-1} \lambda_{k-1,i} \right) \xi_{\Lambda}
\end{align}
where the second sum is zero if $k=1$.

\item 
For $k=1, \ldots, d-1$,
\begin{align}
E_{k,k+1} \xi_{\Lambda} &= -\sum_{i=1}^{k} \frac{(l_{k,i} - l_{k+1,1})\cdots(l_{k,i} - l_{k+1,k+1})}{(l_{k,i} - l_{k,1})\cdots\wedge\cdots(l_{k,i} - l_{k,k})} \xi_{\Lambda+\delta_{k,i}}\\
E_{k+1,k} \xi_{\Lambda} &= \sum_{i=1}^{k} \frac{(l_{k,i} - l_{k-1,1})\cdots(l_{k,i} - l_{k-1,k-1})}{(l_{k,i} - l_{k,1})\cdots\wedge\cdots(l_{k,i} - l_{k,k})} \xi_{\Lambda-\delta_{k,i}}.
\end{align}
\end{itemize}
The arrays $\Lambda \pm \delta_{k,s}$ are obtained from $\Lambda$ by replacing $\lambda_{k,s}$ by $\lambda_{k,s} \pm 1$. It is supposed that $\xi_{\Lambda'} = 0$ if the array $\Lambda'$ is not a valid GT pattern. The symbol $ \wedge$ in the product in the denominator indicates a skipped zero factor.
\end{theorem}

The vector space $L_\lambda^d$ has an inner product $\langle \cdot, \cdot \rangle$ with respect to which the basis vectors $\xi_\Lambda$ above are orthogonal.
Moreover, we have \cite[Prop.~2.4]{molev2006gelfand}
\begin{align}
\langle \xi_{\Lambda}, \xi_{\Lambda} \rangle = \prod_{k=2}^d \left( \prod_{1 \le s < j \le k} \frac{(l_{k,s} - l_{k-1,j})!}{(l_{k-1,s} - l_{k-1,j})!} \prod_{1 \le s \le j < k} \frac{(l_{k-1,s} - l_{k,j} -1)!}{(l_{k,s} - l_{k,j} -1)!} \right),
\label{eq:molev-inner-product}
\end{align}
where again $\Lambda = (\lambda_{k,i})$ and $l_{k,i} = \lambda_{k,i}-i+1$.

The relation \eqref{eq:molev-inner-product} can be used to construct an orthonormal basis for the vector space $L_\lambda^d$.
Expanding an arbitrary element $X\in\gl(d)$ in the standard basis $\lbrace E_{ij}\rbrace$, \Cref{thm:molev} then determines the matrix of the irrep $\varphi_\lambda(X)$, which is exponentiated to obtain the representation matrix of $\exp(X)\in\GL(d)$ acting on the irrep $V_\lambda^d$.

\subsection{Explicit irreps for \texorpdfstring{$\GL(2)$}{GL(2)}}
Constructing the irreps $(q_\lambda,V_\lambda^d)$ for $d=2$ is of special interest since we focus on analyzing qubit channels in this paper.
We record here an explicit construction of $\GL(2)$-irreps that is simpler than the approach in \Cref{sec:molev} and avoids matrix logarithms.
The latter will be important for our numerical studies (see the discussion in \Cref{sec:optimization}).

A standard result in the representation theory of the classical groups states that the irrep $V^2_\lambda$ of $\GL(2)$ for a partition $\lambda = (\lambda_1, \lambda_2)$ is isomorphic to 
\begin{align}
    V^2_\lambda \cong L_{\det}^{\otimes \lambda_2} \otimes \Sym^{m}(\mathbb{C}^2),
\end{align}
where $m = \lambda_1 - \lambda_2$ and $\Sym^m(\mathbb{C}^2)$ is the $m^{th}$ symmetric power of $\mathbb{C}^2$, and $L_{\det}$ is the one-dimensional determinant representation (see, e.g., \cite[Ch.~2.3]{goodman2009symmetry} for a proof for the special linear group $\SL(d)=\lbrace X\in\GL(d):\det(X)=1\rbrace$). 
For an invertible matrix $A = \begin{psmallmatrix} a & b \\ c & d \end{psmallmatrix}$, the corresponding matrix representation is
\begin{align}
    q_{\lambda}(A) = (\det A)^{\lambda_2} S_m(A),
\end{align}
where the entries of the $(m+1) \times (m+1)$ matrix $S_m(A)$ are given by
\begin{align} \label{eq:Sm_A_formula}
    [S_m(A)]_{k,j} = \sqrt{\frac{k!(m-k)!}{j!(m-j)!}} \sum_{p=\max(0, k-m+j)}^{\min(k,j)} \binom{j}{p}\binom{m-j}{k-p} a^p c^{j-p} b^{k-p} d^{m-j-k+p}.
\end{align}
The matrix indices $k$ and $j$ range from $m$ down to $0$. We give a proof of \eqref{eq:Sm_A_formula} in Appendix~\ref{app:gl2_irreps}.

\section{Computing coherent information using symmetries}
\label{sec:coherent-information-symmetries}

\subsection{High-level idea of our approach}
The quintessential example of superadditivity of coherent information is the repetition code \cite{shor1996syndrome} mentioned in \Cref{sec:superadditivity}:
\begin{align}
    \phi_n = \frac{1}{2}\left( \ketbra{0}^{\otimes n} + \ketbra{1}^{\otimes n}\right).
    \label{eq:repetition-code}
\end{align}
This code achieves high values of coherent information for noisy quantum channels such as the qubit depolarizing channel \cite{shor1996syndrome,divincenzo1998capacity}, other Pauli channels such as the BB84-channel (see \Cref{sec:BB84}) \cite{fern2008lower,bausch2021errorthresholds}, and the dephrasure channel (see \Cref{sec:dephrasure}) \cite{leditzky2018dephrasure}.

The starting point of our analysis is to notice the permutation symmetry of the repetition code \eqref{eq:repetition-code}: we have $P_\pi \phi_n P_\pi^\dagger = \phi_n$ for all $\pi\in\kS_n$, where $\pi\mapsto P_\pi$ denotes the representation of $\kS_n$ on $(\mathbb{C}^d)^{\otimes n}$ permuting tensor factors introduced in \Cref{sec:schur-weyl}.
Applying the i.i.d.~channel $\cN^{\otimes n}$ to this state gives an output state
\begin{align}
    \sigma_n = \cN^{\otimes n}(\phi_n) = \frac{1}{2}\left( \cN(\ketbra{0})^{\otimes n} + \cN(\ketbra{1})^{\otimes n}\right).
    \label{eq:output-repetition-code}
\end{align}
Evidently, the permutation invariance at the input of $\cN^{\otimes n}$ is preserved at the output, and the same principle applies more generally to any permutation-invariant state subjected to an i.i.d.~channel $\cN^{\otimes n}$.
We focus in particular on a subset of permutation-invariant states consisting of convex sums of i.i.d.~states: $\sum_{i=1}^k x_i\, \rho_i^{\otimes n},$ where $(x_i)_{i=1}^k$ is a probability distribution and each $\rho_i$ is a quantum state.
We will discuss our plans for future work to handle more general permutation-invariant states in \Cref{sec:discussion}.

The quantum states in \Cref{eq:repetition-code,eq:output-repetition-code} are symmetric with respect to the symmetry group $\kS_n$, and the representation theory of the symmetric group significantly reduces the number of degrees of freedom.
More precisely, following the discussions in \Cref{sec:rep-theory-basics,sec:schur-weyl}, the $\kS_n$-invariance implies that these states can be written as
\begin{align}
    \bigoplus\nolimits_\lambda X_\lambda \otimes \one_{S_\lambda},
    \label{eq:symmetric-state}
\end{align}
where $X_\lambda$ are operators acting on the multiplicity space of the $\kS_n$-irrep $S_\lambda$, which by Schur-Weyl duality is equal to the irrep $V_\lambda^d$ of the general linear group $\GL(d)$ (here, $d=2$).
Moreover, by \eqref{eq:weyl-dimension-formula} the dimension of $V_\lambda^d$ is at most polynomial in $n$, which suggests that entropic quantities of states of the form \eqref{eq:symmetric-state} can be evaluated efficiently.
This observation was used in \cite{kern2008improved} to evaluate entropies of permutation-invariant states to obtain bounds on quantum key distribution.
In this paper, we show that this approach also gives an efficient algorithm for computing the channel coherent information of a channel, a significantly more complicated entropic quantity.
The following sections develop this approach in more detail.

\subsection{Coherent information of convex sums of i.i.d.~states}\label{subsec: ci-convex-sums}

Recall that Schur-Weyl duality gives us a decomposition of $(\bC^d)^{\otimes n}$ as
\begin{align}
	(\bC^d)^{\otimes n} \cong \bigoplus_{\lambda\in\Lambda(n,d)} V_\lambda^d \otimes S_\lambda,\label{eq:schur-weyl}
\end{align}
where $\lambda\in\Lambda(n,d)$ denotes a Young diagram or partition of $n$ with at most $d$ parts.
For each $\lambda\in\Lambda(n,d)$ the spaces $V_\lambda^d$ and $S_\lambda$ are irreducible representations of the general linear group $\GL(d)$ and the symmetric group $\kS_n$, respectively.
We have the following actions of group elements $U\in\GL(d)$ and $\pi\in\kS_n$ with respect to this decomposition:
\begin{align}
	U^{\otimes n} P_\pi = P_\pi U^{\otimes n} \cong \bigoplus_{\lambda\in\Lambda(n,d)} q_\lambda(U) \otimes p_\lambda(\pi),
	\label{eq:U-P-action}
\end{align}
where $q_\lambda$ and $p_\lambda$ denote the irreps of $\GL(d)$ on $V_\lambda^d$ and of $\kS_n$ on $S_\lambda$, respectively.
The operators $q_\lambda(U)$ can be constructed using representations of the Lie algebra $\gl(d)$ as outlined in \Cref{sec:molev}.
Furthermore, this construction can also be extended to arbitrary (not necessarily invertible) positive semidefinite operators $X$ and their action on $(\bC^d)^{\otimes n}$ via $X^{\otimes n}$ by a continuity argument.\footnote{To see this, note that the invertible matrices are dense in the set of all matrices and the representations $q_\lambda(\cdot)$ appearing in \eqref{eq:schur-weyl} are polynomial, so that $q_\lambda(X)$ can be constructed from $q_\lambda(X+\eps\one)$ by taking the limit $\eps\to 0$.}
This is relevant in our analysis of permutation-invariant codes, as we will often consider pure input states $\psi$ for which $\cN(\psi)$ or $\cN_c(\psi)$ might not be full-rank (see \Cref{sec:pure-iid-states}).

We now explain how to use the representation-theoretic methods developed above to compute the coherent information of a large family of permutation-invariant quantum codes.
To this end, let $\rho$ be an arbitrary density operator and assume for the time being that $\rho$ has full rank, i.e., $\rho\in\GL(d)$ (we will return to this assumption in \Cref{sec:optimization} below).
By the above discussion, the $n$-fold tensor product operator $\rho^{\otimes n}$ can be written as
\begin{align}
	\rho^{\otimes n} \cong \bigoplus_{\lambda\in\Lambda(n,d)} q_\lambda(\rho) \otimes \one_{S_\lambda}.\label{eq:iid-state-schur-weyl}
\end{align}

Consider then the following special class of permutation-invariant states consisting of states in the convex hull of i.i.d.~states:
\begin{align} 
\rho_{(n)} = \sum_{i=1}^k x_i \rho_i^{\otimes n},
\label{eq:convex-mixture-iid}
\end{align}
for some quantum states $\rho_i$ and a probability distribution $x=(x_i)_i$.
Let $\cN\colon A\to B$ be a quantum channel and set $\sigma_n = \cN^{\otimes n}(\rho_{(n)}) = \sum_{i=1}^k x_i \,\cN(\rho_i)^{\otimes n}$.
The crucial point for our analysis is that $\cN^{\otimes n}(\rho_{(n)})$ remains permutation-invariant, since the i.i.d.~structure of $\cN^{\otimes n}$ preserves the permutation invariance of the individual $\rho_i^{\otimes n}$, and thus also that of the full state $\rho_{(n)}$ by linearity.
Each $\cN(\rho_i)^{\otimes n}$ can be written as $\bigoplus_{\lambda\in\Lambda(n,d_B)} q_\lambda(\cN(\rho_i)) \otimes \one_{S_\lambda}$ according to \eqref{eq:iid-state-schur-weyl}, with $d_B=\dim\cH_B$.
Hence,
\begin{align}
	\sigma_n = \cN^{\otimes n}(\rho_{(n)}) &= \bigoplus_{\lambda\in\Lambda(n,d_B)} \left(\sumi_{i=1}^k x_i q_\lambda(\cN(\rho_i))\right) \otimes \one_{S_\lambda}\\
	&= \bigoplus_{\lambda\in\Lambda(n,d_B)} c_\lambda \,\overline{Q}_\lambda^\cN \otimes \tau_{S_\lambda}.\label{eq:perm-inv-compact}
\end{align}
In the second line, the density matrices $\overline{Q}_\lambda$ and $\tau_{S_\lambda}$ and the probability distribution $\mathbf{c}=(c_\lambda)_{\lambda\in\Lambda(n,d_B)}$ are defined as follows:
\begin{align}
    \bar{q}_\lambda &= \sum_{i=1}^k x_i q_\lambda(\cN(\rho_i)) &
    \overline{Q}_\lambda^\cN &= \frac{1}{\tr\bar{q}_\lambda} \bar{q}_\lambda \label{eq:Q-operator}\\
    \tau_{S_\lambda} &= \frac{1}{\dim S_\lambda} \one_{S_\lambda} & c_\lambda &= \dim S_\lambda \tr \bar{q}_\lambda.\label{eq:c-distribution}
\end{align}
Recall the direct-sum property of von Neumann entropy:
\begin{align} 
    S\left(\bigoplus\nolimits_i p_i \rho_i\right) = H(\mathbf{p}) + \sum\nolimits_i p_i H(\rho_i),
    \label{eq:von-neumann-direct-sum}
\end{align} 
where $\mathbf{p}=(p_i)_i$ is a probability distribution, $\rho_i$ are states, and $H(\mathbf{p}) = -\sum_i p_i\log p_i$ is the Shannon entropy of $\mathbf{p}$.
Using \eqref{eq:von-neumann-direct-sum} in \eqref{eq:perm-inv-compact}, the von Neumann entropy of $\sigma_n$ is hence equal to
\begin{align}\label{eq:channeloutput-ci}
	S(\sigma_n) = H(\mathbf{c}) + \sum_{\lambda\in\Lambda(n,d_B)} c_\lambda \left(S\left(\overline{Q}_\lambda^\cN\right) + \log\dim S_\lambda\right).
\end{align}

The decomposition \eqref{eq:perm-inv-compact} works analogously for the output state $\cN_c^{\otimes n}(\rho_{(n)})$ of the complementary channel, and we denote by $\overline{Q}_\lambda^{\cN_c}$ the corresponding average operator defined via \eqref{eq:Q-operator}, and by $\mathbf{c}^{\cN_c}$ the corresponding probability distribution defined via \eqref{eq:c-distribution}. 
The formula \eqref{eq:coherent-info-complementary} for the coherent information thus yields a compact expression for the coherent information of a permutation-invariant input state $\rho_{(n)}$ as in \eqref{eq:convex-mixture-iid}, which is our first main result:

\begin{theorem}
\label{thm:coherent-info-iid-mixture}
Let $\cN\colon A\to B$ be a quantum channel with environment $E$, and set $d_B=\dim\cH_B$ and $d_E=\dim\cH_E$.
Let $\rho_{(n)}=\sum_{i=1}^k x_i \rho_i^{\otimes n}$ with quantum states $\rho_i$ on $\cH_A$ and a probability distribution $(x_i)_{i=1}^k$.
Then the coherent information $I_c(\cN^{\otimes n},\rho_{(n)})$ can be expressed as follows:
\begin{align}
    I_c(\cN^{\otimes n}, \rho_{(n)})
    &= \left[H\left(\mathbf{c}^{\cN}\right) - H\left(\mathbf{c}^{\cN_c}\right)\right] + \sum_{\lambda\in \Lambda(n,d_B)} c_\lambda^\cN S\left(\overline{Q}_\lambda^\cN\right) - \sum_{\mu\in\Lambda(n,d_E)} c_\mu^{\cN_c} S\left(\overline{Q}_\mu^{\cN_c}\right) \notag\\
    & \qquad {}+ \sum_{\lambda\in \Lambda(n,d_B)} c_\lambda^{\cN} \log\dim S_\lambda - \sum_{\mu\in\Lambda(n,d_E)} c_\mu^{\cN_c} \log\dim\cS_\mu,
    \label{ci-symmetric-states}
\end{align}
where $\overline{Q}_\lambda^\cN$ and $\overline{Q}_\mu^{\cN_c}$ are defined via \eqref{eq:Q-operator}, and $\mathbf{c}^\cN = (c_\lambda^\cN)_{\lambda\in \Lambda(n,d_B)}$ and $\mathbf{c}^{\cN_c} = (c_\mu^{\cN_c})_{\mu\in\Lambda(n,d_E)}$ are defined via \eqref{eq:c-distribution}. 
\end{theorem}

\subsection{Mixtures of pure i.i.d.~states}
\label{sec:pure-iid-states}

In this section, we show that the formula for the coherent information of symmetric states in \Cref{thm:coherent-info-iid-mixture} simplifies significantly when the input state is a sum of pure i.i.d.~states, i.e., the states $\rho_i$ in \eqref{eq:convex-mixture-iid} are all pure.
The results of \Cref{sec:results} show that our optimization approach often yields such convex mixtures of pure i.i.d.~states as the optimal input states maximizing coherent information, and hence more efficient and compact formulas for such input states are useful.

We prove two versions of this result: the first one in \Cref{thm: coherent-info-pure-iid} is a simplified version of the coherent information formula for symmetric states derived in \Cref{thm:coherent-info-iid-mixture}.
However, this formula makes use of the complementary channel $\cN\colon A\to E$, which creates a bottleneck in our numerics: the dimension $d_E$ of the environment $E$ is typically larger than $d_B=\dim\cH_B$, and the resulting dimensions of the relevant $\GL(d_E)$-irreps $V_\lambda^{d_E}$ grow much faster than those of $V_\lambda^{d_B}$ (see \Cref{tab:dim_comparison}).
To circumvent this bottleneck, \Cref{thm: coherent-info-pure-iid-purified} provides an alternative formula for the coherent information of mixtures of $k$ pure i.i.d.~states in the ``purified picture'' given by \eqref{eq:coherent-info-purification}, which only involves $(k\times k)$-block matrices and operators acting on irreps $V_\lambda^{d_B}$.

The starting point of \Cref{thm: coherent-info-pure-iid,thm: coherent-info-pure-iid-purified} is the following observation.
The formula \eqref{ci-symmetric-states} involves sums over Young diagrams $\lambda\vdash_{d_B} n$ with at most $d_B$ boxes, and Young diagrams $\mu\vdash_{d_E} n$ with at most $d_E$ boxes, respectively.
For example, for the 2-Pauli channel discussed in \Cref{sec:2-pauli} we have $d_B=2$ and $d_E=3$, so the sum over Young diagrams $\mu\vdash_{d_E}n$ will include diagrams not contained in $\Lambda(n,d_B)$.

However, for a pure state $\psi=\ketbra{\psi}$ the operators $\cN(\psi)$ and $\cN_c(\psi)$ have the same non-zero spectrum as marginals of the pure bipartite state $V|\psi\rangle$ (where $V\colon\cH_A\to \cH_B\otimes\cH_E$ is a channel isometry for $\cN\colon A\to B$). 
As we will see below, this fact implies that the environment output of a state $\rho_{(n)}$ as in \eqref{eq:convex-mixture-iid} with pure constituent states $\rho_i$ is only supported on blocks corresponding to Young diagrams $\lambda\vdash_{d_B} n$.\footnote{Note, however, that the dimension of $V_\lambda^{d_E}$ is generally much larger than that of $V_\lambda^{d_B}$ when $d_E>d_B$, which can be seen from \eqref{eq:weyl-dimension-formula} (see also \Cref{tab:dim_comparison}).}
Hence, only irreps associated to the same set of Young diagrams contribute to the coherent information in \eqref{ci-symmetric-states} for both receiver output and environment, giving rise to the following, much more compact formula:

\begin{theorem}\label{thm: coherent-info-pure-iid}
	Let $\cN\colon A\to B$ be a quantum channel with environment $E$, and assume that $d_B=\dim\cH_B \leq \dim\cH_E=d_E$.
	Then the coherent information of an input state $\rho_{(n)} = \sum_{i=1}^k x_i \rho_i^{\otimes n}$, with each $\rho_i$ a pure state and $(x_i)_{i=1}^k$ a probability distribution, is given by the following formula:
	\begin{align}
		I_c(\mathcal{N}^{\otimes n}, \rho_{(n)})
		&= \sum_{\lambda \in \Lambda(n, d_B)} c_\lambda \left[ S\left(\sigma_\lambda^{\mathcal{N}}\right) - S\left(\sigma_\lambda^{\mathcal{N}_c}\right) \right],\label{eq:Ic_simplified}
	\end{align}
	where, setting $\tau_i = \mathcal{N}(\rho_i)$ and $\omega_i = \mathcal{N}_c(\rho_i)$, we have the state
	\begin{align}
		\sigma_\lambda^{\mathcal{N}} = \frac{\sum_{i=1}^k x_i q_\lambda(\tau_i)}{\sum_{i=1}^k x_i \tr q_\lambda(\tau_i)}
	\end{align}
	on the $\GL(d_B)$-irrep space $V_\lambda^{d_B}$, while
	\begin{align}
		\sigma_\lambda^{\mathcal{N}_c} = \frac{\sum_{i=1}^k x_i q_\lambda(\omega_i)}{\sum_{i=1}^k x_i \tr q_\lambda(\omega_i)}
	\end{align} 
	is the state on the $\GL(d_E)$-irrep space $V_\lambda^{d_E}$ corresponding to the same partition $\lambda \in \Lambda(n, d_B)$.
	The coefficients $c_\lambda$ in \eqref{eq:Ic_simplified} are computed as
	\begin{align}
		c_\lambda = \dim S_\lambda \sum_{i=1}^k x_i \tr q_\lambda(\tau_i).
		\label{eq:c-lambda-compact}
	\end{align}
\end{theorem}

\begin{proof}
We will prove the special case $d_B=2$ and $d_E=3$.
The argument easily generalizes to any quantum channel $\cN\colon A\to B$ with environment $E$ and $d_B=\dim\cH_B\leq d_E=\dim\cH_E$.
 
Our goal is to show the following two claims for a state $\rho_{(n)}=\sum_{i=1}^k x_i \rho_i^{\otimes n}$ with all $\rho_i$ pure:
\begin{enumerate}[{Claim} (i):]
    \item\label{claim1} The relative irrep-weights $c_\lambda^{\cN_c}$ in the formula \eqref{ci-symmetric-states} are zero whenever $\lambda$ has more than two rows, i.e., $\lambda \in \Lambda(n,3)\setminus\Lambda(n,2)$.
    \item\label{claim2} On the remaining Young diagrams $\lambda\vdash_2 n$, we have $c_\lambda^\cN = c_\lambda^{\cN_c}$.
\end{enumerate}
To this end, recall from \eqref{eq:c-distribution} that 
\begin{align}
    c_\lambda^{\cN_c} = \dim S_\lambda \tr \bar{q}_\lambda = \dim S_\lambda \sum_{i=1}^k x_i \tr q_\lambda(\cN_c(\rho_i)) = \dim S_\lambda \sum_{i=1}^k x_i s_\lambda(r_{i,1},r_{i,2},r_{i,3}).\label{eq:c-with-schur-polynomials}
\end{align}
Here, $r_{i,1},r_{i,2},r_{i,3}$ are the eigenvalues of the $(3\times 3)$-matrix $\cN_c(\rho_i)$, and $s_\lambda(x)$ with $x=(x_1,\dots,x_f)$ is the \emph{Schur polynomial} indexed by $\lambda\vdash_{f} n$, defined as
\begin{align}
    s_\lambda(x) = \sum_{T\in\ssyt_{f}(\lambda)}\: \prod_{\square\in \lambda} x_{T(\square)},\label{eq:schur-polynomial}
\end{align}
where $T(\square)$ is the entry in $T$ in the box $\square\in\lambda$.

By assumption the state $\rho_i$ is pure, and hence $\cN(\rho_i)$ and $\cN_c(\rho_i)$ have the same non-zero spectrum.
Since $\cN(\rho_i)$ is a $(2\times 2)$-matrix with at most two non-zero eigenvalues $r_{i,1}$, $r_{i,2}$, the spectrum of $\cN_c(\rho_i)$ on which the Schur polynomials in \eqref{eq:c-with-schur-polynomials} are evaluated is thus of the form $\left(r_{i,1},r_{i,2},0\right)$.
If $\lambda$ has three rows, then any SSYT $T$ of shape $\lambda$ necessarily has a $3$ appearing in the third row, as the labels in an SSYT are strictly increasing along columns.
But then all terms in \eqref{eq:schur-polynomial} will involve the factor $x_3=r_{i,3}=0$, and hence $s_\lambda (r_{i,1},r_{i,2},0 )=0$ for such $\lambda\in\Lambda(n,3)\setminus \Lambda(n,2)$.
This holds for all states $\rho_i$, and hence $c_\lambda^{\cN_c}=0$, proving Claim \Cref{claim1} above.

To prove Claim \Cref{claim2} we observe that, for the same reason as above, any SSYT $T$ featuring a label $3$ will cause the respective summand in \eqref{eq:schur-polynomial} to vanish, and only $T\in \ssyt_2(\lambda)$ survive in the sum.
It follows that, for the remaining $\lambda\in\Lambda(n,3)\cap \Lambda(n,2)$,
\begin{align}
    s_\lambda (r_{i,1},r_{i,2},0 ) = \sum_{T\in\ssyt_3(\lambda)} \: \prod_{\square\in\lambda} r_{i,T(\square)} = \sum_{T\in\ssyt_2(\lambda)} \: \prod_{\square\in\lambda} r_{i,T(\square)} = s_\lambda (r_{i,1},r_{i,2} ),
\end{align}
and hence $c_\lambda^{\cN} = c_\lambda^{\cN_c}$, proving Claim \Cref{claim2}.

Denoting by $\mathbf{c} = \mathbf{c}^{\mathcal{N}} = \mathbf{c}^{\mathcal{N}_c}=(c_\lambda)_{\lambda \in \Lambda(n, d_B)}$ the common distribution of irrep-weights defined via \eqref{eq:c-distribution}, the expression \eqref{ci-symmetric-states} for the coherent information simplifies as the Shannon entropy terms $H(\cdot)$ and the sums over the dimensions of the $\kS_n$-irreps $S_\lambda$ cancel.
This concludes the proof.
\end{proof}

Finally, we prove a version of \Cref{thm: coherent-info-pure-iid} that avoids reference to the channel environment, instead involving block matrices of size linear in the irrep-dimensions of the channel output.
For qubit channels with four-dimensional environment (such as the generalized amplitude damping channel discussed in \Cref{sec:gadc}), this provides a significantly more favorable scaling of the resulting coherent information formula that allows us to go to a much larger number $n$ of channel copies.

\begin{theorem}\label{thm: coherent-info-pure-iid-purified}
	Let $\cN\colon A\to B$ be a quantum channel, and set $d_B=\dim\cH_B$.
	The coherent information of an input state $\rho_{(n)} = \sum_{i=1}^k x_i \ketbra{\psi_i}^{\otimes n}$ is given by the following formula:
	\begin{align}
		I_c(\mathcal{N}^{\otimes n}, \rho_{(n)})
		&= \sum_{\lambda\in\Lambda(n,d_B)} c_\lambda \left(S(\sigma_\lambda) - S(\omega_\lambda)\right)
	\end{align}
	where, setting $\sigma_{ij}=\cN(\ketbra{\psi_i}{\psi_j})$ for $i,j=1,\dots,k$, we have defined the normalized states
	\begin{align}
		\sigma_\lambda &= \frac{1}{\bar{q}_\lambda} \sum_{i=1}^k x_i\,  q_\lambda(\sigma_{ii}) \\
		 \omega_\lambda &= \frac{1}{\bar{q}_\lambda} \begin{pmatrix}
			x_1 q_\lambda(\sigma_{11}) & \sqrt{x_1 x_2} q_\lambda(\sigma_{12}) & \dots & \sqrt{x_1 x_k} q_\lambda(\sigma_{1k}) \\
			\sqrt{x_2 x_1} q_\lambda(\sigma_{21}) & x_2 q_\lambda(\sigma_{22}) & \dots & \sqrt{x_2 x_k} q_\lambda(\sigma_{2k})\\
			\vdots & \vdots & \ddots & \vdots\\
			\sqrt{x_k x_1} q_\lambda(\sigma_{k1}) & \sqrt{x_k x_2} q_\lambda(\sigma_{k2}) & \dots & x_k q_\lambda(\sigma_{kk})
		\end{pmatrix}
    \end{align} 
	with the normalization constant $\bar{q}_\lambda = \sum_{i=1}^k x_i \tr q_\lambda(\sigma_{ii})$, and coefficients $c_\lambda = \bar{q}_\lambda \dim S_\lambda$.
\end{theorem}

\begin{proof}
A purification of the input state $\rho_{(n)}$ is given by
\begin{align}
    |\phi\rangle_{RA^n} = \sum_{i=1}^k \sqrt{x_i} |i\rangle_R |\psi_i\rangle_A^{\otimes n},
\end{align}
where $\lbrace \ket{i}_R\rbrace_{i=1}^k$ is an orthonormal basis of the purifying system $R$.
We have
\begin{align}
	\ketbra{\phi}_{RA^n} &= \sum_{i,j=1}^k \sqrt{x_i x_j} \ketbra{i}{j}_R \otimes \ketbra{\psi_i}{\psi_j}_A^{\otimes n}\\
	&= \begin{pmatrix}
		x_1 \ketbra{\psi_1}_A^{\otimes n} & \sqrt{x_1 x_2} \ketbra{\psi_1}{\psi_2}_A^{\otimes n} & \dots & \sqrt{x_1 x_k} \ketbra{\psi_1}{\psi_k}_A^{\otimes n} \\
		\sqrt{x_2 x_1} \ketbra{\psi_2}{\psi_1}_A^{\otimes n} & x_2 \ketbra{\psi_2}_A^{\otimes n} & \dots & \sqrt{x_2 x_k} \ketbra{\psi_2}{\psi_k}_A^{\otimes n}\\
		\vdots & \vdots & \ddots & \vdots\\
		\sqrt{x_k x_1} \ketbra{\psi_k}{\psi_1}_A^{\otimes n} & \sqrt{x_k x_2} \ketbra{\psi_k}{\psi_2}_A^{\otimes n} & \dots & x_k \ketbra{\psi_k}_A^{\otimes n}
	\end{pmatrix},
\end{align}
where the block structure is with respect to the $R$ system.
Applying the i.i.d.~channel $\cN^{\otimes n}$ to the $A^n$-systems of $\ketbra{\phi}_{RA^n}$, we obtain the state
\begin{align}
	&(\id_R\otimes \cN^{\otimes n})\left(\ketbra{\phi}_{RA^n} \right) \notag\\
	&\qquad{} = \begin{pmatrix}
		x_1 \cN(\ketbra{\psi_1}_A)^{\otimes n} & \sqrt{x_1 x_2} \cN(\ketbra{\psi_1}{\psi_2}_A)^{\otimes n} & \dots & \sqrt{x_1 x_k} \cN(\ketbra{\psi_1}{\psi_k}_A)^{\otimes n} \\
		\sqrt{x_2 x_1} \cN(\ketbra{\psi_2}{\psi_1}_A)^{\otimes n} & x_2 \cN(\ketbra{\psi_2}_A)^{\otimes n} & \dots & \sqrt{x_2 x_k} \cN(\ketbra{\psi_2}{\psi_k}_A)^{\otimes n}\\
		\vdots & \vdots & \ddots & \vdots\\
		\sqrt{x_k x_1} \cN(\ketbra{\psi_k}{\psi_1}_A)^{\otimes n} & \sqrt{x_k x_2} \cN(\ketbra{\psi_k}{\psi_2}_A)^{\otimes n} & \dots & x_k \cN(\ketbra{\psi_k}_A)^{\otimes n}
	\end{pmatrix}\\
	&\qquad{} = \begin{pmatrix}
		x_1 \sigma_{11}^{\otimes n} & \sqrt{x_1 x_2} \sigma_{12}^{\otimes n} & \dots & \sqrt{x_1 x_k} \sigma_{1k}^{\otimes n} \\
		\sqrt{x_2 x_1} \sigma_{21}^{\otimes n} & x_2 \sigma_{22}^{\otimes n} & \dots & \sqrt{x_2 x_k} \sigma_{2k}^{\otimes n}\\
		\vdots & \vdots & \ddots & \vdots\\
		\sqrt{x_k x_1} \sigma_{k1}^{\otimes n} & \sqrt{x_k x_2} \sigma_{k2}^{\otimes n} & \dots & x_k \sigma_{kk}^{\otimes n}
	\end{pmatrix},
	\label{eq:block-matrix}
\end{align}
where we defined $\sigma_{ij} = \cN(\ketbra{\psi_i}{\psi_j}_A)$.

The crucial step is now to use the block form \eqref{eq:iid-state-schur-weyl} for the operators $\sigma_{ij}^{\otimes n}$ within each block in \eqref{eq:block-matrix}.
This can be understood as applying the basis change $\bigoplus_{i=1}^k U$, where $U$ is the unitary achieving the Schur-Weyl decomposition in \eqref{eq:schur-weyl}.
Using another suitable basis change that reorders rows and columns, the resulting operator can be written in block-diagonal form as
\begin{align}
	(\id_R\otimes \cN^{\otimes n})(\ketbra{\phi}_{RA^n}) &\cong \bigoplus_{\lambda\in\Lambda(n,d_B)} Q_\lambda \otimes \one_{S_\lambda}\\
	\text{with}\qquad Q_\lambda &= \begin{pmatrix}
		x_1 q_\lambda(\sigma_{11}) & \sqrt{x_1 x_2} q_\lambda(\sigma_{12}) & \dots & \sqrt{x_1 x_k} q_\lambda(\sigma_{1k}) \\
		\sqrt{x_2 x_1} q_\lambda(\sigma_{21}) & x_2 q_\lambda(\sigma_{22}) & \dots & \sqrt{x_2 x_k} q_\lambda(\sigma_{2k})\\
		\vdots & \vdots & \ddots & \vdots\\
		\sqrt{x_k x_1} q_\lambda(\sigma_{k1}) & \sqrt{x_k x_2} q_\lambda(\sigma_{k2}) & \dots & x_k q_\lambda(\sigma_{kk})
	\end{pmatrix},
	\label{eq:regrouped}
\end{align}
where as before we denote by $q_\lambda$ the $\GL(d_B)$-irrep acting on $V_\lambda^{d_B}$.
In order to compute entropies, we renormalize the operators appearing in \eqref{eq:regrouped},
\begin{align}
	(\id_R\otimes \cN^{\otimes n})(\ketbra{\phi}_{RA^n}) \cong \bigoplus_{\lambda\in\Lambda(n,d_B)} c_\lambda\, \omega_\lambda \otimes \tau_{S_\lambda}, \label{eq:purified-output}
\end{align}
where
\begin{align}
	\bar{q}_\lambda &= \sum_{i=1}^k x_i \tr q_\lambda(\sigma_{ii}) & \omega_\lambda &= \frac{1}{\bar{q}_\lambda} Q_\lambda\\
	\tau_\lambda &= \frac{1}{\dim S_\lambda} \one_{S_\lambda} & c_\lambda &= \bar{q}_\lambda \dim S_\lambda.
\end{align}
Recall from the proof of \Cref{thm:coherent-info-iid-mixture} that, using the notation of \eqref{eq:purified-output}, the state $\cN^{\otimes n}(\rho_{(n)})$ can be written as
\begin{align}
	\cN^{\otimes n}(\rho_{(n)}) \cong \bigoplus_{\lambda\in\Lambda(n,d_B)} c_\lambda \, \sigma_\lambda \otimes \tau_{S_\lambda}, \label{eq:output}
\end{align}
where 
\begin{align}
	q_\lambda &= \sum_{i=1}^k x_i \tr q_\lambda(\sigma_{ii}) & \sigma_\lambda &= \frac{1}{\bar{q}_\lambda} \sum_{i=1}^k x_i  q_\lambda(\sigma_{ii})\\
	\tau_\lambda &= \frac{1}{\dim S_\lambda} \one_{S_\lambda} & c_\lambda &= \bar{q}_\lambda \dim S_\lambda.
\end{align}
In particular, the coefficients $c_\lambda$ in \eqref{eq:purified-output} and \eqref{eq:output} are identical, just as in the proof of \Cref{thm: coherent-info-pure-iid}.

Using this fact in the formula \eqref{eq:coherent-info-purification}, the channel coherent information of the input state $\rho_{(n)} = \sum_{i=1}^k x_i \ketbra{\psi_i}^{\otimes n}$ has the simple form
\begin{align}
	I_c(\cN^{\otimes n},\rho_{(n)}) &= S\left(\cN^{\otimes n}(\rho_{(n)})\right) - S\left((\id_R\otimes \cN^{\otimes n})(\ketbra{\phi}_{RA^n})\right) \\
	&= \sum_{\lambda\in\Lambda(n,d_B)} c_\lambda \left(S(\sigma_\lambda) - S(\omega_\lambda)\right),
\end{align}
which concludes the proof.
\end{proof}

\subsection{Optimization method}
\label{sec:optimization}

Our goal is to find permutation-invariant input states $\rho_{(n)}=\sum_{i=1}^k x_i\rho_i^{\otimes n}$ with high coherent information $I_c(\cN^{\otimes n}, \rho_{(n)})$ using~\eqref{ci-symmetric-states} as the objective function.  The optimization is thus over the probability distribution $\{x_i\}$ and the set of states $\{\rho_i\}$. 

We perform this search using the stochastic global optimization algorithm Particle Swarm Optimization (PSO) \cite{kennedy1995particle}. 
The details of the state parametrizations used in PSO are provided in \Cref{app: numerics}. All entropies in~\eqref{ci-symmetric-states} are evaluated using the representation-theoretic machinery of \Cref{subsec: ci-convex-sums}.
A MATLAB implementation of this optimization for all channels considered in this paper will be made available at \cite{perm-inv-codes-github}.

\subsubsection{Optimization vs. single evaluations of the objective function}

 While our search space can include general mixed i.i.d.~states (using the formula in \Cref{thm:coherent-info-iid-mixture} for the objective function), we empirically find that the optimization consistently converges to ensembles of almost-pure states.

This observation motivates a more efficient optimization strategy that directly targets the class of pure i.i.d.~codes using \Cref{thm: coherent-info-pure-iid-purified}. 
This formula has the significant advantage of not requiring access to the complementary channel, which often maps to a higher-dimensional output and is thus computationally more costly. An implementation of this pure-state formula hinges on the method used to compute the irrep matrices $q_\lambda(\mathcal{N}(\ketbra{\psi}{\psi}))$, for which we have two approaches:
\begin{enumerate}
    \item The method described in \Cref{sec:molev} that, for any dimension $d$, computes the Lie algebra representation and obtains the group representation via the matrix exponential, $q_\lambda(\rho) = \exp(\phi_\lambda(\log \rho))$. While this can work well for a single evaluation of the coherent information at a code, it is ill-suited for use within an optimization loop over pure states. A pure input $\ketbra{\psi}{\psi}$ can yield a rank-deficient output state $\mathcal{N}(\ketbra{\psi}{\psi})$, for which the matrix logarithm need not be uniquely defined.
    \item For qubit channels ($d=2$), the explicit algebraic formula of Eq.~\eqref{eq:Sm_A_formula} circumvents this issue by directly expressing the matrix entries of $q_\lambda$ as a polynomial in the entries of its argument.
    This method thus avoids matrix logarithms and is numerically stable for rank-deficient operators.
\end{enumerate}
Therefore, our most effective strategy, used to obtain the results for the highest number of channel copies, is a PSO over either mixed states using \Cref{thm:coherent-info-iid-mixture}, or over pure-state codes using \Cref{thm: coherent-info-pure-iid-purified}. In either case, we find the optimal codes to be ensembles of pure or almost pure states. We can then evaluate the coherent information for the optimized code for a large number of channel copies using \Cref{thm: coherent-info-pure-iid-purified}. 

\subsubsection{Polynomial growth of $\dim V_{\lambda}^d$}

Although the dimensions of the irreps $V_{\lambda}^d$ grow only polynomially in $n$ for fixed $d$ (see \eqref{eq:weyl-dimension-formula}), this growth nevertheless establishes a new computational bottleneck for our approach. 
To illustrate this limit, Table~\ref{tab:dim_comparison} lists the maximum irrep dimensions for $\GL(2)$ and $\GL(4)$ for increasing $n$. 
For comparison, direct optimization techniques (such as those used in~\cite{bausch2020neural} with neural network quantum states) involve a full diagonalization of the output of a channel $\cN^{\otimes n}$ acting on a purified input, or alternatively diagonalizing the outputs of $\cN^{\otimes n}$ and the complementary channel $\cN_c^{\otimes n}$.
For a qubit-qubit channel this is typically feasible (on personal-use laptops) for up to $6$ copies of the channel, which requires diagonalizing a matrix of size $4096 \times 4096$ in each evaluation of the objective function.
We note here that in PSO the objective function is typically evaluated between 100 and 1000 times per iteration \cite{mathworks_particleswarm}.
In our representation-theoretic approach using \Cref{thm:coherent-info-iid-mixture}, similar dimensions are reached at $n = 15$ for channels with a four-dimensional environment space, as illustrated in \Cref{tab:dim_comparison}.

\begin{table}[t]
\centering
\caption{Growth of Hilbert space and maximum irrep dimensions for $n$ copies of a quantum channel with output system dimension $d_B=2$ and environment dimension $d_E=4$. While Schur-Weyl duality offers a drastic reduction compared to the full system-environment space, the polynomial growth of $\max_{\lambda} \dim(V^d_\lambda)$ still imposes a practical limit.}
\label{tab:dim_comparison}
\begin{tabular}{@{}ccccc@{}}
\toprule
$n$ & $\dim \cH_B^{\otimes n} = 2^n$ & $\max_{\lambda} \dim(V^2_\lambda)$ & $\dim \cH_E^{\otimes n} = 4^n$ & $\max_{\lambda} \dim(V^4_\lambda)$ \\ \midrule
6  & 64        & 7   & 4,096 & 140 \\
8  & 256       & 9   & 65,536 & 360  \\
10 & 1,024     & 11  & 1,048,576 & 770 \\
12 & 4,096     & 13  & 16,777,216 & 1,540  \\
15 & 32,768    & 16  & 1,073,741,824 & 4,004  \\
16 & 65,536    & 17  & 4,294,967,296 & 5,376  \\
\bottomrule
\end{tabular}
\end{table}

The partition $\lambda$ that maximizes $\dim V_{\lambda}^4$ for $n = 15, 16$ has $3$ parts. Even if we restrict to partitions with at most $2$ parts and use \Cref{thm: coherent-info-pure-iid} for calculating coherent information for mixtures of pure i.i.d.~states, we still see partitions with irreps of comparable size at $n = 15$ and $n = 16$: $V_{\lambda}^4$ has dimension $3640$ for $\lambda = (11,4)$ and $4725$ for $\lambda = (12,4)$. In contrast, evaluating the coherent information using \Cref{thm: coherent-info-pure-iid-purified} avoids the need for computing $\GL(d)$ irreps with $d = 4$ for qubit channels with $4$-dimensional environment and thus allows computing the coherent information at mixtures of i.i.d pure states for a significantly larger number of channel copies (even up to $n=100$).

\begin{table}[t]
\centering
\caption{Dimensions of the channels $\cN\colon A\to B$ with environment $E$ studied in this paper.  For each
channel we list the input dimension $d_A$, output dimension
$d_B$, and environment dimension $d_E$ that we use in the numerics.}
\label{tab:channel_dims}
\begin{tabular}{@{}lccc@{}}
\toprule
Channel & $d_A$ & $d_B$ & $d_E$ \\ \midrule
$2$-Pauli  (\Cref{sec:2-pauli})               & 2 & 2 & 3 \\
BB84  (\Cref{sec:BB84})               & 2 & 2 & 4 \\
General Pauli noise (\Cref{subsubsec: pauli-simplex}) & 2 & 2 & 4\\
Dephrasure  (\Cref{sec:dephrasure}) & 2 & 3 & 4 \\
Generalized amplitude-damping (\Cref{sec:gadc}) & 2 & 2 & 4 \\
Damping-dephasing (\Cref{sec:damping-dephasing})   & 2 & 2 & 3 \\ \bottomrule
\end{tabular}
\end{table}

\section{Results}
\label{sec:results}

In this section we apply our main theoretical results from \Cref{sec:coherent-information-symmetries}, a representation-theoretic formula for the channel coherent information of mixtures of i.i.d.~states in \Cref{thm:coherent-info-iid-mixture}, and a specialization to mixtures of pure i.i.d.~states in \Cref{thm: coherent-info-pure-iid,thm: coherent-info-pure-iid-purified}.
We use these results to efficiently optimize the channel coherent information of various channel families over this family of permutation-invariant states to obtain lower bounds on the quantum capacity of these channels.
Our method allows us to carry out this optimization for values of $n$ up to at least $100$.
However, for each channel we typically find that the optimal threshold increases up to a peak value of $n$, and then decreases.
For the sake of readability we only show numerical data up to this peak value of $n$ in the plots.

\subsection{Pauli channels}
We first study Pauli channels of the form 
\begin{align} 
	\cN_{\mathbf{p}} (\rho)  =  p_0 \rho + p_1 X\rho X + p_2 Y\rho Y + p_3 Z\rho Z,
\end{align}
where $\rho$ is a qubit density matrix, $\mathbf{p}=(p_0,p_1,p_2,p_3)$ is a probability distribution and $X,Y,Z$ are the Pauli matrices.
We focus in particular on the following special Pauli channels: the 2-Pauli channel $\mathbf{p}=(1-p,p/2,0,p/2)$, the BB84-channel $\mathbf{p}=((1-p)^2,p-p^2,p^2,p-p^2)$ and the depolarizing channel $\mathbf{p}=(1-p,p/3,p/3,p/3)$.
We apply the methods of \Cref{sec:coherent-information-symmetries} to optimize the coherent information of these channels over permutation-invariant codes of the form \eqref{eq:convex-mixture-iid}. 
As a benchmark, we compare such permutation-invariant codes to a weighted version $\phi_n^x$ of the repetition code \eqref{eq:repetition-code}, defined for $x\in[0,1]$ as
\begin{align}
    \phi_n^x = x \ketbra{0}^{\otimes n} + (1-x)\ketbra{1}^{\otimes n}.
    \label{eq:weighted-rep-code}
\end{align}
Note that the usual repetition code \eqref{eq:repetition-code} is recovered as $\phi_n = \phi_n^{1/2}$.
The following theorem gives an explicit formula for the coherent information of this weighted repetition code under Pauli noise.
A proof is given in \Cref{app:ci-pauli-rep}.

\begin{proposition}\label{thm:pauli_ci_rep}
The coherent information of $n$ copies of a Pauli channel 
\begin{align} \cN_{\mathbf{p}}(\rho) = p_0 \rho + p_1 X \rho X+ p_2 Y \rho Y+ p_3 Z \rho Z\end{align} 
evaluated on the weighted repetition code $\phi_n^x$ defined in \eqref{eq:weighted-rep-code} is given by
\begin{align}
    I_c(\cN_{\mathbf{p}}^{\otimes n}, \phi_n^x) &= 
    -\sum_{w = 0}^n \binom{n}{w} y_1(w)\log y_1(w)+ \sum_{w = 0}^n \binom{n}{w}  y_+(w) \log y_+(w) + \sum_{w = 0}^n \binom{n}{w}  y_-(w) \log y_-(w)
\end{align}
where
\begin{align}
    y_1(w) &= x(p_0 + p_3)^{n-w}(p_1 + p_2)^w + (1-x) (p_0 +p_3)^{w}(p_1 + p_2)^{n-w}\\
     y_{\pm}(w) &= \frac{(p_0+p_3)^{n-w}(p_1+p_2)^{w}\pm \sqrt{\Delta(w)}}{2}\\
    \Delta(w) &= [(2x -1)(p_0+p_3)^{n-w}(p_1 + p_2)^{w}]^2+4x(1-x)[(p_0-p_3)^{n-w}(p_1 - p_2)^{w}]^2.
\end{align}

\end{proposition}

Choosing $n=1$ and $x=1/2$ in \Cref{thm:pauli_ci_rep} corresponds to using the completely mixed state $\frac{1}{2}\one$ at the input of a Pauli channel $\cN_{\mathbf{p}}$ with $\mathbf{p}=(p_0,p_1,p_2,p_3)$.
This input state yields the maximal single-letter coherent information for a Pauli channel,\footnote{This is straightforward to prove for the depolarizing channel $\mathbf{p}=(1-p,p/3,p/3,p/3)$, but has only been confirmed numerically for general Pauli channels.} and thus the largest achievable rate of quantum information transmission through Pauli channels using product codes.
This is also known as the \emph{hashing bound} \cite{bennett1996mixed}:
\begin{align}
    I_c(\cN_{\mathbf{p}}) = 1 - H(\mathbf{p}).
    \label{eq:hashing-bound}
\end{align}

\subsubsection{2-Pauli channel}
\label{sec:2-pauli}

A $2$-Pauli channel is a special case of Pauli channels with only $X$ and $Z$ errors occurring with equal probability, defined as 
\begin{align} 
\cN_{\twoP{p}}(\rho) = (1-p) \rho + \frac{p}{2} \left( X\rho X + Z \rho Z \right).
\end{align}
Notably, and unlike the BB84 and depolarizing channels, repetition codes and more general concatenated repetition codes (also called cat codes) for the $2$-Pauli channel do not exceed the hashing bound \eqref{eq:hashing-bound} \cite{smith2006degenerate,fern2008lower}.
This channel has thus served as a challenging benchmark for code ans\"{a}tze giving superadditive coherent information.
The only example of superadditivity for the 2-Pauli channel known before our work are the codes derived in \cite{fern2008lower}, consisting of a suitable cat code concatenated repeatedly with the 5-qubit code from \cite{laflamme1996perfect}, requiring well over $10^9$ channel input systems to achieve a threshold higher than the hashing bound.

Our method of optimizing over permutation-invariant input states yields codes that significantly improve the quantum capacity threshold of the $2$-Pauli channel over those found in \cite{fern2008lower}.
Using our optimization method we find improved thresholds peaking at $n=24$.
These numerical results are summarized in \Cref{fig: 2paulithreshold}.
The coherent information values for $n=9,12,15,18,21,24$ in the plot are achieved using the input state\footnote{The probabilities $x_1,x_2$ in the state $\rho_{(n)} = x_1\psi_1^{\otimes n} + x_2\psi_2^{\otimes n}$ found in our optimization are actually equal to $1/2\pm 10^{-4}$ due to numerical noise.
For our analysis in this and the next section, we set these probabilities equal to $1/2$ for simplicity, and use the fact that the coherent information is continuous in the state \cite{winter2015tight}.}
\begin{align}
    \rho_{(n)} &= \frac{1}{2} \left( \psi_1^{\otimes n} + \psi_2^{\otimes n}\right) \label{eq:non-orthogonal-code}\\
    \text{with}\quad |\psi_1\rangle &= \begin{pmatrix} -0.0330-0.4220i\\0.9060\end{pmatrix}\\
     |\psi_2\rangle &= \begin{pmatrix} 0.0737-0.9030i\\0.4232\end{pmatrix}.
\end{align}
This code is also listed in \Cref{tab:perm-inv-codes}.
The individual states $\psi_1$ and $\psi_2$ are non-orthogonal, $|\langle\psi_1|\psi_2\rangle|=0.7645$, and hence the code in \eqref{eq:non-orthogonal-code} is qualitatively different from a repetition code such as $\phi_n$.
This is visualized in \Cref{fig:2pauli-code-bloch-sphere} using the Bloch sphere representation of qubit states. 

\begin{figure}[H]
\begin{center}
\includegraphics[scale = 0.13]{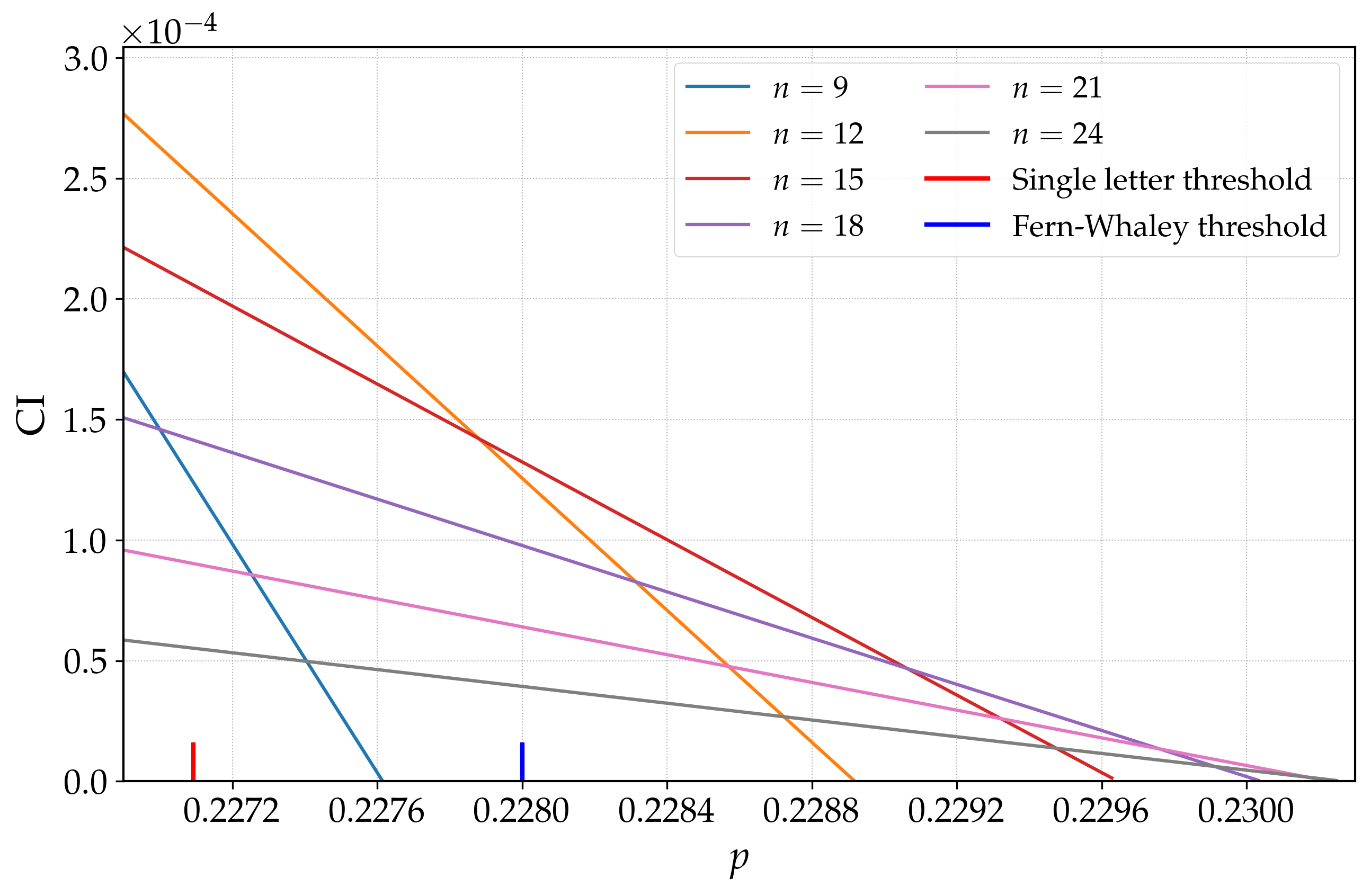}
\end{center}
\caption{Coherent information per channel use for the $2$-Pauli channel $\cN_{\twoP{p}}$ as a function of the error probability $p$. Solid lines represent the achievable rates for numerically optimized permutation-invariant codes for $n=9,12,15,18,21,24$ channel copies. The red vertical line on the $x$-axis is the hashing bound and the blue vertical line indicates the best known threshold using concatenated codes from \cite{fern2008lower}. The threshold error probability increases with $n$ up to $n =24$ and then decreases for higher $n$. 
}
\label{fig: 2paulithreshold}
\end{figure}

\begin{figure}[H]
	\centering
	\includegraphics[width=0.3\linewidth]{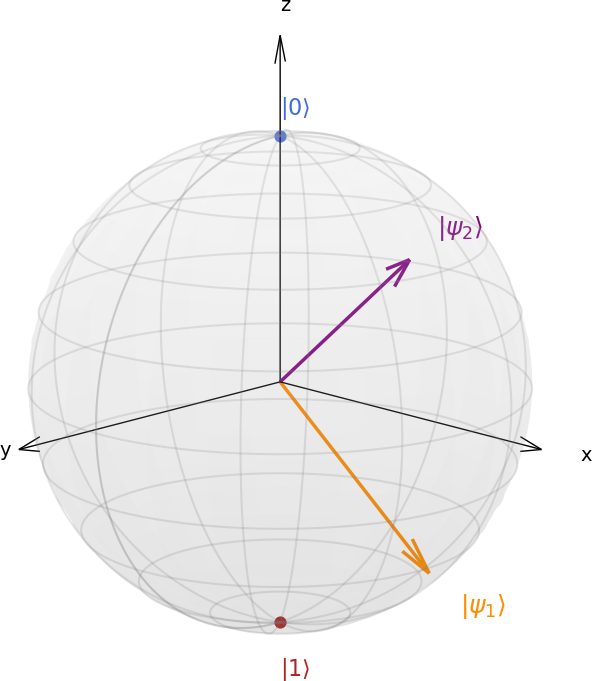}
	\caption{Bloch sphere showing the code states $|0\rangle$ and $|1\rangle$ (with Bloch coordinates $(0,0,\pm 1)$) of the repetition code $\phi_n$, together with the two states $\psi_i$ from \eqref{eq:non-orthogonal-code} (with Bloch coordinates $(-0.0597, 0.7647, -0.6416)$ and $(0.0624, 0.7643, 0.6418)$, respectively) that make up a non-orthogonal repetition code for the $2$-Pauli channel with a higher coherent information (see also \Cref{fig: 2paulithreshold,fig:2pauli-CI-by_partition}).}
	\label{fig:2pauli-code-bloch-sphere}
\end{figure}

We can attempt to understand the difference in performance between repetition codes and non-orthogonal codes such as \eqref{eq:non-orthogonal-code} by varying the overlap of the two states making up the code.
To this end, we consider two pure states $\chi_1,\chi_2$ with some fixed angle $\phi = \arccos|\langle \chi_1|\chi_2\rangle|$, and form a channel input state that is unitarily equivalent to $\frac{1}{2}\left(\chi_1^{\otimes n} + \chi_2^{\otimes n}\right)$ via a product unitary $U^{\otimes n}$ for a given qubit unitary $U$.
This unitary should be optimized in order to achieve the best coherent information in this ansatz.
Note that the choice $\phi=\pi/2$ gives a code that is unitarily equivalent to the repetition code $\phi_n$.
Fixing $n=9$, we optimize the coherent information over $U$ for varying angles $\phi$ and compare the coherent information contributions from the different $\GL(d)$-irreps according to \Cref{thm: coherent-info-pure-iid}.
We will see that certain irreps yield positive coherent information for suitable choices of $\phi \neq \pi/2$ that outweigh the negative contributions from other irreps, resulting in an overall positive coherent information as shown in \Cref{fig: 2paulithreshold}.

We now discuss the details of this approach.
Fix a noise level $p = 0.2271$ for the 2-Pauli channel, and consider the following ansatz for the input state:
\begin{align}
    \rho_U(\phi) = \frac{1}{2}\left( (U|0\rangle\langle 0|U^\dagger)^{\otimes n} + ( U|\psi_\phi\rangle\langle\psi_\phi|U^\dagger)^{\otimes n} \right),\label{eq:ansatz}
\end{align}
where $U$ is a qubit unitary (to be optimized over), and 
\begin{align}
    |\psi_\phi\rangle = \cos(\phi)|0\rangle + e^{i\theta} \sin(\phi) |1\rangle
\end{align}
with $\phi\in[0,\pi/2]$ and some fixed angle $\theta$.
We choose $\theta = -1.5249$ so that the ansatz \eqref{eq:ansatz} reproduces the optimal state \eqref{eq:non-orthogonal-code} in \Cref{fig: 2paulithreshold} for $n=9$, $\phi=0.7005$ and $U$ chosen appropriately.
The choice $\phi=\pi/2$ in \eqref{eq:ansatz} gives a repetition code in the usual sense (since $\ketbra{\psi_{\pi/2}}=\ketbra{1}$), that is, an equal superposition of i.i.d.~copies of two orthogonal pure states.
This ansatz thus parametrizes a ``non-orthogonal repetition code'', with $\phi$ controlling the angle between the two code states.

We compute the coherent information 
\begin{align}
I_c(\phi) = \max_U I_c(\cN_{2P}^{\otimes n},\rho_U)
\end{align}
for $\phi\in[0,\pi/2]$ and $n=9$.
Using the formula \eqref{eq:Ic_simplified}, we can write $I_c(\phi) = \sum_{\lambda\in\Lambda(n,d)} I_c^{(\lambda)}$ with
\begin{align} 
	I_c^{(\lambda)} = c_\lambda\left( S\left(\sigma_\lambda^{\cN}\right)-S\left(\sigma_\lambda^{\cN_c}\right)\right)
	\label{eq:I_c^lambda}
\end{align} 
being the weighted contribution to the coherent information in the irrep block $\lambda$ (see \Cref{thm: coherent-info-pure-iid} for definitions of the quantities appearing in this expression).
For $n=9$ the irreps appearing in this formula are $\lambda=(9,0), (8,1), (7,2), (6,3), (5,4)$.
The results of this analysis are plotted in \Cref{fig:2pauli-CI-by_partition}.
We see that for angles $0.65 \lesssim \phi \lesssim 1.2$ the two irreps $\lambda=(9,0), (8,1)$ contribute coherent information $I_c^{(\lambda)}$ that offsets the negative contributions from the other irreps, resulting in a total positive coherent information.
For the usual unweighted repetition code at $\phi=\pi/2$ the negative CI-contributions from the other irreps dominate.

\begin{figure}[H]
\begin{center}
\includegraphics[scale = 0.2]{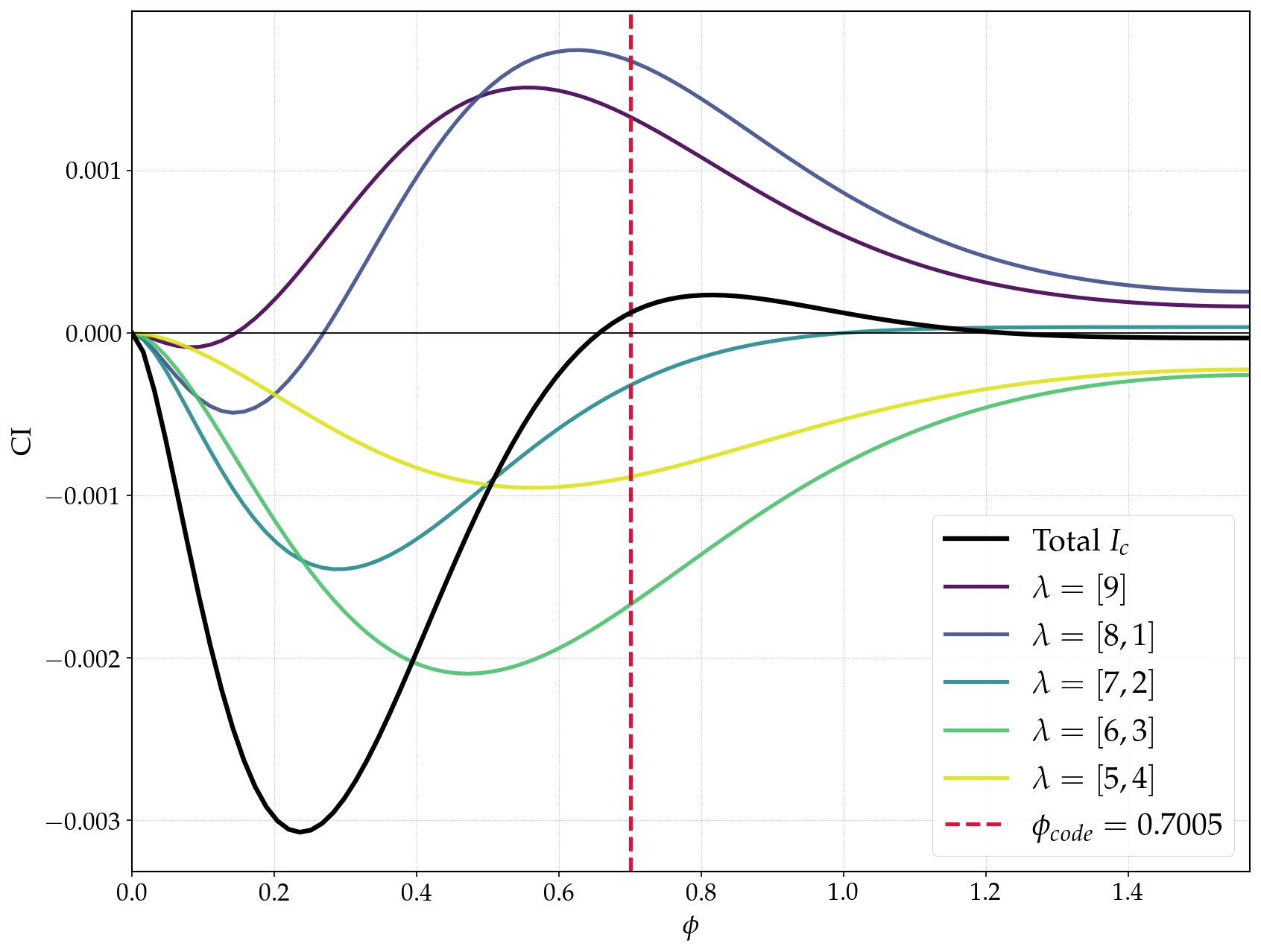}
\end{center}
\caption{Plot of the different irrep-contributions $I_c^{(\lambda)}$ (defined in \eqref{eq:I_c^lambda}) to the coherent information of $n=9$ copies of the 2-Pauli channel at the noise level $p=0.2271$, evaluated at the state \eqref{eq:ansatz}.
The red dashed line marks the value of $\phi$ that gives a code equivalent to the one in \eqref{eq:non-orthogonal-code}.}
\label{fig:2pauli-CI-by_partition}
\end{figure}

\subsubsection{BB84 channel}
\label{sec:BB84}

The BB84 channel is a special Pauli channel of the form 
\begin{align} 
\cN_{\bb{p}}(\rho) = (1-p)^2 \rho + (p-p^2) X\rho X + p^2 Y \rho Y + (p-p^2) Z \rho Z.
\end{align}
It corresponds to independent bit flip and phase flip errors occurring with equal probability. 
This is a common noise model in the analysis of quantum key distribution \cite{bennett1984cryptography,smith2008structured} and error correcting codes \cite{dennis2002topological,niwa-lee-css-codes-ci}.

Once again our permutation-invariant ansatz yields codes that greatly improve the known quantum capacity thresholds for the BB84 channel.
These thresholds peak at $n=18$, which we summarize in \Cref{fig: bb84threshold}.
The coherent information values plotted for $n=6, 9, 12, 15, 18$ are achieved on the state
\begin{align}
    \rho_{(n)} &= \frac{1}{2}\left( \psi_1^{\otimes n} + \psi_2^{\otimes n}\right) \label{eq:non-orthogonal-code-bb84}\\
    \text{with}\quad |\psi_1\rangle &= \begin{pmatrix}
        0.0207+0.3739i\\
        0.9272
    \end{pmatrix}\\
    |\psi_2\rangle &= \begin{pmatrix}
        0.0511-0.9258i\\
        -0.3745
    \end{pmatrix},
\end{align}
which is also listed in \Cref{tab:perm-inv-codes}.
We have $|\langle\psi_1|\psi_2\rangle|=0.6934$, and hence this is another example of a non-orthogonal repetition code giving significantly higher values of coherent information than orthogonal repetition codes.

\begin{figure}[H]
\begin{center}
\includegraphics[scale = 0.13]{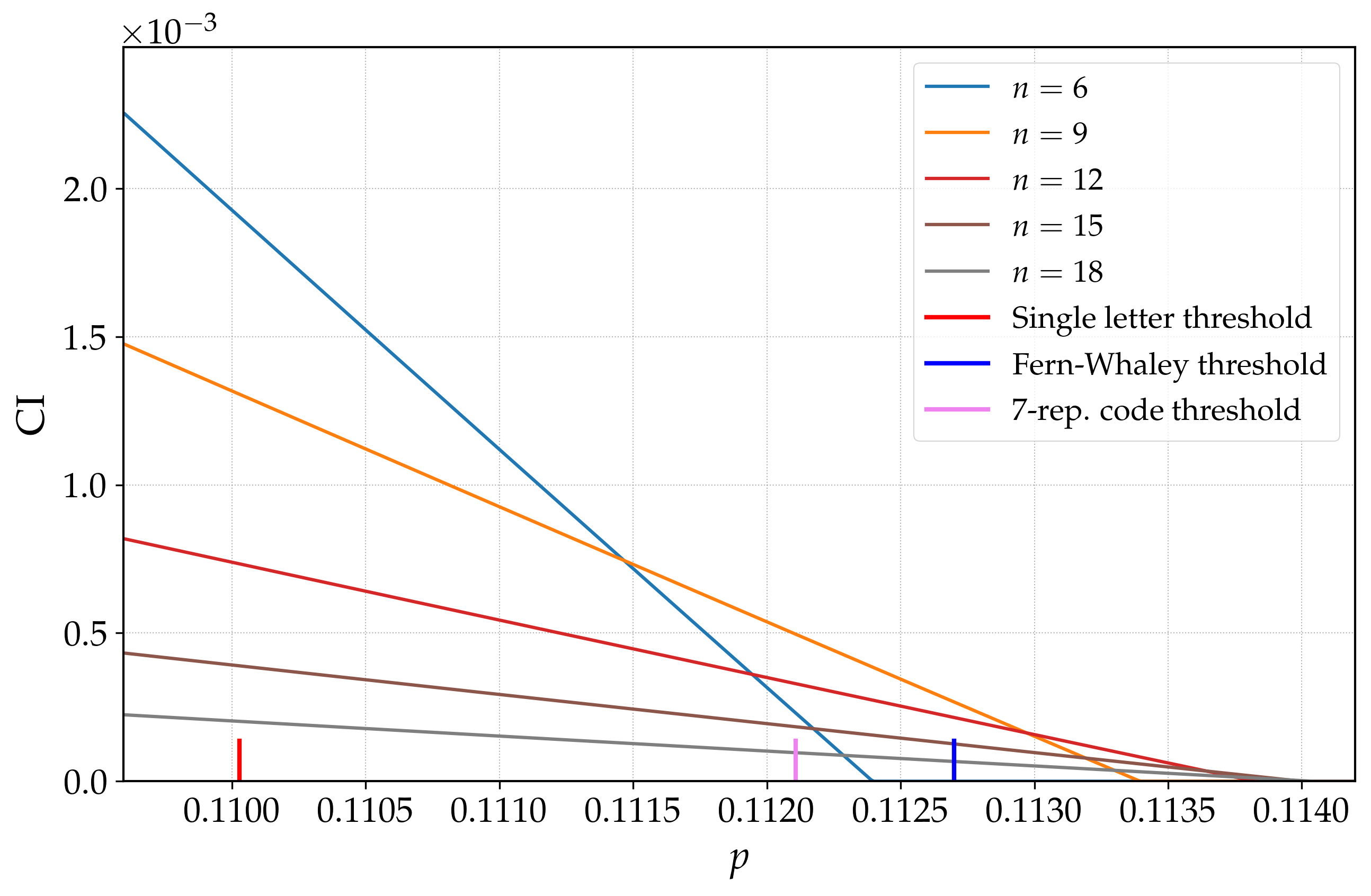}
\end{center}
\caption{Coherent information per channel use for the BB84 channel $\cN_{\bb{p}}$ as a function of the error probability $p$. Solid lines represent the achievable rates for numerically optimized permutation-invariant codes for $n = 6,9,12, 15, 18$ channel copies. The red vertical line on the $x$-axis is the hashing bound, the pink vertical line is the best weighted repetition code threshold (computed via \Cref{thm:pauli_ci_rep}), and the blue vertical line indicates the best known threshold using concatenated codes from \cite{fern2008lower}. We find that the threshold increases with the number of channel copies $n$ up to $n = 18$ and then decreases.
}
\label{fig: bb84threshold}
\end{figure}

We repeat our analysis from \Cref{sec:2-pauli} to study how well these non-orthogonal codes perform for the BB84-channel as a function of the angle $\phi = \arccos|\langle\psi_1|\psi_2\rangle|$ by considering the coherent information contributions within the different irrep blocks for fixed $n=9$.
This time, in the ansatz
\begin{align}
    \rho_U(\phi) &= \frac{1}{2}\left( (U|0\rangle\langle 0|U^\dagger)^{\otimes n} + ( U|\psi_\phi\rangle\langle\psi_\phi|U^\dagger)^{\otimes n} \right) \label{eq:ansatz-bb84}\\
    |\psi_\phi\rangle &= \cos(\phi)|0\rangle + e^{i\theta} \sin(\phi) |1\rangle
\end{align}
we choose $\theta = 1.5922$, so that \eqref{eq:ansatz-bb84} reproduces the optimal state \eqref{eq:non-orthogonal-code-bb84} for $\phi = \arccos|\langle\psi_1|\psi_2\rangle| = 0.8046$ and $U$ chosen accordingly.

We carry out the same optimization of the resulting coherent information over $U$ for $n=9$ as in \Cref{sec:2-pauli}, and plot the coherent information contributions in the different $\lambda$-sectors for $\lambda = (9,0),$ $(8,1),$ $(7,2),$ $(6,3),$ $(5,4)$.
The results are plotted in \Cref{fig:bb84-CI-by-partition}.
Similar to before, the coherent information contributions from the irreps $(9,0)$ and $(8,1)$ offset the negative contributions from the other irreps for angles $0.6\lesssim \phi\leq \pi/2$.

\begin{figure}[H]
\begin{center}
\includegraphics[scale = 0.2]{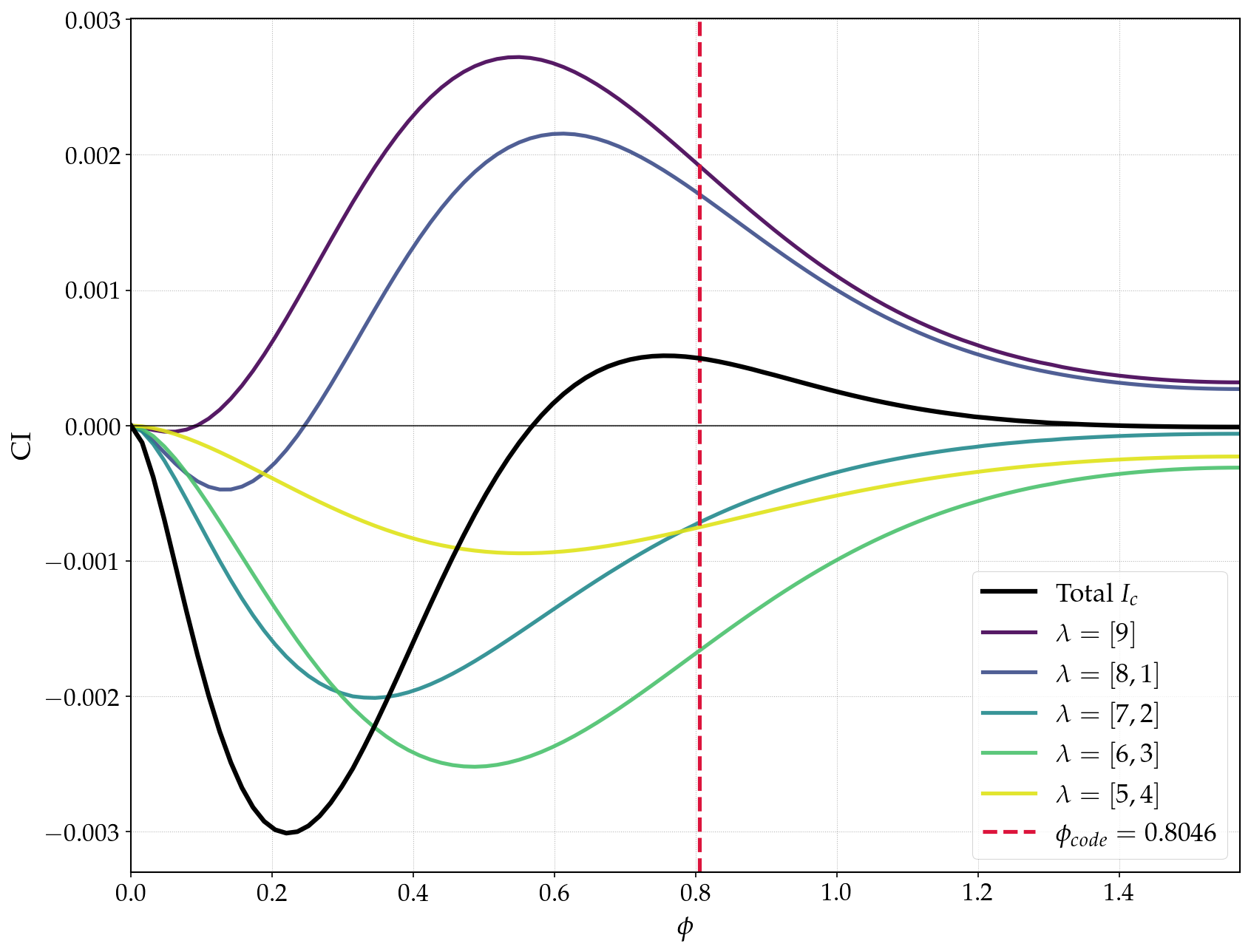}
\end{center}
\caption{Plot of the different irrep-contributions $I_c^{(\lambda)}$ (defined in \eqref{eq:I_c^lambda}) to the coherent information of $n=9$ copies of the BB84 channel at the noise level $p=0.112105$, evaluated at the state \eqref{eq:ansatz-bb84}.
The red dashed line marks the value of $\phi$ that gives a code equivalent to the one in \eqref{eq:non-orthogonal-code-bb84}.}
\label{fig:bb84-CI-by-partition}
\end{figure}

\subsubsection{Pauli channel simplex}\label{subsubsec: pauli-simplex}

The results of \Cref{sec:2-pauli,sec:BB84} suggest that non-orthogonal repetition codes can significantly improve the coherent information of certain Pauli channels over that of orthogonal codes.
This prompts us to study the performance of such non-orthogonal repetition codes on the entire Pauli channel simplex for a fixed number of channel copies using the approach in \eqref{eq:ansatz} in \Cref{sec:2-pauli}. 

To this end, we use the following parametrization for the Pauli channel simplex that was also used in \cite{bausch2021errorthresholds}.
As before, we denote a general Pauli channel by $\cN_{\mathbf{p}}(\rho)  = p_0 \rho + p_1 X \rho X+ p_2 Y \rho Y+ p_3 Z \rho Z$, where $\mathbf{p} = (p_0, p_1, p_2, p_3)$ is a probability distribution.
Since we are mainly interested in quantum capacity thresholds in this section, it is useful to divide the entire family of Pauli channels into one-parameter families $\cN_{\mathbf{p_x}}$ where $\mathbf{p_x} = (1-x, xq_1, xq_2, xq_3)$ for a fixed probability distribution $(q_1,q_2,q_3)$.
It was shown in \cite{bausch2021errorthresholds} that for each probability distribution $(q_1, q_2,q_3)$ the one-parameter family of channels $x \to \cN_{\mathbf{p_x}}$ has a unique value or \emph{threshold} for the noise parameter, defined as the smallest $x$ such that the quantum capacity of $\cN_{\mathbf{p}_y}$ vanishes for all $y \geq x$. 

The set of antidegradable Pauli  channels $\cN_{\mathbf{p}}$ is known to consist of exactly those channels that satisfy the inequality \cite{paddock2017characterization} 
\begin{align} 
    1 \geq 2(p_0^2 + p_1^2+ p_2^2+p_3^2) - 8 \sqrt{p_0p_1p_2p_3}.
\end{align}
We restrict our attention to those Pauli channels within the probability simplex $\{(q_1,q_2, q_3) : q_i \geq 0, \sum_i q_i = 1\}$ that are not anti-degradable and closest to the origin (see \cite[Fig.~1]{bausch2021errorthresholds} for a visualization of this set of channels). We parametrize such channels using spherical coordinates: 
\begin{align} \label{eq:spherical-pauli-param}
(q_1, q_2, q_3) \in \{(\sin^2 \theta \cos^2 \phi, \sin^2 \theta \sin^2 \phi, \cos^2 \theta): \theta, \phi \in \{0, \delta, 2 \delta, \cdots, \pi/2.\}\}.
\end{align}
For $\delta = 2^{-6} \pi$, this gives a covering of the Pauli channel simplex with $32 \times 32$ rays of length $\frac{1}{2}$ starting at the origin. 

We use $n = 5,10$ and $15$ copies of each channel and once again use the ansatz \begin{align} 
    \rho_U(\phi) &= \frac{1}{2}\left( (U|0\rangle\langle 0|U^\dagger)^{\otimes n} + ( U|\psi_\phi\rangle\langle\psi_\phi|U^\dagger)^{\otimes n} \right) \label{eq:ansatz-simplex}\\
    |\psi_\phi\rangle &= \cos(\phi)|0\rangle + e^{i\theta} \sin(\phi) |1\rangle
\end{align}
from \Cref{sec:2-pauli}. We fix $\phi$ to be one of $\frac{\pi}{8},\frac{\pi}{4}$ and $\frac{\pi}{2}$.
Along each ray with the parametrization in \eqref{eq:spherical-pauli-param}, we first determine the single-letter threshold of the one-parameter family $\cN_{\mathbf{p}_x}$ with $\mathbf{p}_x = (1-x, xq_1, xq_2, xq_3)$, i.e., the value of $x$ at which the single letter coherent information of this family $\cN_{\mathbf{p}_x}$ is zero.
This can easily be calculated using the hashing bound \eqref{eq:hashing-bound}.
At the resulting value of $x$ at the single-letter threshold, we optimize the coherent information of the ansatz \eqref{eq:ansatz-simplex} over the unitary $U$ and angle $\theta$ for each $\phi = \frac{\pi}{8},\frac{\pi}{4},\frac{\pi}{2}$ and $n = 5,10,15$ copies of each channel family $\cN_{\mathbf{p}_x}$ with $\mathbf{p}_x = (1-x, xq_1, xq_2, xq_3)$.  For this code found at the hashing bound along each ray given by \eqref{eq:spherical-pauli-param}, we compute the threshold along the same ray. 

We plot the Pauli channel simplex for each $n = 5,10,15$ and $\phi = \frac{\pi}{8},\frac{\pi}{4},\frac{\pi}{2}$ in \Cref{fig: pauli-simplex-thresholds}. Each point in each of the plots is colored based on the distance of the threshold of the best code of the form \eqref{eq:ansatz-simplex} to the single-letter threshold along the same ray. Points in red have thresholds higher than the hashing bound, while points in blue have a threshold lower than the hashing bound. We find that thresholds for the non-orthogonal repetition codes with $\phi = \pi/4$ increase across the Pauli channel simplex relative to the hashing bound as $n$ increases, and $n =6$ is the lowest number of channel copies for which we see superadditivity of the coherent information for the $2$-Pauli channel.

\begin{figure}[H]
\begin{center}
\includegraphics[width=\textwidth]{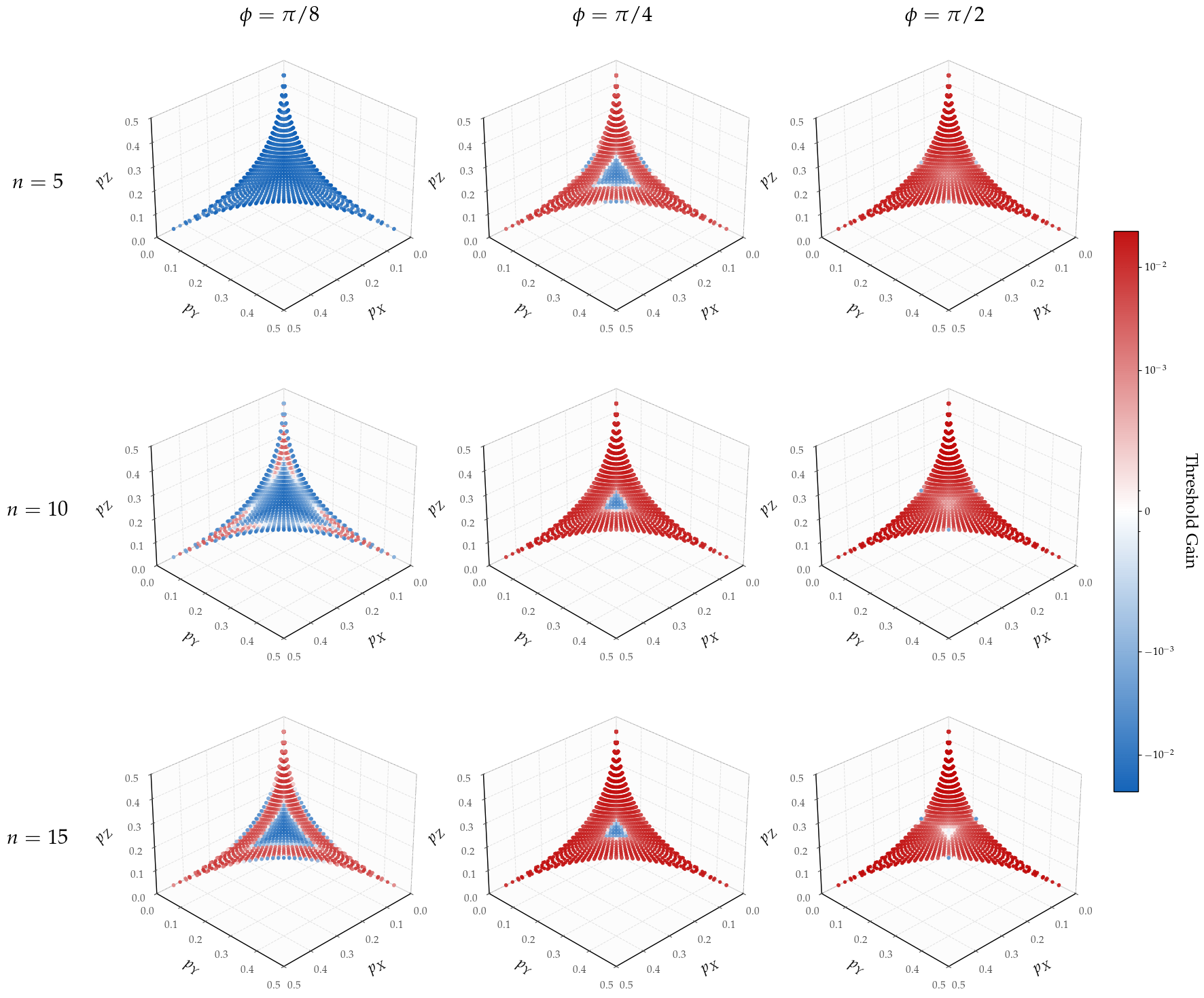}
\end{center}
\caption{Performance of permutation-invariant codes from Eq.~\eqref{eq:ansatz-simplex} with $\phi = \frac{\pi}{8},\frac{\pi}{4},\frac{\pi}{2}$ across the Pauli channel simplex for $n = 5,10,15$ channel copies. Each panel visualizes the simplex of Pauli channels via their error probabilities $(q_1, q_2, q_3)$, sampled along rays from the origin as described in Section~\ref{subsubsec: pauli-simplex}. For each point, the color indicates the difference between the error threshold achieved by the optimized $n$-copy code and the single-letter hashing bound. Red regions signify a threshold higher than the hashing bound, indicating superadditivity of the coherent information, whereas blue regions signify a threshold below the hashing bound. The plots for $\phi=\pi/4$ show that as $n$ increases, the region of superadditivity grows. The plots for $\phi=\pi/2$, which correpond to the usual repetition codes with orthogonal states, show that one does not observe superadditive coherent information using these codes for the $2$-Pauli channel. At $n= 10$, one even sees superadditivity for some channels for $\phi = \pi/8.$}
\label{fig: pauli-simplex-thresholds}
\end{figure}

\subsubsection{Comments on the depolarizing channel}
\label{sec:depolarizing}

The qubit depolarizing channel
\begin{align}
	\cN_{\dep{p}}(\rho) = (1-p)\rho + \frac{p}{3}\left(X\rho X + Y\rho Y + Z\rho Z\right)
\end{align}
is one of the most thoroughly studied Pauli channels, and one of the first quantum channels for which superadditivity was observed \cite{shor1996syndrome,divincenzo1998capacity}.
The depolarizing channel family is the unique family of qubit-qubit channels with full unitary covariance, $\cN_{\dep{p}}(U\cdot U^\dagger) = U\cN_{\dep{p}}(\cdot)U^\dagger$ for all $U\in \cU_2$.

Superadditivity of coherent information for the depolarizing channel is achieved on different code families, including repetition codes \cite{shor1996syndrome}, mixtures of code states of certain quantum error-correcting (QEC) codes \cite{divincenzo1998capacity,smith2006degenerate,fern2008lower} and graph states defined in terms of certain tree graphs \cite{bausch2021errorthresholds}.
A particular example of a QEC-based quantum code is the concatenated repetition code (also called cat code or Shor-type code) of the form
\begin{align}
	\chi_{k,m} &= \frac{1}{2}\left( \ketbra{+_k}^{\otimes m} + \ketbra{-_k}^{\otimes m}\right) \label{eq:cat-code}\\
	|\pm_k\rangle &= \frac{1}{\sqrt{2}}\left(|0\rangle^{\otimes k} \pm |1\rangle^{\otimes k}\right).
\end{align}
Observe that the cat code \eqref{eq:cat-code} does \emph{not} have full permutation symmetry like weighted repetition codes \eqref{eq:weighted-rep-code} or the code in \eqref{eq:non-orthogonal-code}.
Instead, the symmetry group of the cat code $\chi_{k,m}$ is the wreath product $\kS_k\wr \kS_m$ (which can be defined using semidirect products).
Hence, our optimization ansatz developed in \Cref{sec:coherent-information-symmetries} does not include cat codes, which are known to extend the quantum capacity threshold of the depolarizing channel beyond that of repetition codes \cite{divincenzo1998capacity,smith2006degenerate,fern2008lower}.
This opens an avenue of future work which we discuss in \Cref{subsec: open-questions}.

For the depolarizing channel $\cN_{\dep{p}}$ with values of $p\approx 0.19$, our optimization method yields the unweighted repetition code as the optimal permutation-invariant code.
Interestingly, the non-orthogonal repetition codes studied in the previous \Cref{sec:2-pauli,sec:BB84,subsubsec: pauli-simplex} do \emph{not} seem to provide any advantage over orthogonal repetition codes.
To illustrate this, we optimize the coherent information of $n=5$ copies of the depolarizing channel $\cN_{\dep{0.19}}$ over states of the form \eqref{eq:ansatz-bb84} (note that the optimization over the unitary $U$ is not necessary due to the unitary covariance of the depolarizing channel).
For varying angles $\phi$ we plot the different irrep-contributions $I_c^{(\lambda)}$ in \Cref{fig:depol-CI-by-partition}.
Evidently, the negative CI-contributions dominate except for a neighborhood around $\phi=\pi/2$ in which the orthogonal repetition code yields positive coherent information \cite{shor1996syndrome}.
Similar results also hold for higher $n$.
In \Cref{subsec: open-questions} we outline a generalization of our ansatz to codes with wreath product symmetry such as \eqref{eq:cat-code}, which we expect to yield improved quantum capacity thresholds for the depolarizing channel.

\begin{figure}[H]
	\begin{center}
		\includegraphics[width=0.75\textwidth]{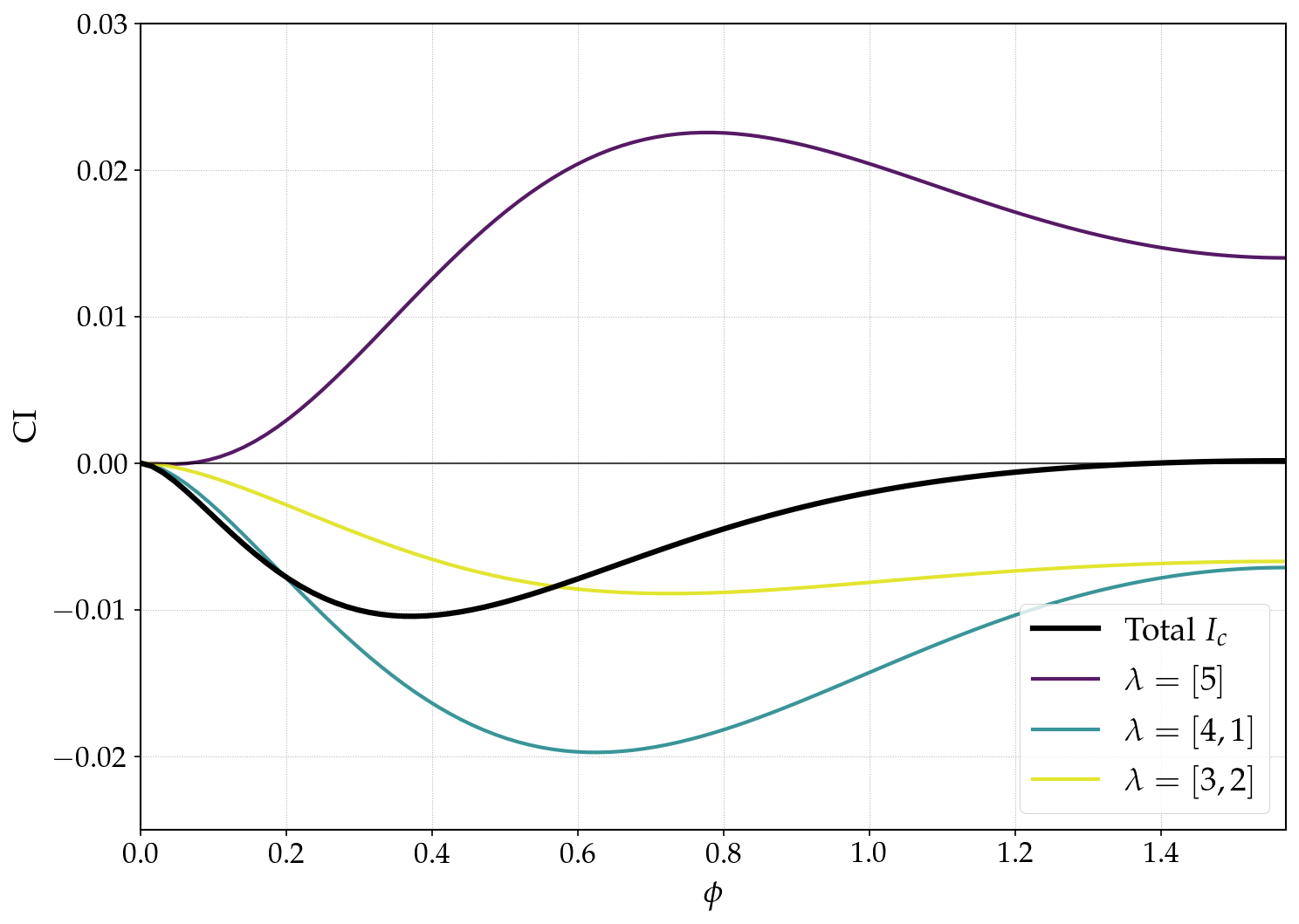}
	\end{center}
	\caption{Plot of the different irrep-contributions $I_c^{(\lambda)}$ (defined in \eqref{eq:I_c^lambda}) to the coherent information of $n=5$ copies of the depolarizing channel at the noise level $p=0.19$, evaluated at the state \eqref{eq:ansatz-bb84}.}
	\label{fig:depol-CI-by-partition}
\end{figure}

\subsection{Dephrasure channel}
\label{sec:dephrasure}

The dephrasure channel is a concatenation of a $Z$-dephasing channel and an erasure channel \cite{leditzky2018dephrasure}. 
It is defined for $p,q \in [0,1]$ as 
\begin{align}
    \cN_{p,q}(\rho) = (1-q)((1-p) \rho + p Z \rho Z) + q \ketbra{e}{e},
\end{align}
where $|e\rangle$ is an erasure flag orthogonal to the input space.
The $Z$-dephasing and erasure channels by themselves are degradable and thus have additive coherent information (see \Cref{sec:degradable-antidegradable}), but the concatenation of these two simple channels shows superadditivity of coherent information \cite{leditzky2018dephrasure,bausch2020neural,siddhu2021positivity}.
The results of \cite{leditzky2018dephrasure} were obtained using weighted repetition codes as in \eqref{eq:weighted-rep-code}, while \cite{bausch2020neural} employed codes obtained using neural networks to compute components of the state vectors.
We give a high-level overview of this construction in \Cref{app: nn-damping-dephasing}, and refer to \cite{bausch2020neural} for details.

Our optimization ansatz once again yields improved quantum communication rates for the dephrasure channel $\cN_{p,q}$ for certain values of $(p,q)$.
Similar to \cite{leditzky2018dephrasure,bausch2020neural}, we consider values $q\in \lbrace 0.1,0.2,0.3,0.4\rbrace$ of the erasure probability, and for each $q$-value we sweep intervals of the dephasing probability $p$ around the known quantum capacity thresholds.
The numerical results are depicted in \Cref{fig:dephrasure}.
We see that for $q=0.3,0.4$ our permutation-invariant codes improve over both the weighted repetition codes $\phi_k^x$ as well as the neural network codes from \cite{bausch2020neural}.
An example of such a code is the following (listed also in \Cref{tab:perm-inv-codes}), achieving a normalized coherent information of $5.4524\times 10^{-5}$ for $(p,q)=(0.08,0.4)$ and $n=9$:
\begin{align}
    \rho_{(n)} &= x_1 \psi_1^{\otimes n} + x_2 \psi_2^{\otimes n}\\
    \text{with}\quad x_1 &= 0.9465, \quad |\psi_1\rangle = \begin{pmatrix}
        0.0031 + 0.0007i\\1
    \end{pmatrix}\\
    x_2 &= 0.0535, \quad |\psi_2\rangle = \begin{pmatrix}
        -0.2335-0.0542i\\0.9709
    \end{pmatrix}.
\end{align}
While the two code states $\psi_1,\psi_2$ have a rather large overlap $|\langle\psi_1|\psi_2\rangle|=0.9701$ and are thus almost parallel, the tensor product code states are pulled further apart.
For example, for $n=9$, 
\begin{align} 
|(\langle\psi_1|)^{\otimes 9}(|\psi_2\rangle)^{\otimes 9}| = |\langle\psi_1|\psi_2\rangle|^9 =0.7608.
\end{align}
This is yet another example of a non-orthogonal (weighted) repetition code improving quantum communication rates over the usual repetition codes.

\Cref{fig:dephrasure} makes it clear that our permutation-invariant ansatz yields improved \emph{rates} in the noisy regimes of the dephrasure channel, but does not seem to increase the known \emph{thresholds} of this channel (say, as a function of $p$ for fixed $q$).
This reinforces earlier results that positive coherent information, certified by weighted repetition codes \eqref{eq:weighted-rep-code}, persists for large intervals of $p$ because of the effect of log-singularities \cite{leditzky2018dephrasure,siddhu2021positivity,siddhu2021entropic,singh2022detecting}.

\begin{figure}[H]
\centering
\includegraphics[width=\textwidth]{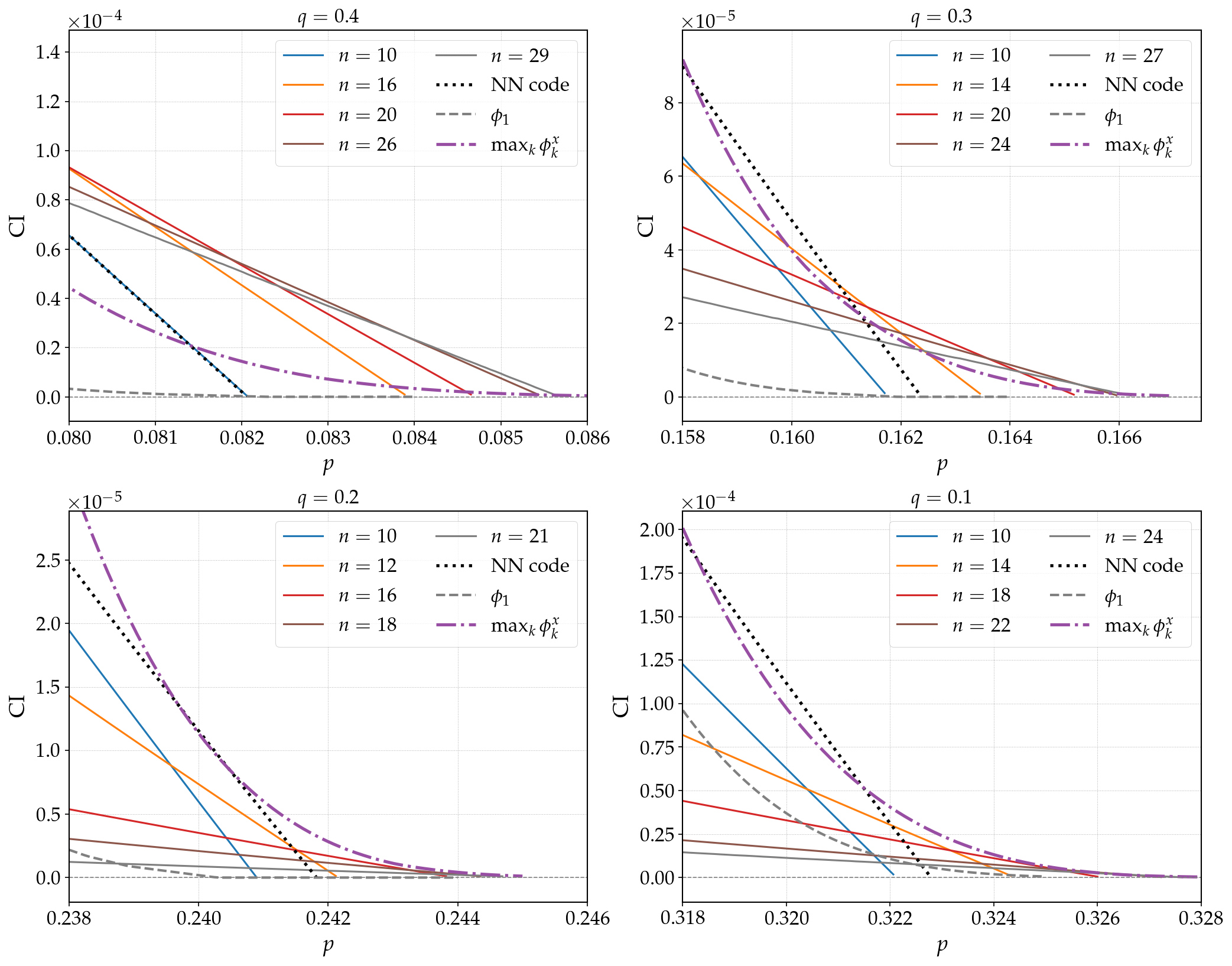}
\caption{Coherent information per channel use for the dephrasure channel $\cN_{p,q}$ as a function of the dephasing probability $p$. Each panel corresponds to a fixed erasure probability $q$. Solid lines represent the achievable rates for our numerically optimized permutation-invariant codes for up to $n=30$ channel copies. In each panel, the maximum $n$ shown is the block length that achieves the highest error threshold up to $n = 30.$ These are benchmarked against the single-letter coherent information ($\phi_1$, dashed gray line), a neural network code (NN code, dotted black), and the optimal weighted repetition code ($\max\phi_k^x$, dash-dot purple). While our codes show superadditivity of coherent information and have higher thresholds than the best neural network codes, the optimized weighted repetition codes generally provide higher rates and improved error thresholds across this parameter range. The relative performance of these benchmark codes shifts as the erasure probability $q$ changes, with the weighted repetition code showing a higher rate and threshold for lower $q$.}
\label{fig:dephrasure}
\end{figure}

\subsection{Generalized amplitude-damping channel}
\label{sec:gadc}

The generalized amplitude damping channel (GADC) $\mathcal{A}_{\gamma,N}$ is defined in terms of two parameters $\gamma, N \in [0, 1]$ and acts on a qubit state $\rho$ as
\begin{align} 
    \mathcal{A}_{\gamma,N}(\rho) = \sum_{i=1}^4 A_i \rho A_i^\dagger,
    \label{eq:gadc}
\end{align}
where
\begin{align} A_1 &= \sqrt{1 - N}(|0\rangle\langle0| + \sqrt{1 - \gamma}|1\rangle\langle1|) & A_2 &= \sqrt{\gamma(1 - N)}|0\rangle\langle1|\\ A_3 &= \sqrt{N}(\sqrt{1 - \gamma}|0\rangle\langle0| + |1\rangle\langle1|)& A_4 &= \sqrt{\gamma N}|1\rangle\langle0|. 
\end{align}
\textcite{khatri2020GADC} discussed a number of upper bounds on the quantum capacity of the GADC, and Ref.~\cite{bausch2020neural} computed lower bounds on its quantum capacity using both weighted repetition codes and a neural network state ansatz.
Both code families gave improved bounds on the quantum capacity relative to the single-letter coherent information.

We perform our optimization over permutation-invariant states for various fixed values of $N$ and intervals of $\gamma$ around the single-letter and neural-network code thresholds derived in \cite{bausch2020neural}.
The results are plotted in \Cref{fig: gadc-thresholds}, which shows increased quantum capacity thresholds.
These improvements are achieved by the following code (also listed in \Cref{tab:perm-inv-codes}):
\begin{align}\label{eq: gadc-non-orthogonal-code}
    \rho_{(n)} &= \frac{1}{2}\left(\psi_1^{\ox n} + \psi_2^{\ox n}\right)\\
    \text{with}\quad |\psi_1\rangle &= \begin{pmatrix}
        0.2960 - 0.2620i\\0.9186
    \end{pmatrix}\\
    |\psi_2\rangle &= \begin{pmatrix}
        -0.2971 + 0.2634i\\0.9178
    \end{pmatrix}.
\end{align}
This code has an interesting symmetry in the Bloch sphere with Bloch vectors for $\ket{\psi_1}$ and $\ket{\psi_2}$ being of the form $(\pm a, \pm b, c)$, as can be seen in \Cref{tab:perm-inv-codes}.
We emphasize that this \emph{single} code has better thresholds at all noise parameter pairs $(N, \gamma)$ in \Cref{fig: gadc-thresholds} than each weighted repetition code and neural network codes, which are in contrast optimized for each $(N, \gamma)$ pair individually. 
It is an interesting open problem to gain a better understanding of why non-orthogonal repetition codes yield such significantly improved thresholds over weighted repetition codes and neural-network codes. 

\begin{figure}[H]
\begin{center}
\includegraphics[width=0.8\textwidth]{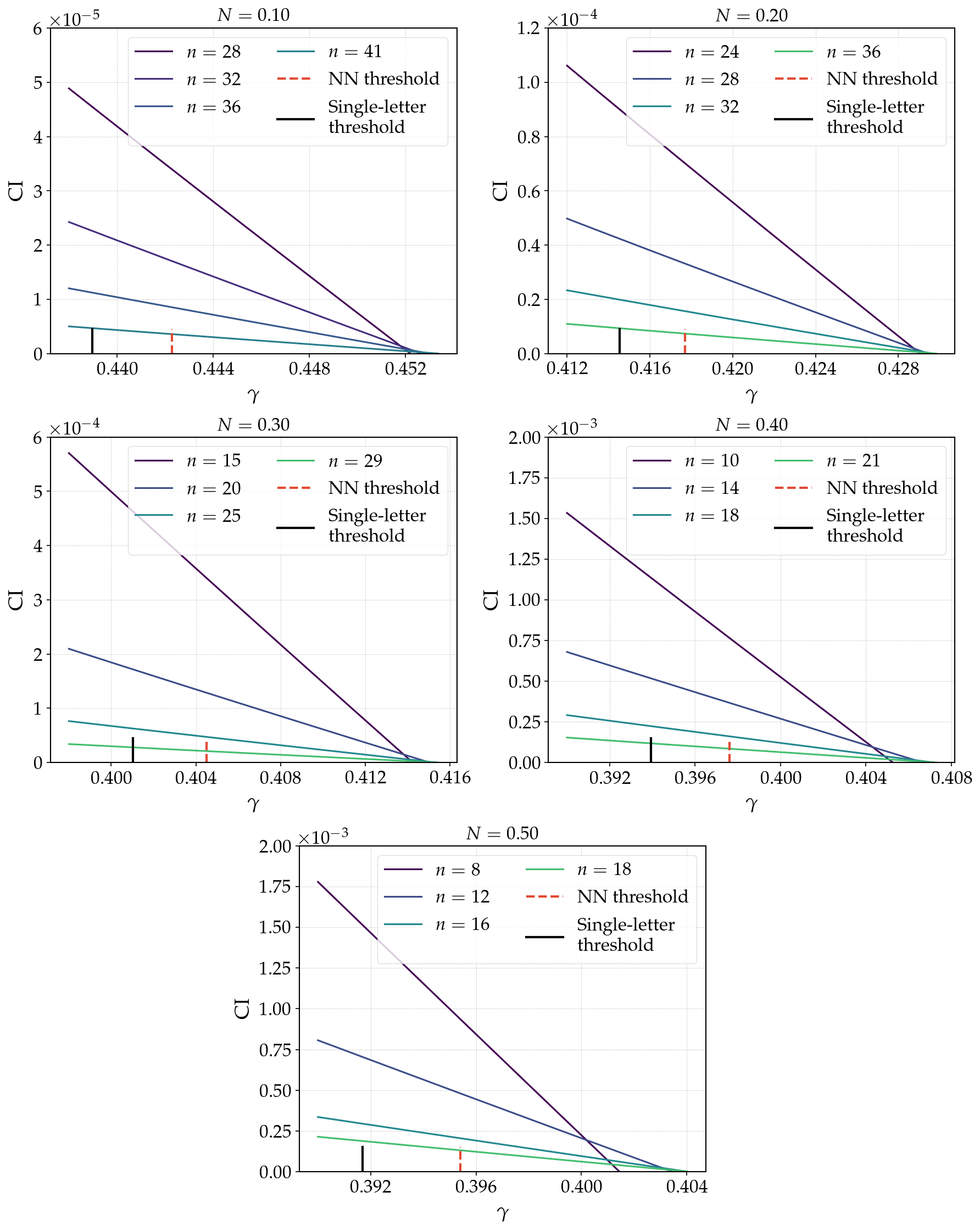}
\end{center}
\caption{Coherent information per channel use for the generalized amplitude-damping channel 
$\mathcal A_{\gamma,N}$, plotted as a function of the damping parameter $\gamma$ for fixed  $N\!=\!0.10, 0.20, 0.30, 0.40,$ and $0.50$ (one panel for each $N$).  
Solid coloured lines show the achievable rates obtained from numerically optimized 
permutation-invariant codes. In each panel, the maximum $n$ shown is the block length that achieves the highest error threshold and beyond which the threshold decreases.  
The red vertical tick on the $x$-axis marks the best neural-network code threshold from~\cite{bausch2020neural}.  
The black vertical tick on the  $x$-axis marks the single-letter threshold.  
For every $N$ displayed, the multi-copy codes achieve positive coherent information in a 
region where the single-letter rate is zero, clearly demonstrating superadditivity.  
Moreover, a single permutation-invariant code has a higher threshold for all $N$ than 
the neural-network construction, which optimizes each value of $N$ separately.}
\label{fig: gadc-thresholds}
\end{figure}

\subsection{Damping-dephasing channel}
\label{sec:damping-dephasing}

The damping-dephasing channel is another example of a concatenation of two degradable channels that results in a channel of experimental interest \cite{damping_deph_surfacecode} displaying non-additivity of coherent information at the two-letter level
\cite{siddhu2024dampingdephasing}. 
These two channels are a dephasing channel $\cD_p(\rho) = (1-p)\rho+pZ\rho Z$, and an amplitude damping channel $\cA_g = \cA_{g,0}$, where $\cA_{g,N}$ is the generalized amplitude damping channel defined in \eqref{eq:gadc} above.
The damping-dephasing channel $\mathcal{F}_{p,g}$ for $p,g\in[0,1]$ is defined as the composition of $\mathcal{D}_p$ and $\mathcal{A}_g$ (the two channel actions commute):
\begin{align}
    \mathcal{F}_{p,g} = \mathcal{D}_p \circ \mathcal{A}_g = \mathcal{A}_g \circ \mathcal{D}_p.
\end{align}
The channel can also be written in Kraus form as $\cF_{p,g} = \sum_i O_i \cdot O_i^\dagger$, where
\begin{align} O_0 &= \sqrt{1 - p} \left(|0\rangle\langle0| + \sqrt{1 - g}|1\rangle\langle1|\right) \\ O_1 &= \sqrt{g}|0\rangle\langle1| \\ O_2 &= \sqrt{p} \left(|0\rangle\langle0| - \sqrt{1 - g}|1\rangle\langle1| \right).
\end{align}
It was shown in \cite{siddhu2024dampingdephasing} that the single-letter coherent information of $\cF_{p,g}$ is maximized on input states with Bloch coordinates $(x,0,z)$, and furthermore it was observed numerically that the optimum is achieved on states with Bloch coordinates $(0,0,z)$.

We provide a formula for the coherent information of weighted repetition codes in Theorem~\ref{thm: rep-code-ci-damping-deph}. Evaluating this formula numerically shows that these codes have higher rates than the single-letter coherent information for a range of damping parameter values. A proof of Theorem~\ref{thm: rep-code-ci-damping-deph} is provided in Appendix~\ref{app:ci-damp-deph-rep}.

\begin{theorem}\label{thm: rep-code-ci-damping-deph}
 The coherent information for the weighted repetition code with the damping-dephasing channel is given by
\begin{align}
I_c(\cF_{p,g}, \phi_n^x) &= -\left(x + (1-x)g^n\right)\log\left(x + (1-x)g^n\right) \notag\\
&\quad - (1-x)(1-g)^n\log\left((1-x)(1-g)^n\right) \notag\\
&\quad + (1-x)g^n\log\left((1-x)g^n\right) \notag\\
&\quad + \eta_+ \log \eta_+ + \eta_- \log \eta_-,
\end{align}
where
\begin{align} 
\eta_{\pm} = \frac{\alpha \pm \sqrt{\alpha^2 - 4\beta}}{2}, 
\end{align}
with $\alpha = x + (1-x)(1-g)^n$ and $\beta = x(1-x)\left[(1-g)^n - ((1-2p)^2(1-g))^n\right]$.
\end{theorem}

Our optimization method over permutation-invariant states for this channel yields results that are summarized in \Cref{fig: damp-dephasing}.
As for the dephrasure channel, we see improved \emph{rates} in the noisy damping regimes of the damping-dephasing channel, but not improved thresholds beyond those guaranteed by the method of log-singularities.
The optimal states found by our optimization are listed in \Cref{tab:perm-inv-codes}.

\begin{figure}[H]
\begin{center}
\includegraphics[width=0.8\textwidth]{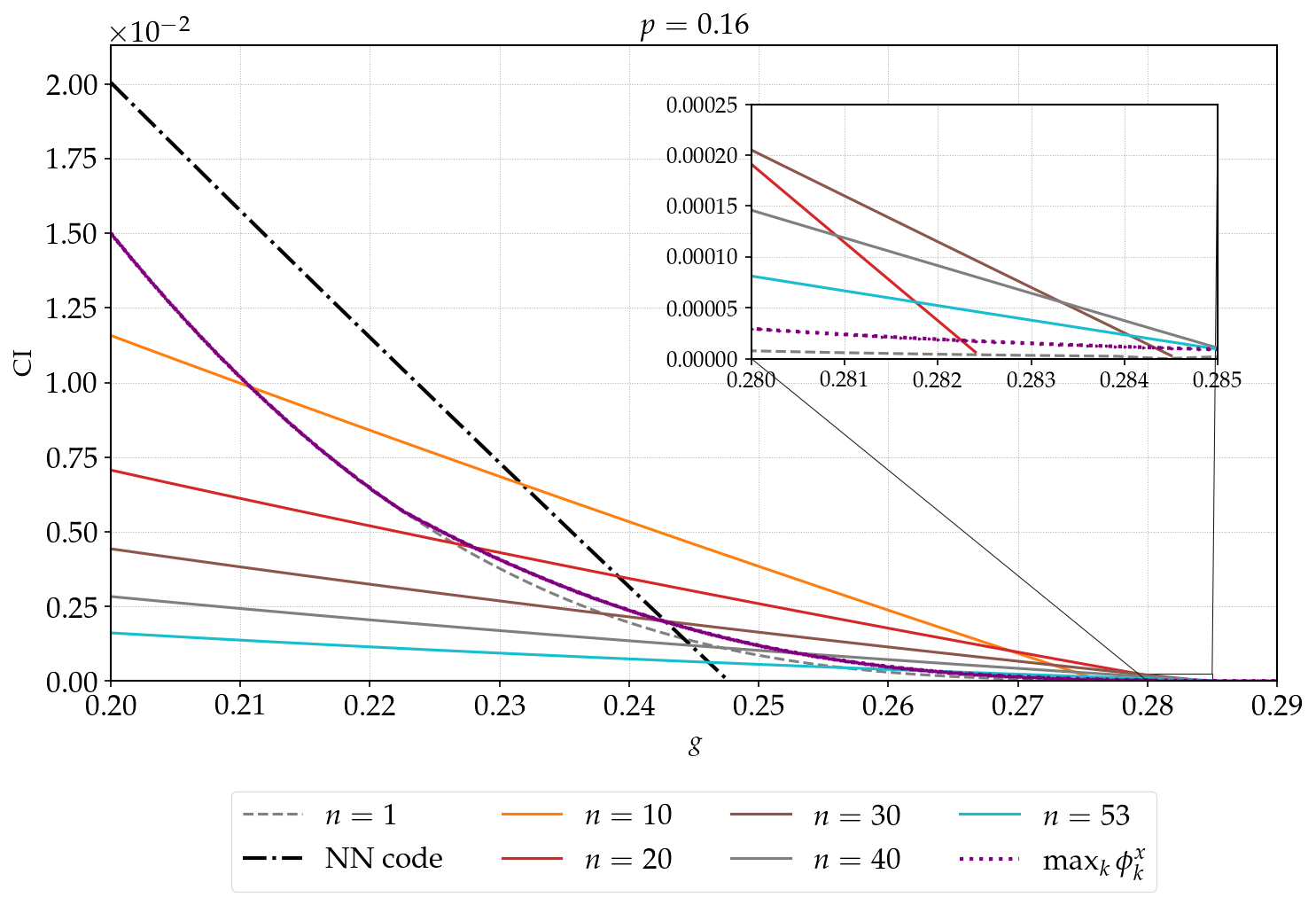}
\end{center}
\caption{Coherent information per channel use for the damping-dephasing channel 
$\mathcal F_{p,g}$ with dephasing probability fixed at $p = 0.16$ ($x$-axis: damping 
probability $g$).  
Solid coloured lines give the rates of numerically optimized permutation-invariant codes 
for $n = 10, 20, 30, 40,$ and $53$ channel copies. The threshold peaks at $n = 53$.  
The dotted grey line is the single-letter coherent information ($n=1$), while the 
black dashed line shows the best neural-network code (given in \Cref{app: nn-damping-dephasing}) up to $5$ channel copies.  
The red dotted line ($\max_k \phi_k^x$) is the analytically tractable weighted 
repetition code of Theorem~\ref{thm: rep-code-ci-damping-deph}.   
For this noise bias ($p=0.16$) the permutation-invariant code found at $g= 0.2$ (see \Cref{tab:perm-inv-codes}) for each $n$ shown has a higher threshold than the best neural-network code at $g= 0.2$ up to $n = 5$, and for $n=10$ has higher rates than the best weighted repetition code for $0.21 \lesssim g \lesssim 0.275$.}
\label{fig: damp-dephasing}
\end{figure}

\subsection{Higher rates in the mid-noise regime with higher-rank input states}\label{app: dampdeph-high-k}

So far our results focused primarily on input states that are convex combinations of two pure i.i.d.~states, i.e., states of the form $\sum_{i=1}^k x_i \ketbra{\psi}_i^{\otimes n}$ with $k=2$. This choice is motivated by the search for codes that achieve high error thresholds at the price of lower rates, which is based on the QEC intuition that one should only try to encode a small number of logical qubits in high-noise regimes. 
A data-processing argument shows that such low-rank codes receive a significant penalty in their \emph{rate}.
However, using a larger number $k$ of states resulting in a mixed input state with higher rank, we observe a different phenomenon for non-Pauli channels: for certain noise parameters and a fixed number of channel copies, employing such higher-rank input states can yield higher communication rates than both the repetition codes and neural network codes even below the thresholds of these codes.

We illustrate this point with the damping-dephasing channel (see \Cref{sec:damping-dephasing}), specifically at the parameter point $(p, g) = (0.16, 0.2)$. This is a region where the $n=1$ coherent information is positive, but one can still obtain examples of superadditive coherent information with higher rates than even the best neural network code up to $n = 5$. We performed optimizations for $k \in \{2, 3, 4, 5\}$ states for $n=5$. Codes found using our optimization for the channel $\cF_{p,g}$ at $(p, g) = (0.16, 0.2)$ and their coherent information are listed in \Cref{tab: damp-deph-higher-k-codes}, with each of the $k$ states found to have roughly equal probability. For a fixed number of states $k$, each state in \Cref{tab: damp-deph-higher-k-codes} is found to have the same $z$-coordinate on the Bloch sphere (up to numerical noise).

Figure~\ref{fig: damp-dephasing-higher-k} shows achievable rates (coherent information per channel use) as a function of the damping parameter $g$ for $n =5$ and $k = 2,3,4,5$. As the plot demonstrates, for a fixed small $n$, the rate consistently increases with $k$. This suggests that while $k=2$ non-orthogonal repetition codes are powerful for pushing error thresholds beyond those of orthogonal repetition codes and neural network codes, exploring codes with $k>2$ is a promising avenue for improving rates in the mid-noise regime for the damping-dephasing channel.

\begin{figure}[H]
\begin{center}
\includegraphics[scale = 0.28]{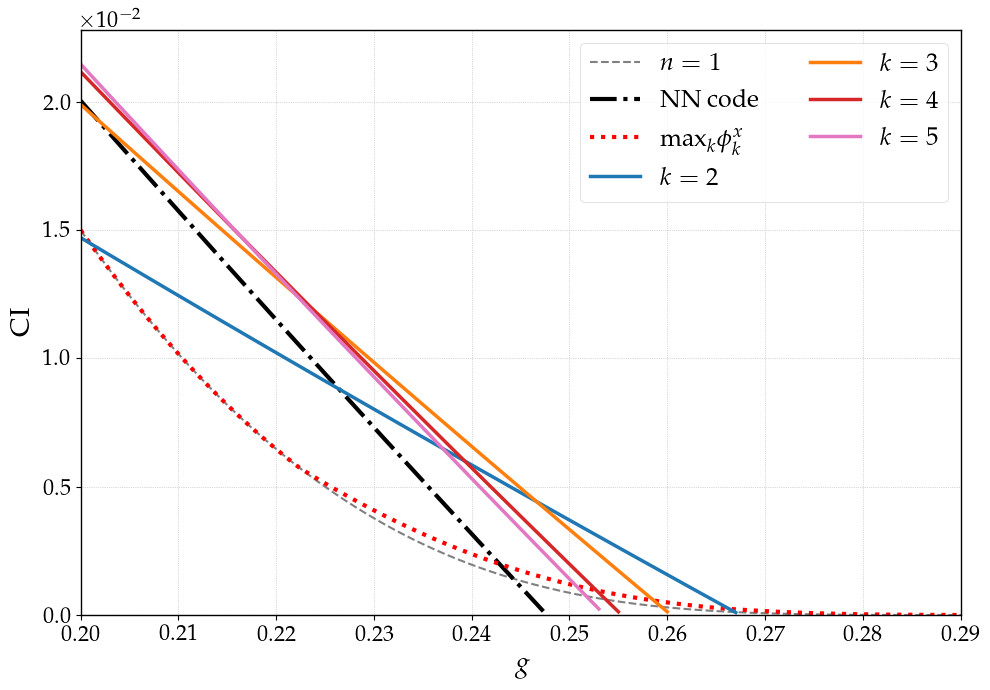}
\end{center}
\caption{Coherent information per channel use for the damping-dephasing channel 
$\mathcal F_{p,g}$ with dephasing probability fixed at $p = 0.16$ ($x$-axis: damping 
probability $g$).  
Solid coloured lines give the rates of numerically optimized $k$-state permutation-invariant codes 
for $n = 5$ channel copies and $k =2,3,4,5$.  
These are benchmarked against the single-letter coherent information (dotted gray), the best neural-network code (dashed black, code given in \Cref{tab:nn-code-dampdeph}) up to $5$ channel copies and the weighted 
repetition code $\max_k \phi_k^x$ (dotted red).   
For this noise bias ($p=0.16$) the $k$-state permutation-invariant code found at $g= 0.2$ at $n=5$ has a higher rate than the weighted repetition code and the best neural-network code at $g= 0.2$ up to $n = 5$, with the rate increasing with the number of states $k.$  The states for each $k$ are given in \Cref{tab: damp-deph-higher-k-codes}.
}
\label{fig: damp-dephasing-higher-k}
\end{figure}

\begin{table}[t]
\centering
\caption{Optimized $k$-state permutation-invariant codes for the damping-dephasing channel for $k = 2,3,4,5$. For each $k$, the table lists the optimal mixture of almost-pure states (defined by their Bloch vectors $\vec r_i$ and probabilities $p_i$) found for $n=5$ channel copies at the noise parameters $(p,g)=(0.16, 0.2)$. The last column gives the corresponding coherent information per channel use.}
\label{tab: damp-deph-higher-k-codes}
\sisetup{exponent-product = \cdot}

\begin{tabular}{
  c
  c
  c
  S[table-format=1.4e-2]
}
\toprule
$k$ & $(\vec r_1, \dots, \vec r_k)$ & $(p_1, \dots, p_k)$ & {CI per channel use} \\
\midrule

 2 &
\makecell[t]{$(-0.2053, -0.3352, 0.9195)$\\$(0.2024, 0.3404, 0.9182)$} &
\makecell[t]{$0.5$\\$0.5$} &
1.4707e-2 \\ 
\cmidrule(lr){1-4}

 3 &
\makecell[t]{$(0.4681    ,0.1773    ,0.8657)$\\$( -0.3875,    0.3147,    0.8665)$\\$(-0.0780,   -0.4946    ,0.8656)$} &
\makecell[t]{     $0.3339$\\
    $0.3337$\\
    $0.3325$} &
1.9899e-2 \\ [4ex]
\cmidrule(lr){1-4}

 4 &
\makecell[t]{$( 0.5122,   -0.2180,    0.8306)$\\$(-0.2214,   -0.5116,    0.8301)$\\$(0.2209,    0.5116,    0.8303)$\\$(-0.5105,    0.2216,    0.8307)$} &
\makecell[t]{ $0.2497$\\
    $0.2523$\\
    $0.2484$\\
    $0.2497$}&
2.1175e-2 \\ [4ex]
\cmidrule(lr){1-4}

 5 &
\makecell[t]{$(-0.5144,   -0.2607,    0.8169)$\\$(0.2582,    0.5166,    0.8164)$\\$(-0.4089,    0.4074,    0.8166)$\\$( 0.0871,   -0.5706,    0.8166)$\\$(0.5722,   -0.0871,    0.8155)$} &
\makecell[t]{  $0.1998$\\
    $0.1986$\\
    $0.1997$\\
    $0.2005$\\
    $0.2014$} &
2.1474e-2 \\ 

\bottomrule
\end{tabular}
\end{table}

\section{Conclusion}
\label{sec:conclusion}

\subsection{Discussion of our results}
\label{sec:discussion}

We have demonstrated that leveraging permutation symmetry provides a computationally tractable framework for obtaining improved quantum communication rates for  quantum channels. By restricting the search space to certain permutation-invariant codes and applying the machinery of Schur-Weyl duality, we transformed the problem of computing coherent information from one that is exponential in the number of channel copies $n$ to one that is polynomial. This enabled us to numerically optimize codes for up to $n = 100$ copies for channels with qubit output and up to $n=30$ copies for channels with qutrit output.
This method goes far beyond the reach of brute-force methods and yields improved communication rates and thresholds for a variety of important channel families.

A key finding is the effectiveness of non-orthogonal repetition codes. Prototypical examples of superadditivity for Pauli channels often rely on orthogonal repetition codes. Our results for the 2-Pauli and BB84 channels (\Cref{sec:2-pauli,sec:BB84}), however, show that better rates can be achieved with non-orthogonal repetition codes. The underlying mechanism for this improvement can be understood through an analysis of each irrep: changing the overlap between the code states alters the coherent information contribution across the different irreps of the general linear group. As shown in \Cref{fig:2pauli-CI-by_partition,fig:bb84-CI-by-partition}, for non-orthogonal codes, the positive coherent information contributions from irreps corresponding to partitions $(n,0)$ and $(n-1,1)$ can overcome the negative contributions from other irreps. This results in a net positive rate in noise regimes where an orthogonal code has zero rate and provides a representation-theoretic explanation for the effectiveness of non-orthogonal code states.

The simple structure of optimal codes found through our approach stands in contrast to other approaches that generate entangled states, such as those constructed with neural networks \cite{bausch2020neural}. For the generalized amplitude-damping channel, for instance, a single non-orthogonal repetition code given by \eqref{eq: gadc-non-orthogonal-code} was found to uniformly outperform noise-adapted neural-network codes across a wide range of noise parameters. The codes that achieve improved rates with our approach also require fewer channel copies than the previously best-known concatenated code constructions, which have large blocklengths from multiple levels of concatenation \cite{fern2008lower}.

Our work also indicates limitations of imposing full-permutation invariance and assuming the simple structure of convex sums of i.i.d.~states. For the depolarizing channel, our optimization recovers the standard repetition codes in the high-noise regime, but does not yield any non-orthogonal repetition codes with higher rates (see \Cref{sec:depolarizing}). Better codes for the high-noise regime of this channel, for example the cat codes family \cite{divincenzo1998capacity,smith2006degenerate,fern2008lower}, possess a more constrained wreath product symmetry which our current framework does not capture. Similarly, for the dephrasure and damping-dephasing channels our codes found improved communication rates in noisy regimes, but they failed to push the capacity thresholds beyond those established by weighted repetition codes. This reinforces the principle that different types of code structures are optimal for different channels and noise regimes. These limitations motivate the proposed future work outlined in \Cref{subsec: open-questions} below, namely generalizing our methods to handle arbitrary permutation-invariant states and restricted symmetry groups.

Our results provide quantitative improvements in quantum communication rates and thresholds across both Pauli and non-Pauli channels together with a deeper qualitative understanding of the structure of permutation-invariant codes based on representation theory. The success of simple non-orthogonal repetition codes opens a promising direction for the design of practical quantum communication schemes.

\subsection{Open questions}\label{subsec: open-questions}

We plan to extend our study of quantum information transmission rates and thresholds using symmetric quantum codes along the following lines:

\subsubsection{Generalization to arbitrary permutation-invariant states}
\label{sec:general-permutation-invariance}

In the present work we restricted the class of permutation-invariant codes to convex combinations of i.i.d.~states such as in \eqref{eq:convex-mixture-iid}.
A priori, this restriction seems reasonable since the resulting class of possible code states includes the weighted repetition codes \eqref{eq:weighted-rep-code} that are known to give superadditive channel coherent information for various channel models.
Indeed, our optimization found \emph{new} examples of superadditive codes in this class, such as the ones in \cref{eq:non-orthogonal-code,eq:non-orthogonal-code-bb84,eq: gadc-non-orthogonal-code}.
Mathematically, the restricted ansatz \eqref{eq:convex-mixture-iid} allowed us to use Molev's approach in \cite{molev2006gelfand} to construct the irreps of $\GL(d)$ (via those of the associated Lie algebra $\gl(d)$) for convex mixtures of i.i.d.~states.

An obvious drawback of this ansatz is that it only includes code states that are \emph{fully separable} across channel input systems.
For such codes, the superadditivity of coherent information stems from the entanglement between the channel inputs and the environment, which is most easily seen using the `purified' formula of coherent information in \eqref{eq:coherent-info-purification}.
Multipartite entanglement of permutation-invariant states is limited due to the `exclusive' nature of entanglement \cite{yang2006simple}, but the coherent information may still benefit from such weak entanglement between channel inputs.
This is for example the case for the dephrasure channel \cite{bausch2020neural} or the damping-dephasing channel \cite{siddhu2024dampingdephasing}.

A general permutation-invariant input state has the form $\rho_{(n)} = \bigoplus_{\lambda\in\Lambda(n,d)} \rho_\lambda \otimes \one_{S_\lambda}$, where $\rho_\lambda$ acts on the $\GL(d)$-irrep $V_\lambda^d$.
To derive an analogue of \Cref{thm:coherent-info-iid-mixture} to efficiently compute the coherent information of $\cN^{\otimes n}(\rho_{(n)})$, we propose to decompose the channel action of $\cN^{\otimes n}$ into a linear combination of completely positive maps $\cN_{\lambda\to\mu}\colon V_\lambda^{d}\to V_\mu^{d'}$ (with $\lambda\in\Lambda(n,d)$ and $\mu\in\Lambda(n,d')$) that map between irreps in the input and output space (see also \cite{harrow2005phd}), so that the channel action of $\cN^{\otimes n}$ can be directly computed on the individual $\rho_\lambda$.

\subsubsection{Restricted symmetry groups}
While imposing full permutation invariance on input codes yields good results for the quantum channels discussed in this paper, this ansatz is too restrictive for certain channels such as the depolarizing channel (see \Cref{sec:depolarizing}).
In the noisy regime, the coherent information of this channel can be increased significantly using cat codes with a Shor code-type symmetry \cite{divincenzo1998capacity,smith2006degenerate,fern2008lower}, corresponding to wreath products of symmetric groups.
Developing an efficient algorithm to compute the coherent information of input states with this symmetry will require a careful adaptation of the representation theory of wreath products to the quantum setting \cite{Ceccherini2014wreath}, and using the insights from the future work outlined in \Cref{sec:general-permutation-invariance}.
We also plan to extend our approach to more general subgroups of the symmetric group with a natural action on quantum codes.

\paragraph{Acknowledgments.}
We appreciate helpful discussions with Johannes Bausch, Philippe Faist, Fernando Granha Jeronimo, Marius Junge, Iman Marvian, Vikesh Siddhu, and Michael Walter. 
This work was supported by University of Illinois Campus Research Board Award No.~RB23076 and National Science Foundation Grant No.~2426103.
Furthermore, to obtain our numerical results we made use of the Illinois Campus Cluster, a computing resource that is operated by the Illinois Campus Cluster Program (ICCP) in conjunction with the National Center for Supercomputing Applications (NCSA), and which is supported by funds from the University of Illinois Urbana-Champaign.

\printbibliography[heading=bibintoc]

\appendix
\section{Optimal two-state permutation-invariant codes}
\label{app:optimal-perminv-codes}
\Cref{tab:perm-inv-codes} lists the optimal permutation-invariant codes with two code states ($k=2$) found by our optimization approach for each channel considered in this paper.

\begin{table}[t]
	\centering
	\caption{Optimized two-state permutation-invariant codes. Each code is a mixture of the pure states with Bloch vectors $\vec r_1$ and $\vec r_2$ with corresponding probabilities $(p_1; p_2)$. The parameters listed are those at which the optimization for each channel was performed. The last column lists the coherent information per channel use, evaluated at the parameters listed using $9$ copies of each channel.}
	\label{tab:perm-inv-codes}
	\sisetup{exponent-product = \cdot}
	
	\begin{tabular}{
			l
			l
			c
			c
			S[table-format=1.4e-2]
		}
		\toprule
		Channel & Parameters & $(\vec r_1; \vec r_2)$ & $(p_1; p_2)$ & {CI per channel use} \\
		\midrule
		\makecell[lt]{2-Pauli\\ (Sec.~\ref{sec:2-pauli})} &
		$p=0.2271$ &
		\makecell[t]{$(-0.0597, 0.7647, -0.6416)$\\$(0.0624, 0.7643, 0.6418)$} &
		\makecell[t]{$0.5001$\\ $0.4999$} &
		1.2475e-04 \\[4ex]
		
		\makecell[lt]{BB84\\ (Sec.~\ref{sec:BB84})} &
		$p=0.112105$ &
		\makecell[t]{$( 0.0383, -0.6934, -0.7196)$\\$(-0.0383, -0.6934, 0.7196)$} &\makecell[t]{$0.5$\\ $0.5$} &
		4.9724e-04 \\[4ex]
		
		\makecell[lt]{GADC\\ (Sec.~\ref{sec:gadc})} &
		\makecell[t]{$\gamma=0.44035$\\$N=0.1$} &
		\makecell[t]{$(0.5437, 0.4813, -0.6875)$\\$(-0.5454, -0.4835, -0.6846)$} &
		\makecell[t]{$0.4995$\\ $0.5005$} &
		8.8918e-04 \\[4ex]
		
		\makecell[lt]{Damping-\\Dephasing\\ (Sec.~\ref{sec:damping-dephasing})} &
		\makecell[t]{$p=0.16$\\ $g=0.2$}  &
		\makecell[t]{$(-0.2053, -0.3352, 0.9195)$\\$(0.2024, 0.3404, 0.9182)$} &
		\makecell[t]{$0.5$\\$ 0.5$} &
		1.2171e-2 \\[4ex]
		
		\cmidrule(lr){1-5}
		\makecell[lt]{Dephrasure\\ (Sec.~\ref{sec:dephrasure})} & \makecell[t]{$q=0.1$\\$p=0.32$} &
		\makecell[t]{$(-0.0421, 0.8554, -0.5163)$\\$(0.0000, -0.0004, -1.0000)$} &\makecell[t]{$0.1215$\\$ 0.8785$} &
		5.2223e-05 \\[4ex]

		& \makecell[t]{$q=0.2$\\$p=0.24$} &
		\makecell[t]{$(-0.0061, -0.9742, -0.2257)$\\$(0.0000, 0.0000, -1.0000)$} &
		\makecell[t]{$0.0317$\\$0.9683$}&
		1.3181e-06 \\[4ex]
		
		& \makecell[t]{$q=0.3$\\$p=0.16$} &
		\makecell[t]{$(-0.6907, -0.0557, 0.7210)$\\$(0.0003, 0.0000, 1.0000)$} &
		\makecell[t]{$0.0207$\\
			$0.9793$}&
		2.3103e-05 \\[4ex]
		
		& \makecell[t]{$q=0.4$\\$p=0.08$} &
		\makecell[t]{$(0.0063, -0.0015, -1.0000)$\\$(-0.4533, 0.1052, -0.8851)$} & \makecell[t]{$0.9465$\\ $0.0535$}
		&
		5.4524e-05 \\
		\bottomrule
	\end{tabular}
\end{table}

\section{Constructing irreps of \texorpdfstring{$\GL(2)$}{GL(2)} on the symmetric subspace}\label{app:gl2_irreps}

We derive the explicit matrix elements for the irreducible representation of $\GL(2)$ corresponding to the partition $\lambda=(m)$ consisting of a single part, i.e., the $m$-th symmetric power $\text{Sym}^m(\mathbb{C}^2)$ denoted by $V_{(m)}$. 
This space is isomorphic to the space of homogeneous polynomials of degree $m$ in two variables $x_1,x_2$.
An unnormalized basis is given by $\{\tilde{e}_k = x_1^k x_2^{m-k}\}_{k=0}^m$. With respect to the inner product $\langle x_1^l x_2^{m-l}, x_1^k x_2^{m-k} \rangle = l! (m-l)! \delta_{lk}$, the squared norm of a basis vector is $\|\tilde{e}_k\|^2 = k!(m-k)!$. This gives an orthonormal basis $\{e_k\}_{k=0}^m$ where $e_k = \frac{1}{\sqrt{k!(m-k)!}} \tilde{e}_k$. We order the basis as $\{e_m, e_{m-1}, \dots, e_0\}$.

Let $A = \begin{pmatrix} a & b \\ c & d \end{pmatrix} \in \text{GL}_2(\mathbb{C})$. Its action on the variables $x_1, x_2$ is determined by $A^T x$:
\begin{align}
    A \cdot x_1 &= ax_1 + cx_2 \\
    A \cdot x_2 &= bx_1 + dx_2
\end{align}
We compute the action of $A$ on a basis vector $e_j$:
\begin{align}
    A \cdot e_j &= \frac{1}{\sqrt{j!(m-j)!}} (ax_1 + cx_2)^j (bx_1 + dx_2)^{m-j}.
\end{align}
Using the binomial theorem for each factor,
\begin{align}
    (ax_1 + cx_2)^j &= \sum_{p=0}^j \binom{j}{p} (ax_1)^p (cx_2)^{j-p} \\
    (bx_1 + dx_2)^{m-j} &= \sum_{q=0}^{m-j} \binom{m-j}{q} (bx_1)^q (dx_2)^{m-j-q},
\end{align}
and multiplying these expansions and collecting terms for the monomial $x_1^k x_2^{m-k} = \tilde{e}_k$ gives:
\begin{align}
    A \cdot e_j = \frac{1}{\sqrt{j!(m-j)!}} \sum_{p=0}^j \sum_{q=0}^{m-j} \binom{j}{p}\binom{m-j}{q} a^p c^{j-p} b^q d^{m-j-q} x_1^{p+q} x_2^{(j-p)+(m-j-q)}.
\end{align}
The exponent of $x_1$ is $k = p+q$ and that of $x_2$ is $m-k$. Substituting $q=k-p$, we extract the coefficient of $\tilde{e}_k = x_1^k x_2^{m-k}$. The summation is over values of $p$ such that $0 \le p \le j$ and $0 \le k-p \le m-j$, which implies $\max(0, k-(m-j)) \le p \le \min(k,j)$.
Hence, 
\begin{align}
    \text{Coeff}(\tilde{e}_k) = \frac{1}{\sqrt{j!(m-j)!}} \sum_{p=\max(0, k-(m-j))}^{\min(k,j)} \binom{j}{p}\binom{m-j}{k-p} a^p c^{j-p} b^{k-p} d^{(m-j)-(k-p)}.
\end{align}
The matrix element $[S_m(A)]_{k,j}$ is the coefficient of the orthonormal basis vector $e_k$ in the expansion of $A \cdot e_j$. Since $e_k = \frac{1}{\sqrt{k!(m-k)!}} \tilde{e}_k$, we have $[S_m(A)]_{k,j} = \sqrt{k!(m-k)!} \cdot \text{Coeff}(\tilde{e}_k)$. This yields Equation~\eqref{eq:Sm_A_formula}. Note that one can view this calculation as a special case of the approach developed by Clebsch and Deruyts for constructing irreps of $\GL(d)$ for any $d$. We refer the reader to \cite[Ex.~15.57]{fulton2013representation} and \cite[Sec.~2.1.1]{mulmuley2007pvsnpgeometric} for further details of this approach.

\section{Weighted repetition codes for Pauli channels}\label{app:ci-pauli-rep}

In this appendix we provide a proof of the formula for the coherent information of repetition codes given in \Cref{thm:pauli_ci_rep}.
A Pauli channel is of the form $\cN_{\mathbf{p}}(\rho) = p_0\rho + p_1 X\rho X + p_2 Y\rho Y + p_3 Z \rho Z$.
The weighted repetition code is $\phi_n^x = x\ketbra{0}{0}^{\otimes n} + (1-x)\ketbra{1}{1}^{\otimes n}.$ A purification of the repetition code is $\ket{\psi_n} = \sqrt{x}\ket{0}^{\otimes n+1} + \sqrt{1-x}\ket{1}^{\otimes n+1}$.
A Pauli channel has output:
\begin{align}
\cN_{\mathbf{p}}(\ketbra{0}{0}) &= (p_0 + p_3) \ketbra{0}{0} + (p_1 + p_2) \ketbra{1}{1}\\
\cN_{\mathbf{p}}(\ketbra{0}{1}) &= (p_0 - p_3) \ketbra{0}{1} + (p_1 - p_2)\ketbra{1}{0}\label{eq:pauli-output2}\\
\cN_{\mathbf{p}}(\ketbra{1}{0}) &= (p_1 - p_2) \ketbra{0}{1} + (p_0 - p_3)\ketbra{1}{0}\label{eq:pauli-output3}\\
\cN_{\mathbf{p}}(\ketbra{1}{1}) &= (p_1 + p_2) \ketbra{0}{0} + (p_0 + p_3) \ketbra{1}{1}.
\end{align}
To compute the coherent information $I_c(\cN_{\mathbf{p}}, \phi_n^x) = S(\cN_{\mathbf{p}}^{\otimes n}(\phi_n^x)) - S(\id_{R} \otimes \cN_{\mathbf{p}}^{\otimes n}(\ketbra{\psi_n}{\psi_n}))$, we need eigenvalues of $\cN_{\mathbf{p}}^{\otimes n}(\phi_n^x)$ and $\id_{R} \otimes \cN_{\mathbf{p}}^{\otimes n}(\ketbra{\psi_n}{\psi_n}).$ The matrix $\cN_{\mathbf{p}}^{\otimes n}(\phi_n^x)$ is a diagonal matrix with entries \begin{align}\left[x(p_0 + p_3)^{n-w}(p_1 + p_2)^w + (1-x) (p_0 +p_3)^{w}(p_1 + p_2)^{n-w}\right] \qquad \text{ for } w = 0, \cdots, n,\end{align} and the entry corresponding to each $w$ appears with multiplicity $\binom{n}{w}$. 
For the second term of the coherent information, $\id_{R} \otimes \cN_{\mathbf{p}}^{\otimes n}(\ketbra{\psi_n}{\psi_n})$ is equal to
\begin{align}
    &(\id_{R} \otimes \cN_{\mathbf{p}}^{\otimes n})(x\ketbra{0}{0}^{\otimes n+1} + \sqrt{x(1-x)}\ketbra{0}{1}^{\otimes n+1}+ \sqrt{x(1-x)}\ketbra{1}{0}^{\otimes n+1}+ (1-x)\ketbra{1}{1}^{\otimes n+1}) \notag\\
    &\qquad =\left[x\ketbra{0}{0} \otimes \cN_{\mathbf{p}}(\ketbra{0}{0})^{\otimes n} + \sqrt{x(1-x)}\ketbra{0}{1} \otimes \cN_{\mathbf{p}}(\ketbra{0}{1})^{\otimes n} \right. + \notag\\
    &\qquad \quad \left.\sqrt{x(1-x)} \ketbra{1}{0} \otimes \cN_{\mathbf{p}}(\ketbra{1}{0})^{\otimes n}+ (1-x)\ketbra{1}{1} \otimes \cN_{\mathbf{p}}(\ketbra{1}{1})^{\otimes n}\right].
\end{align}
This is a matrix of the form
\begin{align}
    M = \begin{pmatrix}
x \cN_{\mathbf{p}}(\ketbra{0}{0})^{\otimes n} & \sqrt{x(1-x)}\cN_{\mathbf{p}}(\ketbra{0}{1})^{\otimes n} \\
\sqrt{x(1-x)}\cN_{\mathbf{p}}(\ketbra{1}{0})^{\otimes n} & (1-x)\cN_{\mathbf{p}}(\ketbra{1}{1})^{\otimes n}
\end{pmatrix} 
\end{align}
with non-zero entries on the diagonal and the anti-diagonal. 
For convenience, we define the following $2^n\times 2^n$ matrices:
\begin{align}
\mathrm{diag}(a_1,a_2,\dots,a_{2^n}) \coloneqq 
\begin{pmatrix}
a_1 & \cdot & \cdots & \cdot \\
\cdot & a_2 & \cdots & \cdot \\
\vdots & \vdots & \ddots & \vdots \\
\cdot & \cdot & \cdots & a_{2^n}
\end{pmatrix},
\end{align}
and
\begin{align}
\mathrm{adiag}(b_1,b_2,\dots,b_{2^n}) \coloneqq
\begin{pmatrix}
\cdot & \cdots & \cdot & b_1\\[1ex]
\cdot & \cdots & b_2 & \cdot\\[1ex]
\vdots & \ddots & \vdots & \vdots\\[1ex]
b_{2^n} & \cdots & \cdot & \cdot
\end{pmatrix}.
\end{align}
Then the matrix $M$ can be written in block form as
\begin{align}
M = \begin{pmatrix}
x\,\mathrm{diag}(a_1,\dots,a_{2^n}) & \sqrt{x(1-x)}\,\mathrm{adiag}(b_1,\dots,b_{2^n}) \\[1ex]
\sqrt{x(1-x)}\,\mathrm{adiag}(c_1,\dots,c_{2^n}) & (1-x)\,\mathrm{diag}(d_1,\dots,d_{2^n})
\end{pmatrix},
\end{align}
with
\begin{align}
a_i &= (p_0+p_3)^{\,n-k_i}(p_1+p_2)^{\,k_i} &
b_i &= (p_0-p_3)^{\,n-k_i}(p_1-p_2)^{\,k_i}\\
c_i &= (p_1-p_2)^{\,n-k_i}(p_0-p_3)^{\,k_i} &
d_i &= (p_1+p_2)^{\,n-k_i}(p_0+p_3)^{\,k_i},
\end{align}
where the integer $k_i$ represents the Hamming weight of the $i$th computational basis element. The Hamming weights can be $0,\cdots, n$, and the terms $a_i, b_i, c_i, d_i$ with Hamming weight $k_i$ occur with multiplicity $\binom{n}{k_i}.$

\subsubsection*{Eigenvalues of M}

 Suppose $N$ is even and that $A$ is an $N\times N$ matrix that is zero except on the diagonal and on the anti-diagonal. That is, assume for $i=1,\dots,N$ that
\begin{align}
A_{ii} &= a_i & 
A_{i, N+1-i} &= b_i.
\end{align}
Define the permutation 
\begin{align}
J\,e_i = e_{N+1-i}.
\end{align}
After a reordering of the basis so that indices $i$ and $N+1-i$ are paired (for $i=1,\dots, N/2$, assuming $N$ is even), the matrix $A$ breaks up into blocks acting on invariant subspaces of dimension $2$. In the subspace spanned by $\{e_i, \, e_{N+1-i}\}$ the restriction of $A$ is
\begin{align}
A^{(i)} = \begin{pmatrix}
a_i & b_i\\[1mm]
b_{N+1-i} & a_{N+1-i}
\end{pmatrix}.
\end{align}
The eigenvalues of this $2\times 2$ block are
\begin{align}
x = \frac{a_i+a_{N+1-i} \pm \sqrt{(a_i-a_{N+1-i})^2+4\,b_i\,b_{N+1-i}}}{2}.
\end{align}

In our situation the matrices $\mathrm{diag}(\cdot)$ and $\mathrm{adiag}(\cdot)$ each have size $2^n\times 2^n$. Since $2^n$ is even, we can pair the indices. Label the indices in the range $1,\dots,2^n$, and define the ``complementary index" that pairs with $i$ to be $\bar{i}=2^n+1-i.$
Then in the matrix $M$ (which is $2^{n+1}\times 2^{n+1}$) the only nonzero entries occur in positions corresponding to pairs $\{\,i,\bar{i}\}$ (from the top-left and bottom-right blocks),
together with their ``cross terms'' coming from the anti-diagonal blocks.

In particular, note that in our block matrix $M$, the top-left block gives the entry
$
M_{i,i} = x\,a_i,
$
and the bottom-right block gives
$
M_{2^n+\bar{i},\,2^n+\bar{i}} = (1-x)\,d_{\bar{i}}.
$
Similarly, the top-right block has
$
M_{i,\,2^n+\bar{i}} = \sqrt{x(1-x)}\,b_i,
$
and the bottom-left block has
$
M_{2^n+\bar{i},\, i} = \sqrt{x(1-x)}\,c_{\bar{i}}.
$ Thus, for each pair $i$ and $\bar{i}$ the $2\times 2$ invariant block is
\begin{align}
M^{(i)} =
\begin{pmatrix}
x\,a_i & \sqrt{x(1-x)}\,b_i \\[1mm]
\sqrt{x(1-x)}\,c_{\bar{i}} & (1-x)\,d_{\bar{i}}
\end{pmatrix}.
\end{align}
A further simplification occurs by observing that because of \Cref{eq:pauli-output2,eq:pauli-output3}, we in fact have $c_{\bar{i}} = b_i$ and $a_i = d_{\bar{i}}$, and so 
\begin{align}
M^{(i)} =
\begin{pmatrix}
x\,a_i & \sqrt{x(1-x)}\,b_i \\[1mm]
\sqrt{x(1-x)}\,b_i & (1-x)\,a_i
\end{pmatrix}.
\end{align}
The eigenvalues of $M^{(i)}$ are obtained by solving
\begin{align}
\Bigl(\eta - x\,a_i\Bigr)\Bigl(\eta - (1-x)\,a_i\Bigr) - x(1-x)\,b_i^2 = 0,
\end{align}
so that
\begin{align}
    \eta_{i,\pm} &= \frac{x\,a_i+(1-x)\,d_{\bar{i}}\pm\sqrt{\Bigl[x\,a_i+(1-x)\,d_{\bar{i}}\Bigr]^2-4x(1-x)(a_i d_{\bar{i}} - \,b_i\,c_{\bar{i}})}}{2}\\
    &=\frac{a_i\pm\sqrt{[(2x -1)a_i]^2+4x(1-x)b_i^2}}{2}.
\end{align}
Substituting the expressions for $a_i, b_i$ and simplifying, we obtain the eigenvalues for the invariant block corresponding to the pair $(i,\bar{i})$ to be
\begin{align}
    \eta_{i,\pm} &= \frac{(p_0+p_3)^{\,n-k_i}(p_1+p_2)^{\,k_i}\pm \sqrt{\Delta_i}}{2}\\
\text{with} \quad\Delta_i &= [(2x -1)(p_0+p_3)^{n-k_i}(p_1 + p_2)^{k_i}]^2+4x(1-x)[(p_0-p_3)^{n-k_i}(p_1 - p_2)^{k_i}]^2.
\end{align}
The eigenvalue pairs $\eta_{i,1}, \eta_{i,2}$ appear with multiplicity $\binom{n}{k_i}.$ The Hamming weights $k_i$ can be one of $0,\ldots, n$, and so the final expression for the coherent information is given by
\begin{align}
    -\sum_{w = 0}^n \binom{n}{w} y_1(w)\log y_1(w)+ \sum_{w = 0}^n \binom{n}{w}  \eta_+(w) \log \eta_+(w) + \sum_{w = 0}^n \binom{n}{w}  \eta_-(w) \log \eta_-(w)
\end{align}
where
\begin{align} 
y_1(w) &= x(p_0 + p_3)^{n-w}(p_1 + p_2)^w + (1-x) (p_0 +p_3)^{w}(p_1 + p_2)^{n-w}\\
     \eta_{\pm}(w) &= \frac{(p_0+p_3)^{n-w}(p_1+p_2)^{w}\pm \sqrt{\Delta(w)}}{2}\\
    \Delta &= [(2x -1)(p_0+p_3)^{n-w}(p_1 + p_2)^{w}]^2+4x(1-x)[(p_0-p_3)^{n-w}(p_1 - p_2)^{w}]^2.
\end{align}

\section{Weighted repetition codes for the damping-dephasing channel}\label{app:ci-damp-deph-rep}

In this appendix, we provide a proof of \Cref{thm: rep-code-ci-damping-deph}. The damping-dephasing channel $\mathcal{F}_{p,g}$ has Kraus operators
\begin{align} K_1 &= \sqrt{1 - p}(|0\rangle\langle0| + \sqrt{1 - g}|1\rangle\langle1|) \\ K_2 &= \sqrt{g}|0\rangle\langle1| \\ K_3 &= \sqrt{p}|0\rangle\langle0| - \sqrt{1 - g}|1\rangle\langle1|.
\end{align}
The action of the channel on a state $\rho$ is given by $\mathcal{F}_{p,g}(\rho) = \sum_{i=1}^3 K_i \rho K_i^\dagger$.  As in Appendix~\ref{app:ci-pauli-rep}, let $\phi_n^x = x\ketbra{0}{0}^{\otimes n} + (1-x)\ketbra{1}{1}^{\otimes n}$ be the repetition code with a purification given by $\ket{\psi_n} = \sqrt{x}\ket{0}^{\otimes n+1} + \sqrt{1-x}\ket{1}^{\otimes n+1}$.
We have the following output of $\cE$ on the basis matrices $\ketbra{i}{j}$ for $i,j \in \{0,1\}$:
\begin{align}
 \mathcal{F}_{p,g}(\ketbra{0}{0}) &= \ketbra{0}{0} \\
     \mathcal{F}_{p,g}(\ketbra{1}{1})  &= g\ketbra{0}{0} + (1-g)\ketbra{1}{1}\\
     \mathcal{F}_{p,g}(\ketbra{0}{1})  &= (1-2p)\sqrt{1-g}\ketbra{0}{1}\\
    \mathcal{F}_{p,g}(\ketbra{1}{0})  &= (1-2p)\sqrt{1-g}\ketbra{1}{0}.   
\end{align}
The output state when the channel acts on the repetition code is $\mathcal{F}_{p,g}^{\otimes n}(\phi_n^x) = x (\ketbra{0}{0})^{\otimes n} + (1-x)(g\ketbra{0}{0} + (1-g)\ketbra{1}{1})^{\otimes n}$. This matrix is diagonal in the computational basis. The eigenvalues $\{p_w\}$ depend on the Hamming weight $w$ of the basis state:
\begin{itemize}
    \item $p_0 = x + (1-x)g^n$, with multiplicity 1.
    \item $p_w = (1-x)g^{n-w}(1-g)^w$, with multiplicity $\binom{n}{w}$ for $w \in \{1, \dots, n\}$.
\end{itemize}
The density matrix corresponding to the purified state on the joint system and reference is
\begin{align} \ketbra{\psi_n}{\psi_n} = x \ketbra{00}{00}^{\otimes n+1} + (1-x)\ketbra{11}{11}^{\otimes n+1} + \sqrt{x(1-x)}\left(\ketbra{00}{11}^{\otimes n+1} + \ketbra{11}{00}^{\otimes n+1}\right). \end{align}
Applying $\id_{R} \otimes \mathcal{F}_{p,g}^{\otimes n}$ to each term, we obtain:
\begin{align}
     \id_{R} \otimes \mathcal{F}_{p,g}^{\otimes n}(\ketbra{\psi_n}{\psi_n}) &= x \ketbra{0}{0}_R \otimes \mathcal{F}_{p,g}^{\otimes n}(\ketbra{0}{0}^{\otimes n}) \notag\\
    &\quad {}+ (1-x) \ketbra{1}{1}_R \otimes \mathcal{F}_{p,g}^{\otimes n}(\ketbra{1}{1}^{\otimes n}) \notag\\
    &\quad {}+ \sqrt{x(1-x)} \ketbra{0}{1}_R \otimes \mathcal{F}_{p,g}^{\otimes n}(\ketbra{0}{1}^{\otimes n}) \notag\\
    &\quad{}+ \sqrt{x(1-x)} \ketbra{1}{0}_R \otimes \mathcal{F}_{p,g}^{\otimes n}(\ketbra{1}{0}^{\otimes n}).
\end{align}
Substituting the channel action on the basis states $\ketbra{i}{j}$, we see that the state $\rho_{RE}$ is almost diagonal. The only off-diagonal terms correspond to the basis states $\ket{0}_R\ket{0\dots0}$ and $\ket{1}_R\ket{1\dots1}$. This means the matrix is block-diagonal, with one $2 \times 2$ block and the rest being $1 \times 1$ blocks (i.e., diagonal entries).

The $2 \times 2$ block corresponds to the subspace spanned by $\{\ket{0}_R\ket{0\dots0}, \ket{1}_R\ket{1\dots1}\}$. The matrix in this basis is
\begin{align} M = \begin{pmatrix}
    x & \sqrt{x(1-x)}\left((1-2p)^2(1-g)\right)^{n/2} \\
    \sqrt{x(1-x)}\left((1-2p)^2(1-g)\right)^{n/2} & (1-x)(1-g)^n
\end{pmatrix}. \end{align}
Let $T = \tr(M) = x + (1-x)(1-g)^n$ and let $C = \left((1-2p)^2(1-g)\right)^{n/2}$. The determinant is $D = \det(M) = x(1-x)(1-g)^n - x(1-x)C^2$. The eigenvalues of this block, which we denote $\eta_{\pm}$, are given by:
\begin{align} \eta_{\pm} = \frac{T \pm \sqrt{T^2-4D}}{2} = \frac{x+(1-x)(1-g)^n \pm \sqrt{(x-(1-x)(1-g)^n)^2 + 4x(1-x)C^2}}{2}. \end{align}
The other eigenvalues are the diagonal entries corresponding to the remaining basis states. These come from the term $(1-x) \ketbra{1}{1}_R \otimes \mathcal{F}_{p,g}^{\otimes n}(\ketbra{1}{1}^{\otimes n})$. The basis vectors are $\ket{1}_R\ket{y}$ for computational basis vectors $y$. The corresponding diagonal entry for a vector $y$ of Hamming weight $w$ is $(1-x)g^{n-w}(1-g)^w$.
Thus, the remaining eigenvalues are:
\begin{align} q_w = (1-x)g^{n-w}(1-g)^w \quad \text{for } w \in \{0, 1, \dots, n-1\}, \end{align}
with multiplicity $\binom{n}{w}$.
The entropy of this state is therefore
\begin{align}  -(\eta_+ \log \eta_+ + \eta_- \log \eta_-) - \sum_{w=0}^{n-1} \binom{n}{w} q_w \log q_w. \end{align}
The coherent information is then 
\begin{align}
    I_c(\mathcal{F}_{p,g}^{\otimes n}, \phi_n^x) &= \left( -p_0\log p_0 - \sum_{w=1}^n \binom{n}{w}p_w\log p_w \right) \notag\\
    &\quad - \left( -(\eta_+\log\eta_+ + \eta_-\log\eta_-) - \sum_{w=0}^{n-1} \binom{n}{w}q_w\log q_w \right).
\end{align}
We note that the eigenvalues $p_w$ and $q_w$ are identical for $w \in \{1, \dots, n-1\}$. The sums can be split:
\begin{align}
    \sum_{w=1}^n \binom{n}{w}p_w\log p_w &= \sum_{w=1}^{n-1} \binom{n}{w}p_w\log p_w + \binom{n}{n}p_n\log p_n, \\
    \sum_{w=0}^{n-1} \binom{n}{w}q_w\log q_w &= \binom{n}{0}q_0\log q_0 + \sum_{w=1}^{n-1} \binom{n}{w}q_w\log q_w.
\end{align}
The sums from $w=1$ to $n-1$ cancel, leaving:
\begin{align}
     -p_0\log p_0 - p_n\log p_n + \eta_+\log\eta_+ + \eta_-\log\eta_- + q_0\log q_0.
\end{align}
Finally, we substitute the following expressions for the eigenvalues:
\begin{itemize}
    \item $p_0 = x + (1-x)g^n$
    \item $p_n = (1-x)(1-g)^n$ (with multiplicity $\binom{n}{n}=1$)
    \item $q_0 = (1-x)g^n$ (with multiplicity $\binom{n}{0}=1$)
    \item $\eta_{\pm}$ are the eigenvalues of the matrix $M$.
\end{itemize}

We have thus proved that the coherent information for the weighted repetition code with the damping-dephasing channel is given by
\begin{align}
I_c(\mathcal{F}_{p,g}^{\otimes n}, \phi_n^x) &= -\left(x + (1-x)g^n\right)\log\left(x + (1-x)g^n\right) \notag\\
&\quad - (1-x)(1-g)^n\log\left((1-x)(1-g)^n\right) \notag\\
&\quad + (1-x)g^n\log\left((1-x)g^n\right) \notag\\
&\quad + \eta_+ \log \eta_+ + \eta_- \log \eta_-,
\end{align}
where
\begin{align} \eta_{\pm} = \frac{\tr(M) \pm \sqrt{(\tr(M))^2 - 4\det(M)}}{2}, \end{align}
with $\tr(M) = x + (1-x)(1-g)^n$ and $\det(M) = x(1-x)\left[(1-g)^n - ((1-2p)^2(1-g))^n\right]$.

\section{Numerical optimization}\label{app: numerics}

\subsection{Particle swarm optimization}

Particle swarm optimization (PSO) is a stochastic global optimization technique. It is a population-based metaheuristic, where a collection of candidate solutions, termed particles, are moved around in the search space. The movement of each particle is influenced by its own best-known position and the best-known position found by any particle in the entire swarm. This collaborative search aims to guide the swarm towards the global optimum of the objective function. In our context, the objective function is the coherent information of a mixture of i.i.d.~states. The parameters defining these states and probabilities constitute the search space for the PSO algorithm. We refer to MATLAB  documentation \cite{mathworks_particleswarm} for details of this algorithm and to our code \cite{perm-inv-codes-github} for choices of parameters in the algorithm used to obtain results in \Cref{sec:results}.
In our optimization we made use of a fast MATLAB implementation of the matrix exponential \cite{hogben2011spinach,kuprov2011diagonalization,fastexpm-github}. 

\subsection{Parametrizing  qubit input states}
The optimization procedure described in \Cref{sec:optimization} requires a method to map a vector of real parameters $x \in \mathbb{R}^N$ to an ensemble of quantum states $\{\rho_j\}_{j=1}^k$ with corresponding probabilities $\{p_j\}_{j=1}^k$, such that $\sum_{j=1}^k p_j = 1$, $p_j \ge 0$, and each $\rho_j$ is a valid quantum state. We list several parametrization schemes.

\subsubsection{Using Bloch vectors}
This method generates $k$ qubit states, which can be mixed, along with their probabilities. The input vector $x$ has $k$ parameters for probabilities and $3k$ parameters for the states.
For each qubit state $\rho_j$, $j \in \{1, \dots, k\}$, three real parameters $(s_{j1}, s_{j2}, s_{j3})$ are extracted from $x$. Let $\mathbf{s}_j = (s_{j1}, s_{j2}, s_{j3}) \in \mathbb{R}^3$. This vector is mapped to a Bloch vector $\mathbf{r}_j = (r_{jx}, r_{jy}, r_{jz}) \in \mathbb{R}^3$ using the transformation
\begin{align}
\mathbf{r}_j = \frac{\mathbf{s}_j}{1 + \|\mathbf{s}_j\|},
\end{align}
where $\|\mathbf{s}_j\| = \sqrt{s_{j1}^2 + s_{j2}^2 + s_{j3}^2}$. This mapping ensures that $\|\mathbf{r}_j\| \le 1$, so that $\mathbf{r}_j$ lies within or on the surface of the Bloch sphere. The density matrix for the $j$-th qubit is then given by its Bloch representation
\begin{align}
\rho_j = \frac{1}{2}(I + r_{jx}X + r_{jy}Y + r_{jz}Z),
\end{align}
where $I$ is the $2 \times 2$ identity matrix and $X,Y,Z$ are Pauli matrices.

\subsection{Parametrizing qudit input states}

\subsubsection{A direct parametrization with $M^\dagger M$}
 The input vector $x$ consists of $k \cdot 2d^2$ real parameters for the states and $k$ real parameters for the probabilities. For each state $\rho_j$, $j \in \{1, \dots, k\}$, a block of $2d^2$ real parameters, $x_j'$, is extracted from $x$. These $2d^2$ real parameters are interpreted as $d^2$ complex numbers, forming a vector $v_j \in \mathbb{C}^{d^2}$. This vector $v_j$ is then reshaped into a $d \times d$ complex matrix $M_j$. The density matrix $\rho_j$ is constructed as
\begin{align}
\rho_j = \frac{M_j^\dagger M_j}{\tr(M_j^\dagger M_j)}.
\end{align}

\subsubsection{Using a measurement}
 The input vector $x$ contains $2kd^2$ real parameters. These parameters are first paired to form a single complex vector $\psi_{\text{raw}} \in \mathbb{C}^{kd^2}$. This vector is then normalized:
$\ket{\Psi} = \frac{\psi_{\text{raw}}}{\|\psi_{\text{raw}}\|_2}.$
The vector $\ket{\Psi}$ is viewed as an element of $\mathcal{H} = \mathcal{H}_S \otimes \mathcal{H}_A$, where $\mathcal{H}_S \cong \mathbb{C}^k$ is an ancillary system indexing the states and $\mathcal{H}_A \cong \mathbb{C}^{d^2}$. Let $\{\ket{j}_S\}_{j=1}^k$ be the standard orthonormal basis for $\mathcal{H}_S$.
For each $j \in \{1, \dots, k\}$, an unnormalized vector $|\phi'_j\rangle_A \in \mathcal{H}_A$ is obtained by projecting $\ket{\Psi}$ onto the subspace of $\mathcal{H}_S$ spanned by $\ket{j}$, i.e., 
$|\phi'_j\rangle_A = (\bra{j}_S \otimes I_A) \ket{\Psi},$
where $I_A$ is the identity operator on $\mathcal{H}_A$. The probability $p_j$ associated with the $j$-th state is then given by $p_j = \langle\phi'_j|\phi'_j\rangle_A$.
Due to the initial normalization of $\ket{\Psi}$, it follows that $\sum_{j=1}^k p_j = \|\ket{\Psi}\|^2 = 1$. The normalized state vector $\ket{\phi_j}_A \in \mathcal{H}_A$ is obtained as $\ket{\phi_j}_A = |\phi'_j\rangle_A / \sqrt{p_j}$ (if $p_j > 0$).
This vector $\ket{\phi_j}_A \in \mathbb{C}^{d^2}$ is interpreted as a pure bipartite state vector $\ket{\phi_j}_{A_1A_2}$ in $\mathcal{H}_{A_1} \otimes \mathcal{H}_{A_2}$, where $\mathcal{H}_{A_1} \cong \mathbb{C}^d$ and $\mathcal{H}_{A_2} \cong \mathbb{C}^d$. The $j$-th density matrix $\rho_j$ is then obtained by tracing out the second system: $\rho_j = \tr_{A_2}(\ket{\phi_j}\bra{\phi_j}).$
The resulting $\rho_j$ is a $d \times d$ density matrix acting on $\mathcal{H}_{A_1}$.

\subsection{Neural network code for the damping-dephasing channel}\label{app: nn-damping-dephasing}
We use a code found using the neural network state ansatz introduced in \cite{bausch2020neural} as a benchmark for the permutation-invariant codes found using our approach for the damping-dephasing channel. This method represents the quantum code as a pure state $|\nu_k\rangle \in \cH_R \otimes \cH_A^{\otimes k}$ whose amplitudes are parametrized by a feed-forward neural network. The network takes as input a binary string $s_n$ (where $n = k + r$ is the sum of the number of channel qubits $k$ and reference qubits $r$) representing a computational basis state of the joint system and outputs its complex amplitude $\psi(s_n)$. The network's parameters are then optimized to maximize the coherent information, which is calculated using the expression in \Cref{eq:coherent-info-purification}. 

Table~\ref{tab:nn-code-dampdeph} lists the non-zero amplitudes for a code found for the damping-dephasing channel with this method for $k = 3$ channel copies. This code has the best threshold and rate among neural network codes found for up to $5$ copies of the damping-dephasing channel. The column `$s_n(A^k|R)$' specifies the basis states of the system ($A^k$) and reference ($R$) parts that have non-zero support, and `$\psi(s_n)$' gives their corresponding complex amplitudes. The threshold (for $p = 0.16$, as a function of the damping error probability $g$) for this code is plotted in \Cref{fig: damp-dephasing}.

\begin{table}[t]
\centering
\caption{Neural network code for the damping-dephasing channel $\cF_{p,g}$ with $(p,g) = (0.16, 0.2)$ using $3$ channel copies. The column `$s_n(A^k|R)$' specifies the basis states of the system ($A^k$) and reference ($R$) parts that have non-zero support, and `$\psi(s_n)$' gives their corresponding complex amplitudes. The threshold (for $p = 0.16$, as a function of $g$) for this code is plotted in \Cref{fig: damp-dephasing}.\label{tab:nn-code-dampdeph}}
\begin{tabular}{l c r c}
\toprule
$|\nu_k\rangle$ & $s_n (A^k|R)$ & \multicolumn{1}{c}{$\psi(s^n)$} & $\frac{1}{k}Q^{(1)}( \cF_{p,g}^{\otimes k},\nu_k)$ \\
\midrule
\multirow{10}{*}{$k=3$} & $000|000$ & $+0.2351 - 0.1769i$ & \multirow{10}{*}{$2.0046\cdot10^{-2}$}\\
                       & $001|000$ & $+0.2351 - 0.1769i$ & \\
                       & $010|000$ & $-0.0306 - 0.4689i$ & \\
                       & $011|001$ & $+0.2351 - 0.1769i$ & \\
                       & $011|010$ & $+0.2351 - 0.1769i$ & \\
                       & $011|100$ & $+0.2351 - 0.1769i$ & \\
                       & $100|000$ & $+0.2351 - 0.1769i$ & \\
                       & $101|000$ & $+0.2351 - 0.1769i$ & \\
                       & $110|000$ & $+0.2351 - 0.1769i$ & \\
                       & $111|000$ & $+0.2351 - 0.1769i$ &  \\
\bottomrule
\end{tabular}
\end{table}

\end{document}